\RequirePackage[l2tabu,orthodox]{nag}
\documentclass
[11pt,letterpaper]
{article} 

\usepackage[notes=true,later=false,camera=false]{dtrt}
\usepackage[utf8]{inputenc}
\usepackage{ stmaryrd }
\usepackage{xspace,enumerate}
\usepackage[T1]{fontenc}
\usepackage[full]{textcomp}
\usepackage[american]{babel}
\usepackage{mathtools}

\renewcommand{\hat}[1]{\widehat{#1}}
\usepackage{amsthm}
\usepackage{thmtools}
\usepackage[capitalise,nameinlink]{cleveref}

\usepackage{empheq}

\hypersetup{
colorlinks=true,
urlcolor=Cerulean,
linkcolor=NavyBlue,
citecolor=PineGreen,
linktocpage=true,
}
\renewcommand*\backref[1]{\ifx#1\relax \else (pg. #1) \fi}

\renewcommand{\tilde}{\widetilde}

\usepackage{paralist}
\usepackage{turnstile}
\usepackage[framemethod=TikZ]{mdframed}
\mdfsetup{frametitlealignment=\center}
\usepackage{tikz}
\usepackage{caption}
\DeclareCaptionType{Algorithm}
\usepackage{newfloat}

\newtheorem{theorem}{Theorem}[section]
\newtheorem{lemma}[theorem]{Lemma}
\newtheorem*{lemma*}{Lemma}
\newtheorem{claim}[theorem]{Claim}
\newtheorem{proposition}[theorem]{Proposition}
\newtheorem{fact}[theorem]{Fact}
\newtheorem{corollary}[theorem]{Corollary}

\newtheorem{question}[theorem]{Question}

\theoremstyle{definition}
\newtheorem{definition}[theorem]{Definition}
\newtheorem*{definition*}{Definition}
\newtheorem{remark}[theorem]{Remark}

\newtheorem{example}[theorem]{Example}

\usepackage[linesnumbered,ruled]{algorithm2e}
\crefname{lemma}{Lemma}{Lemmas}
\crefname{fact}{Fact}{Facts}
\crefname{theorem}{Theorem}{Theorems}
\crefname{mtheorem}{Theorem}{Theorems}
\crefname{itheorem}{Theorem}{Theorems}
\crefname{corollary}{Corollary}{Corollaries}
\crefname{claim}{Claim}{Claims}
\crefname{example}{Example}{Examples}
\crefname{algorithm}{Algorithm}{Algorithms}
\crefname{problem}{Problem}{Problems}
\crefname{definition}{Definition}{Definitions}
\crefname{equation}{Eq.}{Eq.}
\crefname{strategy}{Strategy}{Strategies}
\crefname{observation}{Observation}{Observations}

\Crefname{algocf}{Algorithm}{Algorithms}

\usepackage[
letterpaper,
top=1.2in,
bottom=1.2in,
left=1in,
right=1in]{geometry}
\usepackage{newpxtext} 
\usepackage{textcomp} 
\usepackage[scr=rsfso]{mathalfa}
\usepackage{bm} 

\usepackage{microtype}

\usepackage{footnotebackref}

\renewcommand{\bar}{\overline}
\newcommand{\rank}{\operatorname{rank}}

\newcommand{\NRD}{\operatorname{NRD}}
\newcommand{\op}{\operatorname{op}}
\newcommand{\supp}{\operatorname{supp}}
\newcommand{\AND}{\operatorname{AND}}
\newcommand{\Aut}{\operatorname{Aut}}

\allowdisplaybreaks
\newcommand{\FormatAuthor}[3]{
\begin{tabular}{c}
#1 \\ {\small\texttt{#2}} \\ {\small #3}
\end{tabular}
}

\newcommand{\R}{{\mathbb R}}
\newcommand{\N}{{\mathbb N}}

\newcommand{\eps}{\varepsilon}

\newcommand{\F}{{\mathbb F}}

\newcommand{\E}{{\mathbb E}}
\newcommand{\cE}{{\mathcal E}}
\newcommand{\one}{\mathbf 1}

\newcommand{\1}{\mathbf{1}}

\newcommand{\C}{\mathbb C}
\newcommand{\Bits}{\{0,1\}}
\newcommand{\zo}{\Bits}

\newcommand{\ket}[1]{\lvert #1 \rangle}
\newcommand{\braket}{\ket}

\newcommand{\cH}{\mathcal H}

\newcommand{\cG}{\mathcal G}

\newcommand{\poly}{\mathrm{poly}}

\newcommand{\mper}{\,.}
\newcommand{\mcom}{\,,}

\newcommand{\Id}{\operatorname{Id}}

\newcommand{\seq}{\subseteq}
\newcommand{\spn}{\operatorname{span}}

\newcommand{\cM}{\mathcal M}

\newcommand{\cL}{\mathcal L}

\renewcommand{\emptyset}{\varnothing}
\renewcommand{\geq}{\geqslant}
\renewcommand{\leq}{\leqslant}
\renewcommand{\preceq}{\preccurlyeq}
\renewcommand{\succeq}{\succcurlyeq}
\renewcommand{\epsilon}{\varepsilon}

\newcommand{\ignore}[1]{}
\newcommand{\SPR}{\operatorname{SPR}}
\newcommand{\Span}{\spn}

\begin{document}
\author{
  \begin{tabular}{ccc}
    \FormatAuthor{Arpon Basu}{arpon.basu@princeton.edu}{Princeton University} &
    \FormatAuthor{Joshua Brakensiek\thanks{Supported by NSF Award DMS-2503280}}{josh.brakensiek@berkeley.edu}{UC Berkeley} &
    \FormatAuthor{Aaron Putterman\thanks{Supported in part by the Simons Investigator Awards of Madhu Sudan and Salil Vadhan and AFOSR award FA9550-25-1-0112.}}{aputterman@g.harvard.edu}{Harvard University}
  \end{tabular}
}

\title{Many Hamiltonians Are Sparsifiable}
\maketitle

\vspace{-1em}
\begin{abstract}
 We study the problem of Hamiltonian sparsification: given a parameter $\epsilon \in (0,1)$ and an $n$-qubit Hamiltonian $H$ which is the sum of $r$-local positive semi-definite (PSD) terms $H_1, \dots H_m$, our goal is to compute a \emph{sparse} set $L \subseteq [m]$, along with weights $w: L \rightarrow \R_{\geq 0}$ such that for every state $\braket{\psi} \in \mathbb{C}^{2^n}$, 
 \[
\sum_{i \in L} w(i) \cdot \langle \psi {\vert} H_i {\vert} \psi \rangle \in (1 \pm \eps) \cdot \sum_{i = 1}^m\langle \psi {\vert}  H_i {\vert} \psi \rangle.
 \]

When the set $L$ is significantly smaller than $m$, this reduces the number of terms in the underlying system, while still ensuring that the behavior of the system is essentially unchanged. Perhaps surprisingly, despite the strength of this approximation condition, we show that many Hamiltonians indeed are sparsifiable to a number of terms much smaller than $n^r$, including: (a) Hamiltonians where each term is an $r$-local Pauli string, (b) Hamiltonians where each term is an $r$-local \emph{random} operator of rank $R$, for $R \geq 2^{r-1}+1$, and (c) Hamiltonians where each term is an arbitrary $r$-local operator of rank $\geq 2^r -1$ (a.k.a. Quantum SAT). 

Taken together, our results show that the sparsifiability of Hamiltonians is a robust phenomenon, contrary to prevailing belief (see for instance, Aharonov-Zhou ITCS 2019, QIP 2019). Our results find applications, for instance, to better (semi-)streaming algorithms for quantum Max-Cut, answering a question left open by Kallaugher and Parekh (FOCS 2022). In fact, our results even codify that quantum systems are often \emph{easier to sparsify} than their classical counterparts.
    
 On a technical level, our results rely on an extensive analysis of a quantity that we call the ``non-redundancy'' of a Hamiltonian. Roughly speaking, this measures the largest set of terms where each can be ``uniquely satisfied'': i.e., such that for each term $H_i$ in the set, there is a state $\braket{\psi}$ which gives term $H_i$ non-zero energy, and all other terms $H_{\neq i}$ zero energy. This quantity is inherently related to the sparsifiability of the Hamiltonian, but is often easier to work with, as it is governed entirely by the manner in which the kernels of the Hamiltonian terms interact.
\end{abstract}

\pagenumbering{gobble}

\clearpage
 \microtypesetup{protrusion=false}
 \setcounter{tocdepth}{2}
 \renewcommand{\baselinestretch}{1.0}\normalsize
\tableofcontents
\renewcommand{\baselinestretch}{1.1}\normalsize
  \microtypesetup{protrusion=true}
\thispagestyle{empty}
\setcounter{page}{0}

\clearpage

\pagestyle{plain}
\pagenumbering{arabic}

\section{Introduction}

We study the problem of Hamiltonian sparsification:\footnote{Occasionally, this task is also referred to as \emph{dilution}.} here, we are given a parameter $\eps \in (0,1)$, along with an $n$-qubit Hamiltonian $\mathcal{H}$ which is the sum of $r$-local positive semi-definite (PSD) terms $H_1, \dots H_m$. Our goal is to find a \emph{sparse} subset of terms $L \subseteq [m]$, along with new weights $w: L \rightarrow \R_{\geq 0}$ such that for every state $\braket{\psi} \in \C^{2^n}$, 
    \begin{align}\label{eq:preserveEnergy}
     \sum_{i \in L} w(i) \cdot \langle \psi {\vert} H_i {\vert} \psi \rangle \in (1 \pm \eps) \cdot \sum_{i = 1}^m\langle \psi {\vert}  H_i {\vert} \psi \rangle.
    \end{align}
We denote the Hamiltonian with terms $w(i) \cdot H_i: i \in L$, by $\widetilde{\mathcal{H}}$, and call $\widetilde{\cH}$ an $\eps$-sparsifier of $\cH$. Intuitively, \cref{eq:preserveEnergy} ensures that for every state $\braket{\psi}$, the energy of the original Hamiltonian $\mathcal{H}$ and the sparsified Hamiltonian $\widetilde{\mathcal{H}}$ are approximately the same, even for states $\braket{\psi}$ where $\sum_{i = 1}^m\langle \psi {\vert}  H_i {\vert} \psi \rangle$ is small (or zero).

Ideally, the number of terms in the sparsifier $\widetilde{\mathcal{H}}$ should be significantly smaller than the original number of terms in $\mathcal{H}$; in this way, the sparsifier $\widetilde{\mathcal{H}}$ is much simpler, while still ensuring that the behavior of the system is largely unchanged. Note that \cref{eq:preserveEnergy} is particularly hard to guarantee for the \emph{low-energy} subspaces, and in general these ``low-energy simulations'' have seen a considerable amount of work \cite{csahinouglu2021hamiltonian, zlokapa2024hamiltonian}.

To appreciate the implications of \cref{eq:preserveEnergy}, we can consider the case when $\cH$ is a \emph{frustration-free} Hamiltonian (e.g., \cite{AharonovALV10,MichalakisZ13,SattathMLM16}), meaning that the ground-state of $\cH$ is also in the ground-state of each individual term $H_i$. Then, because each term $H_i$ is PSD, the ground-state of $\cH$ has energy $0$. So, by \cref{eq:preserveEnergy}, any sparsifier of $\cH$  will imply preservation of the ground state and spectral gap of the Hamiltonian $H_i$: indeed, any state $\braket{\psi}$ such that $\sum_{i = 1}^m\langle \psi {\vert}  H_i {\vert} \psi \rangle = 0$ also receives zero energy in the sparsifier, and any state $\braket{\psi}$ such that $\sum_{i = 1}^m\langle \psi {\vert}  H_i {\vert} \psi \rangle = \delta$, still receives $\geq (1 - \eps) \delta$ energy in the sparsifier.

Hamiltonian sparsification is not a new direction of study; indeed, inspired by the success of analogous sparsification methods in the classical world (see for instance, works on graph sparsification \cite{BenczurK96, SpielmanS11, SpielmanT11, BatsonSS14}), several works have sought to understand whether this task of Hamiltonian sparsification is feasible, as it would have wide-ranging impacts in both algorithm design and complexity theory. Perhaps most notably, motivated by potential applications of sparsification to the Quantum PCP conjecture \cite{aharonov2013guest}, the work of Aharonov and Zhou \cite{AZ19} provided a negative result in this direction. Indeed, even for the weaker goal of so-called ``gap simulation,'' which intuitively only requires preserving the energy of ground states of the Hamiltonian along with the spectral gap, \cite{AZ19} shows that there \emph{exist} Hamiltonians $\mathcal{H}$ for which sparsification is not possible. In the time since this work \cite{AZ19}, this result has often been cited as a ``no-go'' theorem, showing that Hamiltonian sparsification has inherent limitations (see, i.e., \cite{Gharibian2019complexityof, Kohler2020TranslationallyIU, KallaugherP22}).

However, a recent and separate line of work has instead studied the ability to sparsify constraint satisfaction problems (CSPs), the natural classical analog of Hamiltonians. In this setting, a CSP $C$ is given by a collection of $n$ variables $x_1, \dots x_n \in \zo$, along with $m$ \emph{constraints} of arity $r$, where each constraint applies a predicate $P_i: \zo^r \rightarrow \zo$ to a subset of $r$ variables from $x_1, \dots x_n$. Whereas Hamiltonian sparsification seeks to preserve the energy on every state $\braket{\psi}$, CSP sparsification instead seeks to preserve, for every assignment to the variables $x_1, \dots x_n$, the number of satisfied constraints.

Despite the tremendous success of sparsification in graphs \cite{BenczurK96, SpielmanS11, SpielmanT11, BatsonSS14}, recent work \cite{FiltserK17, ButtiZ20, ChenKN20, KhannaPS24, KhannaPS25, BrakensiekG25, BrakensiekGJLW25} has shown that the sparsification landscape in CSPs is significantly more intricate; indeed, just as shown in Aharonov and Zhou's work \cite{AZ19}, there \emph{exist} classical CSPs where sparsification is \emph{not possible}. However, the results are not entirely negative: the works of \cite{KoganK15, FiltserK17, ButtiZ20, KhannaPS24, KhannaPS25, BrakensiekG25} show that there are huge classes of CSPs (including many natural classes like $\mathbf{XOR}$, $3$-$\mathbf{SAT}$, and $r$-$\mathbf{SAT}$), for which \emph{significant} sparsification is still possible. Intuitively, these works show that the sparsifiability of CSPs is not absolute, but rather depends inherently on the types of interactions (predicates) that are allowed in the CSP. 

This more general understanding of CSP sparsification motivates a natural question about Hamiltonian sparsification:

\begin{quote}
    \center{\emph{Is Hamiltonian sparsification always impossible? Or are there natural classes of Hamiltonians for which significant sparsification is possible?}}
\end{quote}

In what follows, we give a strongly affirmative answer to the second question. 

\subsection{Our Work}

As our first result, we study the sparsifiability of Hamiltonians specified by \emph{Pauli strings} (we call these \emph{Pauli Hamiltonians}). In this setting, we consider a Hamiltonian $\mathcal{H}$ where each local term is of the form $P_1\otimes\cdots\otimes P_n$, where $P_j\in\{\Id, X, Y, Z\}$ acts on the $j^{\mathrm{th}}$ qubit. When we say that these terms are $r$-local, this is equivalent to requiring that $\#\{j: P_j\neq\Id\}\leq r$. Note that arbitrary Pauli strings define operators which have eigenvalues that are $\pm 1$, thus, in our setting we consider Hamiltonians where each local term is a Pauli string shifted by the identity matrix (to ensure that each term is PSD), i.e., $H_i = P_1\otimes\cdots\otimes P_n + \Id$.\footnote{Our results also apply to Pauli strings of the form $\Id - P_1 \otimes \cdots \otimes P_n$.} Such Hamiltonians naturally arise in the analysis of quantum stabilizer codes, with each Pauli string corresponding to a stabilizer term of the code~\cite{AnshuBN23}. For this class of Hamiltonians, we show the following:

\begin{theorem}[Sparsifying Pauli Hamiltonians]
\label{thm:sparspaulihamIntro}
    Let $\cH$ be an $r$-local Pauli Hamiltonian with $m$ terms. Then, for any $\eps > 0$, there exists an $\eps$-sparsifier $L \subseteq [m]$, $w:L\to\R_{\geq 0}$ of size ${\vert}L{\vert}\leq \widetilde{O}_r(\eps^{-2}n)$, i.e. for any $\braket{\psi}\in\C^{2^n}$, 
 \[
\sum_{i \in L} w(i) \cdot \langle \psi {\vert} H_i {\vert} \psi \rangle \in (1 \pm \eps) \cdot \sum_{i = 1}^m\langle \psi {\vert}  H_i {\vert} \psi \rangle.
 \]
    
    Furthermore, there is a randomized (classical) $\poly(n, m)$ time algorithm to compute such a sparsifier.
\end{theorem}
An arbitrary Pauli Hamiltonian can have $\sim n^{r}$ many local terms: one for each possible subset of $r$ out of $n$ qubits. \cref{thm:sparspaulihamIntro} shows that any such system can be significantly sparsified, to a number of terms which is only $\widetilde{O}_r(\eps^{-2}n)$, while still preserving the energy of \emph{every} state $\braket{\psi}$.\footnote{Note that here, $\widetilde{O}_r(\cdot)$ hides poly-logarithmic factors in $(\cdot)$, and $O_r$ hides arbitrary factors in $r$. For our purposes, $r$ is constant and so does not alter the complexity by more than a constant.}

As an aside, it is tempting here to say that \cref{thm:sparspaulihamIntro} implies that \emph{any} Hamiltonian can be sparsified, as one can decompose any $r$-local Hamiltonian term $H_i$ into a weighted sum of $r$-local Pauli strings, which can then be sparsified. However, in general, such a decomposition yields matrices with \emph{negative eigenvalues} in the constituent Pauli strings. To get this into the required PSD form requires shifting \emph{each term} in the decomposition by an identity matrix.

As our second main result however, we show that this sparsification phenomenon extends far beyond just Pauli strings. Indeed, we study the case of \emph{random} Hamiltonians: in this setting, we consider each $H_i$ to be specified by a random PSD matrix $M_i \in \C^{2^r \times 2^r}$ of rank $R$ with $0 < R \leq 2^r$, and a subset $T_i \subseteq [n]$ of qubits, with $H_i = M{\vert}_{T_i} \otimes \Id_{\overline{T_i}}$. As an example, we can take random to mean that $M_i =  \sum_{j = 1}^R v_j v_j^{\dagger}$, where each $v_j \in \C^{\zo^r}$ is a random vector of norm $1$. As we discuss in \cref{sec:sparse-genericLouie}, it turns out that the result below holds for operator matrices $M_i$ sampled from any absolutely continuous distribution; in this sense, the result is quite robust. Note that we \emph{do not} assume that the interaction hypergraph (i.e., the sets $T_i \subseteq [n]$ of which qubits each term operates on) is random; rather, this can be fixed arbitrarily or adversarially.

In this setting, we show that even for these very rank-deficient Hamiltonians, significant sparsification is still possible:

\begin{theorem}[Generic Hamiltonian Sparsification, Informal]\label{thm:genericSparsificationIntro}
    Let $\mathcal{H}$ be a Hamiltonian with $m$ terms $H_1, \dots H_m$, where each term $H_i = (M_i){\vert}_{T_i} \otimes \Id_{\overline{T_i}}$, for $M_i$ a random PSD matrix of rank $R$, where $R$ is a fixed constant $\geq 2^{r-1} + 1$. Then, with high probability over the random matrices $M_i$, for any $\eps > 0$, there is a randomized (classical) algorithm running in time $\mathrm{poly}(n,m)$ which, with probability $1 - 2^{-\Omega(n)}$, computes an $\eps$-sparsifier $L \subseteq [m], w: L \rightarrow \R_{\geq 0}$ of $\cH$ such that ${\vert}L{\vert} = \widetilde{O}_r(\eps^{-2} n)$.
\end{theorem}

In the above, we use $O_r$ to hide large constant factors that depend on the parameter $r$, and $\widetilde{O}$ to hide poly-logarithmic factors in $n$.

\cref{thm:genericSparsificationIntro} shows that for arbitrarily large constants $r$, it is still the case that Hamiltonians with random $r$-local terms of rank $2^{r-1} + 1$ are sparsifiable to size $\widetilde{O}(\eps^{-2}n)$. For comparison, in the classical setting of CSP sparsification, the analog of a random $r$-local operator of rank $2^{r-1} + 1$ would be a random \emph{predicate} $P: \zo^r \rightarrow \zo$ with $2^{r-1} + 1$ satisfying assignments (i.e., ${\vert}P^{-1}(1){\vert} = 2^{r-1} + 1$). In \cref{subsec:separation}, we show that in general these random-predicate CSPs require sparsifiers of size $\Omega(n^{\lceil \log_2(r) - 3\rceil})$ (where the exponent grows, even for constant $r$). Together with \cref{thm:genericSparsificationIntro}, this shows that, contrary to prevailing belief \cite{AZ19}, random Hamiltonians are significantly \emph{more} sparsifiable than their classical random CSP counterparts.

As our final setting, we explore sparsification when each local term $H_i = (M_i){\vert}_{T_i} \otimes \Id_{\overline{T_i}}$, with $M_i$ being a PSD matrix of rank $R \geq 2^{r} - 1$ (equivalently, with nullity $\leq 1$). This problem is often referred to as Quantum SAT in the literature~\cite{Bravyi06,GossetN16,BeaudrapG16}. Here, we show that \emph{any} such collection of Hamiltonian terms is sparsifiable: 

\begin{theorem}[Nullity $1$ Sparsification, Informal (see \cref{thm:Nullity1Sparsifier} for formal statement)]\label{thm:sparsifyNullity1Intro}
    Let $\mathcal{H}$ be a Hamiltonian with $m$ terms $H_1, \dots H_m$, where each term $H_i = (M_i){\vert}_{T_i} \otimes \Id_{\overline{T_i}}$, for $M_i$ being a PSD matrix of rank $R \geq 2^{r} - 1$. Then, for any $\eps > 0$, there is a randomized (classical) algorithm running in time $\mathrm{poly}(n,m)$ which, with probability $1 - 1 / \mathrm{poly}(n)$, computes an $\eps$-sparsifier $L \subseteq [m], w: L \rightarrow \R_{\geq 0}$ of $\cH$ such that ${\vert}L{\vert} = \widetilde{O}(\eps^{-2}n^2)$.
\end{theorem}

Cumulatively, our results show that Hamiltonian sparsification \emph{is possible}, and in fact, \emph{the vast majority} of Hamiltonians can be sparsified to a small number of terms. Our perspective yields immediate insights into basic problems. For instance, one frequently studied class of Hamiltonians is those arising from quantum MAX-CUT. Here, the Hamiltonian $\cH$ is $2$-local, with its terms specified by an interaction graph $G = (V, E)$, where ${\vert}V{\vert} = n$ coincides with the qubits. Each term is of the form $H_e = M_e \otimes \Id_{[n] - e}$, where for an edge $e = (u,v)$, $M_e = \frac{1}{2}(\Id{\vert}_u \otimes \Id{\vert}_v - X{\vert}_u \otimes X{\vert}_v - Y{\vert}_u \otimes Y{\vert}_v - Z{\vert}_u \otimes Z{\vert}_v)$ (see \cref{sec:pauli} for more discussion of this notation), and we let $\cH(G)$ denote the Hamiltonian whose terms correspond to the edges of $G$. Given such a Hamiltonian, the goal is to find a state $\braket{\psi}$ which \emph{maximizes} the energy. We denote this maximum energy by $\mathrm{OPT}(G)$. As a simple consequence of our techniques, we are able to show the following:\footnote{In fact, in \cref{sec:spar-qmcsp}, we even show that similar sparsification can be done for arbitrary quantum MAX-CSP problems.}

\begin{corollary}\label{thm:maxcutSparsifierIntro}
    Let $G = (V, E)$ be an instance of quantum MAX-CUT with corresponding Hamiltonian $\cH$. Then, for any parameter $\eps > 0$, there exists a set $L \subseteq E$, ${\vert}L{\vert} \leq O(\eps^{-2}n)$, along with weights $\tilde{w}: L \rightarrow \R_{\geq 0}$ such that, letting $\tilde{G} = (V, L, \tilde{w})$, for any state $\braket{\psi}$ which satisfies 
    \[
    \langle \psi {\vert} \cH(\tilde{G}) {\vert} \psi \rangle \geq (1 - \eps) \cdot \mathrm{OPT}(\tilde{G}),
    \]
    then it must also be the case that 
    \[
    \langle \psi {\vert} \cH(G) {\vert} \psi \rangle \geq (1 - 2\eps) \cdot \mathrm{OPT}(G).
    \]
    Moreover, this sparsifier $L, \tilde{w}$ can be found efficiently with a classical algorithm, in time polynomial in $m, n$.
\end{corollary}

Note that these MAX-CUT sparsifiers are \emph{weaker} than full-spectrum simulation that the preceding theorems guarantee, as the only guarantee here is that states near the \emph{maximum energy} are preserved. Nevertheless, as the goal in quantum MAX-CUT is to \emph{maximize} energy, \cref{thm:maxcutSparsifierIntro} suffices as a tool to reduce the number of terms prior to solving the MAX-CUT instance, as any near-optimal optimal state in $\cH(\widetilde{G})$ will also be a near-optimal state in $\cH(G)$. In fact, as we show in \cref{sec:spar-qmcsp}, these sparsifiers can even be constructed in the \emph{streaming} setting while still requiring only $\widetilde{O}(\eps^{-2}n)$ bits of space; in so doing, this answers a question left open by Kallaugher and Parekh \cite{KallaugherP22}, as this shows that with only  $\widetilde{O}(\eps^{-2}n)$ bits of space, one can ``information-theoretically'' solve the quantum MAX-CUT problem to accuracy $(1 + \eps)$ in streams. This complements the result of \cite{KallaugherP22}, which shows that any algorithm with $o(n)$ bits of space can at best achieve an approximation factor of $2 - \eps$. Thus, there is an inherent phase transition in the approximability of quantum MAX-CUT as the space allotted to the algorithm grows from $o(n)$ to $\widetilde{O}(n)$.

\subsection{A Quantum Theory of Non-Redundancy}

Underlying our improved sparsifier constructions is an extensive study of a quantity that we call the \emph{non-redundancy} of Hamiltonians, which generalizes the classical theory of CSP non-redundancy~\cite{BessiereCHKLNQW13,BessiereCK20,Carbonnel22,BrakensiekG25,BrakensiekGJLW25,LagerkvistGE26,BrakensiekGP26}. More formally, given a Hamiltonian $\mathcal{H}$ with PSD terms $H_1, \dots H_m$, we say that $\mathcal{H}$ is \emph{non-redundant} if, for every $i \in [m]$, there exists a state $\braket{\psi}$ such that 
\[
\langle \psi {\vert} H_i {\vert} \psi \rangle \neq 0,
\]
yet for every $j \neq i \in [m]$, 
\[
\langle \psi {\vert} H_j {\vert} \psi \rangle = 0.
\]

It turns out that any non-redundant Hamiltonian $\cH$ is an immediate barrier to sparsification: indeed, consider a sparsifier $\widetilde{\cH}$ (with terms $L \subseteq [m]$, $w: L \rightarrow \R_{\geq 0}$) which is missing some term $H_i$ that was present in $\cH$. Then, by the definition of $\cH$ being non-redundant, there must exist a state $\braket{\psi}$ such that $\langle \psi {\vert} H_i {\vert} \psi \rangle \neq 0$. Since $\langle \psi {\vert} H_j {\vert} \psi \rangle = 0$ for every $j \in [m] \setminus \{i\}$, we have that $\langle \psi {\vert}\cH {\vert} \psi \rangle \neq 0$. However  $\langle \psi{\vert}\widetilde{\cH}{\vert} \psi\rangle = 0$ because $H_i$ is \emph{not} in the sparsifier $\widetilde{\cH}$. Thus, for this state $\braket{\psi}$, the sparsification condition is violated, as 
\[
0 =  \langle \psi {\vert}\widetilde{\cH} {\vert} \psi \rangle \notin (1 \pm \eps) \cdot \langle \psi {\vert} \cH {\vert} \psi \rangle.
\]
In fact, sparsification is obstructed in an even stronger sense here: the above shows that the \emph{space of ground-states of the sparsifier} is strictly larger than the ground-space of the original Hamiltonian $\cH$.

In general, when a Hamiltonian $\mathcal{H}$ is redundant (i.e., not non-redundant), we define its non-redundancy to be the size of the largest \emph{subset} of terms which is non-redundant. 

Given the ample success non-redundancy has had in characterizing the sparsifiability of various classical CSPs (in particular, \cite{BrakensiekG25}), we believe that a strong theory of Hamiltonian non-redundancy is necessary toward the further construction of Hamiltonian sparsifiers.

To understand the non-redundancy of a Hamiltonian, one must reason about the manner in which the the ground states of the various terms of $\cH$ interact. For example, as we elaborate upon in \cref{lem:trivialIntersectionIntro}, the crux of proving \cref{thm:genericSparsificationIntro} is to study the interactions between the kernels of two random $r$-local Hamiltonians which share a single qubit. More precisely, we show that if the rank of these individual Hamiltonian terms is sufficiently large, then the ground states of these two local Hamiltonians are entirely \emph{disjoint}! This disjointness condition is critical for us getting lower bounds on the spectrum of our Hamiltonian, leading to a near-optimal sparsifier--this is discussed in much more detail in \cref{subsec:tech-generic}. Similar observations are used in the proof of \cref{thm:sparsifyNullity1Intro}.

However, our contributions to the study of non-redundancy extend far beyond the technical concerns of these sparsifier constructions. In particular, one of our primary contributions is a more systematic understanding of non-redundancy in \emph{translation invariant} Hamiltonian systems where each $r$-local Hamiltonian is the same PSD predicate matrix $M \in \C^{2^r \times 2^r}$. We let $\NRD(M, n)$ denote the size of the largest non-redundant Hamiltonian $\cH$ all of whose terms are copies of $M$ applied to different $r$-qubit subsets (i.e., each term is of the form $M{\vert}_T \otimes \Id_{[n] - T}$).

This is the natural Hamiltonian analogue to the study of the non-redundancy (or sparsifiability) of classical CSPs. In fact, we show in \cref{prop:classical-equivalence} that the study of non-redundancy in Hamiltonians is \emph{strictly} more general than in the CSP world, as one can represent any classical predicate as a \emph{diagonal} matrix $M \in \C^{2^r \times 2^r}$ with $\{0,1\}$ entries. Thus, any classical result on the sparsifiability or non-redundancy of constraint satisfaction problems~\cite{KoganK15, FiltserK17, JansenP19, ButtiZ20, ChenJP20, LagerkvistW20, ChenKN20, Carbonnel22, KhannaPS24, KhannaPS25, BrakensiekG25, BrakensiekGJLW25} immediately translates to understanding Hamiltonians with \emph{diagonal} operators $M$. 

Obviously, the class of $r$-local Hamiltonians is much richer than the space of diagonal matrices. As such, we require novel tools for studying the non-redundancy of $r$-local Hamiltonians.

As our first contribution, we give a criterion for which sparsification is \emph{impossible}. In particular, if our predicate matrix $M$ can be expressed as the tensor product of rank $1$ matrices, then we can construct $r$-local non-redundant Hamiltonians with $\Omega_r(n^r)$ terms. This result can be viewed as a quantum analogue of known lower bounds for the $\AND_r$ predicate~\cite{FiltserK17,ButtiZ20,Carbonnel22,KhannaPS25} as well as a vast generalization of  ``Example A'' in \cite{AZ19}. See \cref{sec:tensor-lb} for full details. 

\begin{theorem}[Tensor Lower Bound]\label{thm:tensor-intro}
    Let $M = M_1\otimes M_2\otimes \cdots\otimes M_r$ be a predicate matrix where $M_1, \ldots, M_r\in\C^{2\times 2}$ are non-zero rank $1$ PSD matrices. Then $\NRD(M, n) \geq (\lfloor n/r\rfloor)^r\geq\Omega_r(n^r)$.
\end{theorem}

As mentioned above, \cref{thm:tensor-intro} implies immediate, stronger lower bounds for Hamiltonian sparsification, both in the sense of preserving the energy on every state $\braket{\psi}$, and also for the weaker task of only preserving the spectral gap. 

Our second contribution is a systematic understanding of predicate matrices $M$ acting on $2$-qubits (analogous to the arity-$2$ predicate classifications of \cite{FiltserK17,ButtiZ20,BessiereCK20}). In particular, we show that the setting of \cref{thm:tensor-intro} is the \emph{only} situation in which $\NRD(M, n)$ is superlinear!

\begin{theorem}\label{thm:2-qubit-intro}
Let $M \in \C^{4 \times 4}$ be a nonzero predicate matrix. Then $\NRD(M, n) = \Theta(n)$ unless $M = M_1 \otimes M_2$ where unary predicate matrices $M_1, M_2 \in \C^{2\times 2}$ have rank one, in which case $\NRD(M, n) = \Theta(n^2)$.
\end{theorem}

Given a Hamiltonian $\cH$ whose terms are each of the form $M{\vert}_T \otimes \Id_{[n] - T}$, the crucial idea behind \cref{thm:2-qubit-intro} is to study the \emph{automorphism group} of the ground state of $\cH$. In particular, we isolate any two terms $H_a, H_b$ in our Hamiltonian $\cH$ which share a single qubit. Say the first term acts on the first two qubits $\{1, 2\}$ and the second term acts on the first and third qubits $\{1, 3\}$. Roughly speaking, we show that for any state in the common ground of both of these two terms, the state \emph{must be} invariant under swapping the second and third qubits (see \cref{lem:nonsingular-aut}). Crucially, the automorphisms of a quantum state can be composed, so from just a handful of these automorphisms we can deduce a large family of additional ones. As such, we can deduce that whenever the underlying interaction graph of Hamiltonians forms certain (bipartite) cycles, the Hamiltonian itself must be redundant~(see \cref{lem:bipartite-cycle-redundant})! From here, a simple linear bound on the non-redundancy quickly follows. See \cref{sec:2-qubit} for further details.

Our final contribution to the study of NRD is a further understanding of the non-redundancy of generic $r$-local Hamiltonians. In particular, while \cref{thm:genericSparsificationIntro} applies to $M \in \C^{2^r \times 2^r}$ of rank at least $2^{r-1}+1$, we can adapt the techniques of \cref{thm:2-qubit-intro} to improve our understanding of translation-invariant Hamiltonians of slightly \emph{lower rank}.

\begin{theorem}\label{thm:GenericNRDUpperBound-intro}
    Let $M \in \C^{2^r \times 2^r}$ be a random predicate matrix of rank $R \geq 2^{r-1}-1$. Then, with probability $1$ over $M$, $\NRD(M, n) \leq \widetilde{O}(n)$.
\end{theorem}

The key technical lemma needed to prove \cref{thm:GenericNRDUpperBound-intro} is that two terms which share a single qubit again induce an automorphism in the ground state. However the automorphism which is created is now a composition of $r-1$ transpositions rather than a single one (see \cref{lem:deriveAutomorphismGeneric}). This makes the structure of the automorphism group of the overall ground state a bit more complex, but we can get a sharp bound on the non-redundancy by combining two facts: (1) each introduced automorphism increases the size of the overall automorphism group by a factor of at least 2, and (2) the total size of the automorphism group is at most ${\vert}S_n{\vert} = n!$. Thus, there can be at most $2\log_2(n!) = \widetilde{O}(n)$ non-redundant constraints.

\subsection{Broader Context and Comparison to Related Work.}

We do remark that our setting differs slightly from that of \cite{AZ19}. In their work, Aharonov and Zhou primarily study the weaker problem of \emph{gap-simulation} mentioned above, which requires preserving the ground states along with the spectral gap of the Hamiltonian (and places no stronger requirements on the higher end of the spectrum). \cite{AZ19} studies this task through the lens of \emph{degree-reduction} (reducing the maximum number of terms that any qubit participates in) and \emph{dilution} (reducing the total number of terms). Our setting aligns with this goal of dilution.

Of particular note however, is the fact that the negative results in \cite{AZ19} rely on an inability to increase the \emph{interaction strength} between any two qubits, which is often important for building physical quantum systems. In our stronger setting which asks for full-spectrum simulation, reducing the interaction strength between qubits is not possible, as this inherently reduces the energy on the upper end of the spectrum. This setting where increased interaction strength is permitted is also studied in \cite{AZ19}, where it is shown that \emph{degree reduction} is possible, albeit, by introducing $\mathrm{poly}(n)$ extra auxiliary qubits (a commonly allowed operation, see \cite{CubittM16, zhou2021strongly}). However, for the task we study in this paper of \emph{dilution} (reducing the number of terms in the Hamiltonian), \cite{AZ19} prove no analogous results, instead showing only impossibility results for certain classes of Hamiltonians. Our results hold \emph{without} introducing extra qubits; the only modification to the underlying Hamiltonian is the removal of certain terms, and the re-weighting of others. Indeed, even the underlying operators in the individual terms are unchanged (up to scaling). 

Nevertheless, the problem we study of Hamiltonian sparsification (or sometimes called ``full spectrum simulation'') is still natural in many contexts. Indeed, as discussed above, many problems like quantum MAX-CUT (i.e., see \cite{PiddockM17,GharibianP19, anshu2020beyond, parekh2021application, KallaugherP22, HwangNP0W23,Piddock25}) or in general, the maximization of the energy of a stoquastic Hamiltonian (see \cite{jordan2010quantum}) require sparsifiers that do preserve the upper end of the Hamiltonian's spectrum, while others require the preservation of the lower end of the spectrum \cite{csahinouglu2021hamiltonian, zlokapa2024hamiltonian}. We thus view our full-spectrum sparsification requirement as a natural \emph{algorithmic} question, as it serves as an important preprocessing step for downstream quantum algorithms. Indeed, by reducing the number of distinct terms (i.e, performing sparsification), downstream algorithms can operate more efficiently on Hamiltonians with smaller descriptive complexity and fewer terms. 

At a broader level, our results point to an ability to perform general size reduction in certain classes of Hamiltonians. Such procedures have been widely studied in the context of the Quantum PCP conjecture previously; see, for instance, \cite{arad2010note, hastings2012trivial, aharonov2013guest, AZ19}. However, our results are lacking in two regards: first, as mentioned above, the sparsification can increase the \emph{interaction strengths} of the remaining terms, and second, our sparsification result \emph{does not} hold for every Hamiltonian. Nevertheless, we view it as a promising open question whether our techniques can be repurposed for this setting. In particular, our primary technical insight of studying \emph{Hamiltonian non-redundancy} deals with this same fundamental objective of trying to preserve the ground state of a Hamiltonian while reducing the number of terms (and thus degree) as much as possible.

Finally, we note that the study of quantum non-redundancy is closely related to the recent study of conditions for the sparsifiability of general families of PSD matrices. Indeed, while the study of CSP sparsification (as described above) using the notion of non-redundancy has been very fruitful, the techniques used to prove these results are highly combinatorial, thus preventing their use in more general contexts. However, \cite{BasuKLM26} generalizes the notion of non-redundancy in a more quantitative manner to define the quantity known as \emph{connectivity threshold}, which allows them to prove sparsification results about families of PSD matrices even when there might not be any underlying combinatorial structure. \cite{BasuKLM26} analyzes a simple ``importance sampling sparsification'' algorithm (\'a la \cite{SpielmanS11}) and shows that just like non-redundancy is the appropriate notion of sparsification in CSP sparsification, connectivity thresholds are the appropriate notion of sparsification in \emph{PSD sparsification}. Our results can be interpreted as developing a framework towards rigorously understanding what these ``connectivity thresholds'' look like in Hamiltonians.

\subsection{Technical Overview}

To give the reader a sample of the technical methods used in this paper, we give proof sketches of two of our main results: \cref{thm:sparspaulihamIntro} on sparsifying Pauli Strings and \cref{thm:genericSparsificationIntro} on sparsifying generic Hamiltonians.

\subsubsection{Pauli String Sparsification}

We now give a brief overview of the proof of \cref{thm:sparspaulihamIntro}. In this setting, we are given some local Pauli strings (each term shifted by identity so that the minimum eigenvalue becomes $0$). Before we describe how we prove \cref{thm:sparspaulihamIntro}, we invite the reader to appreciate why Hamiltonian sparsification is in general a harder problem than classical CSP sparsification. When we frame classical CSP sparsification in terms of Hamiltonians, the corresponding Hamiltonians are diagonal (in the computational basis $\{\braket{0} , \braket{1}\}^{\otimes n}$), and in particular it suffices to reason about the energies of the basis vectors alone, which gives the problem a very combinatorial flavor.

However, given two local Hamiltonians $M_T = M\vert_T\otimes \Id_{\bar{T}}$ and $M_{T'}$, where $T, T'$ intersect, $M_T, M_{T'}$ no longer commute, and the usual combinatorial tools fail.

Nevertheless, for the case of Pauli strings, we can \emph{restore commutativity} by partitioning the underlying hypergraph into \emph{partite} pieces. More formally, suppose the non-identity terms in our Pauli string are $P_1, \ldots, P_r$. Further suppose the hypergraph $\cH$ underlying our Hamiltonian was $r$-partite, i.e. the qubits can be partitioned into sets $V_1, \ldots, V_r$ such that $\cH\seq\bigtimes_{i = 1}^rV_i$. Then notice that when the Pauli string $P_1\otimes\cdots\otimes P_r$ (padded appropriately with $\Id$) acts on $\cH$, $P_1$ only acts on qubits from $V_1$, $P_2$ only acts on qubits from $V_2$, and so on. In particular, even if two hyperedges $T, T'\in\cH$ share some qubits, the respective Hamiltonians on $T, T'$ still commute, thanks to the tensor product structure of Pauli strings!

Thus from a non-commuting family of Pauli Hamiltonians, we can extract a commuting family of Hamiltonians by passing on to a partite sub-instance. Now, because the Hamiltonians commute in a partite instance, they are simultaneously diagonalizable. But notice that a diagonal (in some rotated basis) system of Hamiltonians is equivalent to a classical system of CSPs, which means we can use tools from the theory of classical CSP sparsification \cite{KhannaPS24,KhannaPS25,BrakensiekG25} to resolve the sparsifiability of this family of Pauli Hamiltonians! In fact, the eigenvalues of a partite family of PSD Pauli Hamiltonians behave exactly like the assignment values of a corresponding instance of $\mathbf{XOR}$, and using classical CSP sparsification results on $\mathbf{XOR}$ then yields the near-linear sparsifiability of Pauli Hamiltonians as in \cref{thm:sparspaulihamIntro}.

Finally, it remains to show that any hypergraph can be partitioned into (not too many) partite instances. For the technical overview, we will just explain the proof of this for graphs, i.e., we show that any graph can be decomposed into $\widetilde{O}(1)$ many bipartite graphs: Indeed, we know that the MAX-CUT of any graph contains $\geq 1/2$ the edges of the graph, and furthermore, a cut containing $\geq 1/2$ fraction of the edges can be found very easily, for instance, by choosing every vertex to be in the cut at random with probability $1/2$.\footnote{This random algorithm can also be derandomized without too much effort using a simple ``potential function'' argument to keep track of which vertices to include in the cut depending on which side of the cut it sends more edges to} Consequently, peeling off the MAX-CUT reduces the number of edges in the graph by at least a factor of $2$, showing that this process terminates in $O(\log n)\leq\widetilde{O}(1)$ many rounds. This proof sketch can be generalized to hypergraphs, see \cref{prop:r-partite-WLOG} for further details.

\subsubsection{Generic Sparsification}\label{subsec:tech-generic}

We now provide some intuition for \cref{thm:genericSparsificationIntro}, which relies more heavily on our non-redundancy machinery. Recall that in this setting, we are given a Hamiltonian $\mathcal{H}$ with $r$-local terms $H_1, \dots H_m$. Notationally, each $H_i = (M_i){\vert}_{T_i} \otimes \Id_{\overline{T_i}}$, where $M_i$ is a \emph{random} PSD matrix of rank $R \geq 2^{r-1} + 1$. For now, we take random to mean that $\spn(M) = \spn(v_1, \dots v_R)$, where each $v_j \in \C^{2^r}$ and is chosen \emph{at random} among unit norm vectors in $\C^{2^r}$. As we shall see however, the exact manner in which these vectors is chosen is not important, as sparsifiability of these systems proves to be a robust phenomenon. 

As discussed above, our key technique towards understanding the sparsifiability of a Hamiltonian $\mathcal{H}$ is in understanding the maximal \emph{non-redundant sub-instances}. In this direction, it is natural to then consider when, for a state $\braket{\psi}$, it is the case that there are many terms $H_i: i \in [m]$ such that $\langle \psi {\vert} H_i {\vert} \psi \rangle = 0$.

\paragraph{Understanding the Non-Redundancy of Generic Hamiltonians}

In fact, in this setting with random operators $M_i$, it turns out that non-redundancy is relatively well-behaved. The key insight here is that whenever two terms $H_i, H_j$ \emph{intersect} (in the sense that they operate on a shared qubit), this imposes a very powerful constraint on which states $\braket{\psi}$ can be in the kernel of $H_i$ and $H_j$. Formally, we show the following:

\begin{lemma}\label{lem:trivialIntersectionIntro}
    Let $H_i = (M_i){\vert}_{T_i} \otimes \Id_{\overline{T_i}}$ and  $H_j = (M_j){\vert}_{T_j} \otimes \Id_{\overline{T_j}}$, where $T_i, T_j \subseteq [n]$, and $T_i \cap T_j \neq \emptyset$. Then, with probability $1$ over $M_i, M_j$ (as defined above), $\ker(H_i) \cap \ker(H_j) = 0 \in \C^{2^n}$.
\end{lemma}

That is to say, the above lemma implies that there are \emph{no} states of norm $1$ which are in the kernels of both $H_i$ and $H_j$. With \cref{lem:trivialIntersectionIntro} in hand, we can now bound the size of any non-redundant Hamiltonian of the above form by $n$: indeed, let us suppose for the sake of contradiction that a Hamiltonian of the above form has $\geq n+1$ terms. Then, it must be the case that there is some qubit which is shared between two terms. Call these terms $H_a$ and $H_b$. Now, let us consider any other term $H_c$.

In order for $\cH$ to be a non-redundant Hamiltonian, it must be the case that there exists a state $\braket{\psi}$ such that 
\[
\langle \psi {\vert} H_a {\vert} \psi \rangle = \langle \psi {\vert} H_b {\vert} \psi \rangle = 0,
\]
while $\langle \psi {\vert} H_c{\vert} \psi \rangle > 0$. But, this implies that $\braket{\psi} \in \ker(H_a) \cap \ker(H_b)$, and so by \cref{lem:trivialIntersectionIntro}, $\braket{\psi}$ must be trivial. But then $\langle \psi {\vert} H_c{\vert} \psi \rangle = 0$ as well, which is a contradiction.

Proving \cref{lem:trivialIntersectionIntro} is more subtle; the intuition here is that a constraint of the form $\langle \psi {\vert} H_a {\vert} \psi \rangle = 0$ is imposing $2^{n -r} \cdot R$ many linear constraints on the state $\braket{\psi} \in \C^{2^n}$. The factor of $R$ comes from the rank of the matrix $M_a$, and the factor of $2^{n-r}$ comes from the tensor product with $\Id_{\overline{T_a}}$. Thus, when we have two terms $H_a$ and $H_b$, any state $\braket{\psi}\in \ker(H_a) \cap \ker(H_b)$ must obey a set of $2 \cdot 2^{n -r} \cdot R$ many linear constraints. Proving that there is no non-trivial state $\braket{\psi}$ which satisfies these constraints is then equivalent to showing that the \emph{dimension} spanned by these combined constraints is $2^n$.

Unfortunately, it turns out that if $H_a$ and $H_b$ are \emph{disjoint} (i.e., operating on non-intersecting groups of qubits), then the combined constraints have dimension $< 2^n$ and so there are non-trivial states $\braket{\psi}$ in $ \ker(H_a) \cap \ker(H_b)$. Our argument in \cref{sec:intersectingTermGenericLouie} shows however that, when $H_a$ and $H_b$ \emph{do} share qubits, and satisfy a non-degeneracy condition among their spans which we call ``genericity'' (which holds with probability $1$ for random operators), then in fact the combination of their implied linear constraints \emph{does} span the entire $2^n$ dimensional space. 

We omit the technical analysis here, but our argument relies on carefully casting this span condition as the determinant of a specific $2^n \times 2^n$ sub-matrix of these $2 \cdot 2^{n -r} \cdot R$ many linear constraints. If this determinant is non-zero, then this immediately implies that their span is $2^n$ dimensional. Importantly, we show that when we treat the entries of the matrices $M_a, M_b$ as variables, this determinant polynomial is \emph{not} the zero polynomial: thus, when the entries are chosen at random, the determinant is non-zero with probability $1$.

\paragraph{Building Sparsifiers For Generic Hamiltonians}

\cref{lem:trivialIntersectionIntro} also has immediate applications in building sparsifiers. Indeed, consider the case when we have two terms, $T_a, T_b$ such that $T_a \cap T_b \neq \emptyset$. By \cref{lem:trivialIntersectionIntro}, we know that $\ker(H_a + H_b)$ is trivial, as there is no other non-trivial state which can give both $H_a$ and $H_b$ zero energy. 

Even more importantly, we know that $H_a + H_b$ can be written as a tensor product of some matrix $(M_{a, b}){\vert}_{T_a \cup T_b} \otimes \Id_{[n] - (T_a \cup T_b)}$, where $M_{a, b}$ operates on $\leq 2r-1$ many qubits. Thus, the eigenvalues of $H_a + H_b$, are exactly the eigenvalues of this matrix $M_{a, b}$ (with multiplicity). Here is where we make use of an additional key observation: the matrix $M_{a, b} \in \C^{2^{O(r)} \times 2^{O(r)}}$, where $r$ is a \emph{constant} in our setting. Thus, we have a constant size PSD matrix, with an empty null-space. So, we can then say that all eigenvalues of $M_{a, b}$ are at least $\Omega_{r}(1)$ (we use the subscript of $r$ to denote that there may be (even exponential) dependence on the values in the matrices $M_a, M_b$ or the constant $r$). We formalize this reasoning in \cref{sec:intersectingTermGenericLouie}; the key intuition is that these matrices are sampled from a distribution parameterized only by $r$, thus \emph{independent} of $n$, the number of qubits.

Continuing, suppose we are able to find many ``disjoint pairs of intersecting terms''; call these $H^{(j)}_a, H^{(j)}_b: j \in [C]$, where $H^{(j)}_a$ operates on qubits $T_a^{(j)} \subseteq [m]$ and $H^{(j)}_b$ operates on qubits $T_b^{(j)} \subseteq [m]$ such that $T_a^{(j)} \cap T_b^{(j)} \neq \emptyset$. Then, we know that the smallest eigenvalue of $H$ must scale with the number of these disjoint pairs that we find, and so must be at least $C \cdot \Omega_{ r}(1)$. But, we now can provide a tight bound on the so-called \emph{importance score} of each term $H_i \in H$. Indeed, we see that:
\[
\mathrm{Importance}(H_i) := \max_{\braket{\psi}} \frac{\langle \psi {\vert} H_i {\vert} \psi \rangle}{\langle \psi {\vert} H {\vert} \psi \rangle} \leq \max_{\braket{\psi}} \frac{\langle \psi {\vert} H_i {\vert} \psi \rangle}{\sum_{j \in [C]}\langle \psi {\vert} (H^{(j)}_a+ H^{(j)}_b) {\vert} \psi \rangle} 
\]
\[
\leq \frac{\lambda_{\max}(H_i) \cdot \Vert \psi \Vert_2^2}{\lambda_{\min}(\sum_{j \in [C]}H^{(j)}_a+ H^{(j)}_b) \cdot \Vert \psi \Vert_2^2}  \leq O_{r}\left ( \frac{1}{C} \right ).
\]

The key observation is that once $C$ is large, say $\frac{n}{\eps^2}$, then this provides a very strong bound on the ``importance score'' of all constituent terms of the Hamiltonian. In fact, this bound is sufficiently strong such that we can \emph{directly} invoke Matrix Chernoff to argue that if we randomly sample the terms of $H$ at rate $\approx \frac{n}{\eps^2 C}$, and give the sampled terms weight $\approx \frac{\eps^2 C}{ n}$ (to yield a new, sparser, Hamiltonian $\widetilde{H}$), then it is the case that, with high probability, for every state $\braket{\psi}$,
\[
\langle \psi {\vert} \widetilde{H} {\vert} \psi \rangle \in (1 \pm \eps) \langle \psi {\vert} H {\vert} \psi \rangle.
\]

The argument then concludes by observing that whenever the number of terms $m$ is larger than (say) $10n$, we can indeed find $\geq \frac{m}{4}$ disjoint pairs of intersecting terms. Thus, in the above, $C = \frac{m}{4}$, and so the resulting sample retains only $\approx m \cdot \frac{n}{C \eps^2} \leq O(\eps^{-2} n)$ many terms.

A more careful version of this argument is used to ultimately prove \cref{thm:genericSparsificationIntro}.

\subsubsection*{Outline}

In \cref{sec:prelims}, we formally define Hamiltonian sparsification, set up the notation we use throughout the paper, and prove a number of basic results on the structure of sparsifiers. In \cref{sec:pauli}, we prove \cref{thm:sparspaulihamIntro} on the sparsifiability of Pauli strings. In \cref{sec:sparse-genericLouie}, we prove \cref{thm:genericSparsificationIntro} on the sparsifiability of generic Hamiltonians. In \cref{sec:nullity-1}, we prove \cref{thm:sparsifyNullity1Intro} on sparsifying Hamiltonians where each term has nullity at most one. In \cref{sec:spar-qmcsp}, we show how to sparsify any quantum MAX-CSP (including quantum MAX-CUT), proving \cref{thm:maxcutSparsifierIntro} in the process. In \cref{sec:tensor-lb}, we prove \cref{thm:tensor-intro} on the non-redundancy of Hamiltonians which are tensor products. In \cref{sec:2-qubit}, we prove \cref{thm:2-qubit-intro} on the non-redundancy of arity-$2$ local Hamiltonians. In \cref{sec:nrd-generic}, we prove \cref{thm:GenericNRDUpperBound-intro} on the non-redundancy of random Hamiltonians. In \cref{sec:conclusion}, we give some concluding remarks and open questions.

\section{Preliminaries}\label{sec:prelims}

Throughout the paper we will work with Hamiltonians on $n$ qubits. The ambient Hilbert space of this quantum system is $\bigotimes_{i = 1}^n\C^2\cong\C^{\zo^n}\cong\C^{2^{[n]}}\cong\C^{2^{n}}$. For any $\psi\in\C^{\zo^n}$, we will equivalently treat $\psi$ as a function $\zo^n\to\C$. Note that we do not require the state $\psi$ to be normalized such that $\Vert \psi \Vert_2 = 1$, as the energy will simply scale proportionally with the square of the norm.

\begin{definition}[Restriction of a state]\label{def:restrictedState}
If $S\seq[n]$ is a subset and if $z\in\zo^S$, then we define $\psi_z$ to be a function $\zo^{[n]\setminus S}\to\C$, where for any $y\in\zo^{[n]\setminus S}$, $\psi_z(y):= \psi(z \circ y)$, where $z\circ y\in\zo^S\times\zo^{[n]\setminus S}\cong\zo^{[n]}$ is just the concatenation of $z, y$.
\end{definition}

We define a \emph{predicate matrix} of \emph{arity $r$} to be a Hermitian Positive Semi-Definite (PSD) matrix $M \in \C^{2^r \times 2^r}$. Given a sequence $T \in [n]^r$ of $r$ distinct qubits, we let $M_T:= M\vert_T\otimes\Id_{[n]\setminus T}$. For a state $\psi \in \C^{2^n}$, we define its \emph{energy} with respect to $M_T$ to be $\cE_{T,M}(\psi) = \langle\psi, M_T \psi \rangle$.

Note that 
\begin{align}
\label{eq:etmexpression}
    \cE_{T, M}(\psi) = \sum_{z\in\zo^{[n]\setminus T}}\langle\psi_z, M\psi_z\rangle\mper
\end{align}
Note that $\psi_z$ can be treated as a vector in a $2^r$-dimensional Hilbert space.

We define a Hamiltonian $\cH$ to be a set of pairs $(T, M)$ where $T \in [n]^r$ and $M \in \C^{2^r \times 2^r}$. Given a Hamiltonian $\cH$, we define the energy of state $\psi$ to be $\cE_{\cH}(\psi) := \sum_{(T, M) \in \cH} \cE_{T,M}(\psi)$.  If $\cE_{\cH}(\psi) = 0$, then we say that $\psi$ is in the \emph{ground state} of $\cH$. We let $\Psi_{\cH}$ denote the entire ground state. Observe that
\begin{align}
\Psi_{\cH} := \bigcap_{(T,M) \in \cH} \ker M_T.\label{eq:ground-state}
\end{align}

Given a weight function $w : \cH \to \R_{\geq 0}$, we define the $w$-weighted energy of $\cH$ for all states $\psi \in \C^{2^n}$ to be
\[
    \cE_{\cH, w} := \sum_{(T, M) \in \cH} w(T, M) \cE_{T,M}.
\]

An $\eps$-\emph{sparsifier} of $\cH$ is a weight function $w : \cH \to \R_{\geq 0}$ such that for all states $\psi \in \C^{2^n}$ we have that
\begin{align}
    \cE_{\cH, w}(\psi) \in [1-\eps, 1+\eps] \cdot \cE_{\cH}(\psi).\label{eq:quantum-sparsifier}
\end{align}
Note that \cref{eq:quantum-sparsifier} is equivalent to saying that 
\[(1 - \eps)\sum_{(T, M) \in \cH} M_T\preceq\sum_{(T, M) \in \cH} w(T, M)M_T\preceq(1 + \eps)\sum_{(T, M) \in \cH} M_T\]
where $\preceq$ refers to the Loewner order on Hermitian matrices. \footnote{For two Hermitian matrices $A, B$, we say $A\preceq B$ if $B - A$ is PSD}
\begin{remark}\label{rem:convex}
Note that the set of $\eps$-sparsifiers of $\cH$ form a convex set. In particular, any Hamiltonian can be evolved to any of its sparsifiers in a continuous (in fact, linear) manner.
\end{remark}
We define the \emph{size} of $w$ to be cardinality of $w^{-1}(\R_{> 0}) = \supp(w)$. We denote by $\SPR(\cH, \eps)$ the minimum size of all $\eps$-sparsifiers of $\cH$.

In some situations, we assume that all predicates of $\cH$ are the same matrix $M \in \C^{2^r \times 2^r}$, an assumption known as \emph{translation invariance}. In that case, we call $\cH$ an $M$-Hamiltonian, and simplify notation by having $\cH$ be the set of all $r$-tuples $T \in [n]^r$ for which $M$ is applied. We let $\SPR(M, n, \eps)$ be the maximum value of $\SPR(\cH, \eps)$ among all $M$-Hamiltonians on $n$ qubits.

Since $M\in\C^{2^r\times 2^r}$, any parameter associated to $M$ --- e.g., the largest eigenvalue of $M$ --- can be bounded as a function of $r$ alone. Thus, throughout the paper we will use notations $O_M(\cdot), \Omega_M(\cdot), \Theta_M(\cdot)$ to suppress constants dependent only on $r$ and $M$.

We say that Hamiltonian $\cH$ is \emph{$r$-partite} if there exists a partition of the vertices $[n] = V_1 \cup \cdots \cup V_r$ such that $\cH \subseteq V_1 \times \cdots \times V_r$. In many proofs it is convenient to assume that $\cH$ is $r$-partite. We show this assumption can be made without loss of generality up to $\widetilde{O}_r(1)$ factors in the sparsifier size.

\begin{proposition}\label{prop:r-partite-WLOG}
Let $\cH$ be an $r$-local Hamiltonian on $n$ qubits. Then in deterministic $\poly_r({\vert}\cH{\vert})$ time one can compute a partition $\cH = \bigsqcup_{i\in\mathcal{I}}\cH_i$ such that each $\cH_i$ is $r$-partite, and ${\vert}\mathcal{I}{\vert}\leq O_r(\log n)$. In particular, there exists an $r$-partite Hamiltonian $\cH'\seq\cH$ such that $\SPR(\cH', \eps) \geq \Omega_{r}\left(\frac{\SPR(\cH, \eps)}{\log n}\right)$.
\end{proposition}

\begin{proof}
Note that if $(T, M_1), (T, M_2) \in \cH$, then can create an equivalent Hamiltonian by replacing both with $(T, M_1 + M_2)$. Thus, we may assume without loss of generality that ${\vert}\cH{\vert} \leq n^r$. We now describe a procedure for constructing this partition $\cH = \bigsqcup_{i\in\mathcal{I}}\cH_i$ such that each $\cH_i$ is $r$-partite. We begin by finding one large $r$-partite subgraph of $\cH$.

\paragraph{Extracting a large $r$-partite subgraph.} Consider a uniformly random function $p : [n] \to [r]$. Given such a function $p$, we can construct an $r$-partite graph $\cH_p$ such that for every $(t_1, \hdots, t_r) \in \cH$, we have that $(t_1, \hdots, t_r) \in \cH_p$ if and only if $(p(t_1), \hdots, p(t_r)) = (1, \hdots, r)$. If our function $p : [n] \to [r]$ is chosen uniformly at random, the probability that $(p(t_1), \hdots, p(t_r)) = (1, \hdots, r)$ is precisely $r^{-r}$. In particular, $\mathbb E_{p}[{\vert}\cH_p{\vert}] = r^{-r} {\vert}\cH{\vert}$. Thus, there exists a function $\widehat{p} : [n] \to [r]$ such that ${\vert}\cH_{\widehat{p}}{\vert} \geq r^{-r}{\vert}\cH{\vert}$. To deterministically construct $\widehat{p}$, consider the following potential function for any partial assignment $p: [n] \to [r] \cup \{\bot\}$ (where $\bot$ denote an undefined output).
\begin{align*}
    \Phi_{(t_1, \hdots, t_r)}(p) &:= \prod_{a \in [r]} \delta_r(a, p(t_a)),\\
    \Phi_{\cH}(p) &:= \sum_{(t_1, \hdots, t_r) \in \cH}  \Phi_{(t_1, \hdots, t_r)}(p),
\end{align*}
where
\[
    \delta_r(a, b) = \begin{cases}
    1 & a = b\\
    \frac{1}{r} & b = \bot\\
    0 & \text{otherwise}.
    \end{cases}
\]
\begin{algorithm}
\caption{Find $p : [n] \to [r]$ such that ${\vert}\cH_p{\vert} \geq r^{-r} {\vert}\cH{\vert}$}\label{alg:find-partition}
Start with $p_0(i) = \bot$ for all $i \in [n]$.\\
\For{$i \in [n]$} {
    \For{$a \in [r]$} {
        Let $p_{i,a}$ be $p_{i-1}$ except $p(i) = a$.
    }
    Assign $p_i$ to be some $p_{i,a}$ which maximizes $\Phi_{\cH}(p_{i,a})$.
}
\Return{$p_n$.}
\end{algorithm}

The algorithm is clearly deterministic and runs in $\poly_r({\vert}\cH{\vert})$ time. To see why ${\vert}\cH_{p_n}{\vert} \geq r^{-r}{\vert}\cH{\vert}$. Observe that $\Phi_{\cH}(p_0) = r^{-r}{\vert}\cH{\vert}$. And $\Phi_{\cH}(p_n) = {\vert}\cH_{p_n}{\vert}$. Thus, it suffices to prove that $\Phi_{\cH}(p_i) \geq \Phi_{\cH}(p_{i-1})$ for all $i \in [n]$. We do this by comparing $\Phi_{(t_1, \hdots, t_r)}(p_{i-1})$ and $\Phi_{(t_1, \hdots, t_r)}(p_{i,a})$ for all $(t_1, \hdots, t_r) \in \cH$ and all $a \in [r]$. We break our analysis into cases.

First, if none of $t_1, \hdots, t_r$ equals $i$, then for all $a \in [r]$ we have that
\[
    \Phi_{(t_1, \hdots, t_r)}(p_{i,a}) = \prod_{b \in [r]} \delta_r(b, p_{i,a}(t_b)) = \prod_{b \in [r]} \delta_r(b, p_{i-1}(t_b)) = \Phi_{(t_1, \hdots, t_r)}(p_{i-1}).
\]
Otherwise, assume that $t_a = i$ for some $a \in [r]$. Then,
\begin{align*}
    \Phi_{(t_1, \hdots, t_r)}(p_{i,a}) = \prod_{b \in [r]} \delta_r(b, p_{i,a}(t_b)) &= \delta_r(a, p_{i,a}(a))\prod_{b \in [r] \setminus \{a\}} \delta_r(b, p_{i-1}(t_b))\\
    &= r\delta_r(a, p_{i-1}(a))\prod_{b \in [r] \setminus \{a\}} \delta_r(b, p_{i-1}(t_b))\\
    &= r\Phi_{(t_1, \hdots, t_r)}(p_{i-1}).
\end{align*}
By similar logic, $\Phi_{(t_1, \hdots, t_r)}(p_{i,b}) = 0$ for all $b \in [r] \setminus \{a\}$. Thus, we have shown for any $(t_1, \hdots, t_r) \in \cH$ we have that
\[
    \frac{1}{r}\sum_{a \in [r]} \Phi_{(t_1, \hdots, t_r)}(p_{i,a}) = \Phi_{(t_1, \hdots, t_r)}(p_{i-1}).
\]
Therefore,
\[
    \Phi_{\cH}(p_i) \geq \frac{1}{r}\sum_{a \in [r]} \Phi_{\cH}(p_{i,a}) = \Phi_{\cH}(p_{i-1}),
\]
as desired. Hence, \cref{alg:find-partition} succeeds in finding $p : [n] \to [r]$ such that ${\vert}\cH_p{\vert} \geq r^{-r} {\vert}\cH{\vert}$ .

\paragraph{Building a partition.} We now describe how to deterministically construct a partition $\cH = \bigsqcup_{i\in\mathcal{I}}\cH_i$ such that each $\cH_i$ is $r$-partite. As a base case, we use \cref{alg:find-partition} to find $\cH_1 \subseteq \cH$ with ${\vert}\cH_1{\vert} \geq r^{-r}{\vert}\cH{\vert}$. Then, given $\cH_1, \hdots, \cH_i$ for some integer $i$, we consider $\cH_{-i} := \cH \setminus (\cH_1 \cup \cdots \cup \cH_i)$ and run \cref{alg:find-partition} of $\cH_{-i}$ to find $\cH_{i+1} \subseteq \cH_{-i}$ such that ${\vert}\cH_{i+1}{\vert} \geq r^{-r}{\vert}\cH_{-i}{\vert}.$ This procedure stops when $\cH = \cH_1 \cup \cdots \cup \cH_{\ell}$ for some positive integer $\ell$ (i.e., $\cH_{-\ell} = \emptyset$). This procedure is entirely deterministic and runs in finite time as every step assigns at least one hyperedge of $\cH$ to the partition that was not assigned before.

We claim more strongly that $\ell \leq O_r(\log n)$. Define $\cH_{-0} := \cH$. For all $i \in [\ell]$, note by design that ${\vert}\cH_{i}{\vert} \geq r^{-r}{\vert}\cH_{-(i-1)}{\vert}$. Thus, ${\vert}\cH_{-i}{\vert} \leq (1-r^{-r}){\vert}\cH_{-(i-1)}$. By combining these inequalities, we have that
\[
    {\vert}\cH_{-\ell}{\vert} \leq (1-r^{-r})^{\ell}{\vert}\cH{\vert} \leq (1-r^{-r})^{\ell} n^r.
\]
In particular, if $\ell \geq \frac{\log (2n^r)}{-\log (1-r^{-r})} = \Theta_r(\log n)$, then $(1-r^{-r})^{\ell} n^r \leq \frac{1}{2}$, so $\cH_{-\ell} = \emptyset$, showing that our procedure terminated in $O_r(\log n)$ steps.

\paragraph{Bounding the sparsifier size.} Since $\cH_1, \hdots, \cH_\ell$ partition $\cH$, we can combine (by taking a disjoint union) any $\eps$-sparsifiers of $\cH_1, \hdots, \cH_\ell$ to obtain an $\eps$-sparsifier of $\cH$. Therefore, we have the inequality.
\[
    \sum_{i=1}^{\ell} \SPR(\cH_i, \eps) \geq \SPR(\cH, \eps).
\]
Therefore, there exists $i \in [\ell]$ for which $\SPR(\cH_i, \eps) \geq \frac{1}{\ell}\SPR(\cH, \eps) \geq \Omega_{r}\left(\frac{\SPR(\cH, \eps)}{\log n}\right)$, as desired.
\end{proof}
The way we will typically use \cref{prop:r-partite-WLOG} is as follows: Given a Hamiltonian sparsification problem supported on a hypergraph $\cH$, we use \cref{prop:r-partite-WLOG} to partition $\cH$ into partite hypergraphs, and assume thereon that our underlying hypergraph is partite. After we have sparsified the partite hypergraphs, we can put them together, incurring only a $O_r(\log n)$ blow-up in the size of our sparsifier, which is negligible in everything we consider in this paper.

\subsection{Non-redundancy}

A crucial notion used in the classical CSP sparsification literature is that of \emph{non-redundancy}. We define a quantum analgoue here. Given a Hamiltonian $\cH$, we say that $\cH$ is \emph{non-redundant} if for all $T,M) \in \cH$, 
we have that
\[
    \Psi_{\cH \setminus (T, M)} \neq \Psi_{\cH}.
\]
In other words, each local Hamiltonian of $\cH$ makes a meaningful contribution to the ground state. Conversely, if $\Psi_{\cH \setminus (T, M)} = \Psi_{\cH}$ from some $(T,M) \in \cH$, then we say that $\cH$ is \emph{redundant}. Given a predicate matrix $M \in \C^{2^r \times 2^r}$, we define $\NRD(M, n)$ to be size of the largest non-redundant $M$-Hamiltonian $\cH$ on $n$-qubits. By adapting the proof of \Cref{prop:r-partite-WLOG}, we can show that without loss of generality that our instance $\cH$ is $r$-partite.

We also use the connectivity parameter notion from PSD sparsification:

\begin{definition}[Connectivity Parameter]
Let $\mathcal{A}= \{A_1, \dots A_m\}$ be a collection of PSD matrices. For any $\alpha \geq 0$, the \emph{connectivity parameter} with strength $\alpha$, denoted $N(\alpha ; \mathcal{A})$ is the smallest integer such that for any $T \subseteq [m]$ with ${\vert}T{\vert} \geq N(\alpha ; \mathcal{A})$, there exists an $i \in T$ satisfying $\alpha A_i \preceq \sum_{j \in T \setminus \{i\}} A_j$. If no such integer exists, we say $N(\alpha ; \mathcal{A}) = m+1$.
\end{definition}

The key fact we know about the connectivity parameter is:

\begin{theorem}[Theorem 1.5 in \cite{BasuKLM26}]\label{thm:sparsifyBKLM}
\label{thm:nstarsparsifier}
    Let $\mathcal{A} = \{A_1, \dots A_m\} \subseteq \C^{2^n \times 2^n}$ be a collection of Hermitian PSD matrices. Then, for any $T \subseteq [n]$ and $\eps \leq 1/2$, there exists a sparse set of weights $w: [m] \rightarrow [0, \infty)$ such that 
    \[
    (1 - \eps)\sum_{i \in [m]} A_i\preceq\sum_{i \in [m]} w(i) A_i \preceq (1 + \eps)\sum_{i \in [m]} A_i
    \]
    and 
    \[
    {\vert}\mathrm{supp}(w){\vert} \leq O(1) \cdot \left ( \min_{\alpha \in (0,1]} \frac{N(\alpha ; \mathcal{A})}{\alpha} \right ) \cdot \left ( \frac{n \log(m)}{\eps^2}\right ).
    \]
\end{theorem}
\begin{remark}
    Although \cite{BasuKLM26} prove the above theorem only for real symmetric PSD matrices, the proof goes through verbatim for (complex) Hermitian PSD matrices.
\end{remark}

Going forward, we will often refer to the connectivity threshold of a single constraint matrix $M \in \C^{2^r \times 2^r}$:

\begin{definition}
    For a single constraint matrix $M \in \C^{2^r \times 2^r}$, we say that $N(\alpha ; n ; M)$ is the maximum over constraints $T_1, \dots T_m$, $T_i \in [n]^r$ of $N(\alpha; \mathcal{A})$, where $\mathcal{A} = \{M_{T_i}: i \in [m] \}$.
\end{definition}

For any predicate matrix $M$, write $N(\alpha; M, n)$ to denote $N(\alpha; \{M_{T}\}_{T\in[n]^r})$. We'll drop `$n$' from $N(\alpha; M, n)$ when it is clear from context.

\begin{lemma}[$N(\alpha; M, n)$ lower bound]
\label{lem:nalphalowerbound}
    Suppose $M\neq 0$ is not full rank. Then $N(\alpha; M, n)\geq \lfloor n/r\rfloor + 1$ for any $\alpha > 0$. In particular, $\NRD(M, n)\geq\lfloor n/r\rfloor$.
\end{lemma}
\begin{proof}
    Write $m:= \lfloor n/r\rfloor$, and let $T_1, \ldots, T_m\in[n]^r$ be disjoint sets. Let $u\in\C^{2^r}$ be a unit vector in the kernel of $M$. Since $M\neq 0$, let $v\in\C^{2^r}$ be a unit vector such that $Mv\neq 0$. Now consider the vector $\psi_i\in\C^{2^n}$ given by $\bigotimes_{j < i}u\vert_{T_j}\otimes v\vert_{T_i}\otimes\bigotimes_{j > i}u\vert_{T_j}$. \footnote{if $\lfloor n/r\rfloor < n/r$, we ``pad'' the remaining dimensions in the expression for $\psi_i$ by any unit vector on the remaining qubits} Then 
    \[\langle\psi_i, M_{T_i}\psi_i\rangle = \prod_{j < i}\langle u, \Id u\rangle\cdot\langle v, Mv\rangle\cdot\prod_{j > i}\langle u, \Id u\rangle = \langle v, Mv\rangle > 0\mcom\]
    while a similar expansion of the quadratic form $\langle\psi_i, M_{T_j}\psi_i\rangle$ (for $j\neq i$) witnesses the term $\langle u, Mu\rangle = 0$, and thus $\langle\psi_i, M_{T_j}\psi_i\rangle = 0$. Thus the vectors $\psi_1, \ldots, \psi_m$ witness an instance of $N(\alpha; M)$ for any $\alpha > 0$, as desired. Finally, the assertion about $\NRD(M, n)$ follows since $\lim_{\alpha\to 0}N(\alpha; M, n) = \NRD(M, n) + 1$. 
\end{proof}

\subsection{Comparison to Classical CSP Sparsification.} We say that a predicate matrix $M \in \C^{2^r \times 2^r}$ is \emph{classical} if it is a diagonal matrix with $\{0,1\}$ entries. In this case, $M$ has a corresponding relation $R \subseteq \{0,1\}^r$ such that 
\[
M{\vert}x\rangle = \begin{cases}
{\vert}x\rangle & x \in R\\
0 & \text{otherwise}.
\end{cases}
\]
Given any $r$-uniform hypergraph $\cH$ on vertex set $[n]$, we define a $\eps$-\emph{CSP sparsifier} of $\cH$ with respect to the relation $R$ to be a reweighting $w : \cH \to \R_{\geq 0}$ such that for all assignments $x \in \{0,1\}^r$ we have that
\begin{align}
    \sum_{T \in \cH} w(T) \one[x{\vert}_{T} \in R] \in [1-\eps, 1+\eps] \cdot \sum_{T \in \cH} \one[x{\vert}_{T} \in R]\label{eq:CSP-sparsification}
\end{align}
We now show when $M \in \C^{2^r \times 2^r}$ is classical, Hamiltonian sparsification and the corresponding CSP sparsification problem are equivalent.

\begin{proposition}\label{prop:classical-equivalence}
Let $\cH$ be an $r$-uniform hypergraph on vertex set $[n]$. Let $M \in \C^{2^r \times 2^r}$ be a classical predicate with corresponding relation $R$. Given any map $w : \cH \to \R_{\geq 0}$, the following are equivalent.
\begin{enumerate}[(1)]
\item $w$ is an $\eps$-sparsifier of $\cH$ with respect to $M$.\label{item:quantum-sparsifier}
\item $w$ is an $\eps$-sparsifier of $\cH$ with respect to $R$.\label{item:classical-sparsifier}
\end{enumerate}
\end{proposition}
\begin{proof}
First we show that \cref{item:quantum-sparsifier} implies \cref{item:classical-sparsifier}. Assume $w$ is an $\eps$-sparsifier of $\cH$ with respect to $M$. Let $\psi$ be a state of the form ${\vert}x\rangle$ for some $x\in \{0,1\}^n$. For every $T \in \cH$, we have that
\[
    \cE_{T,M}(\psi) = \langle \psi, M_{T} \psi\rangle = \langle \psi{\vert}_{T}, M{\vert}_{T} \psi_{T}\rangle = M_{x{\vert}_{T}, x{\vert}_{T}} = \one[x{\vert}_{T} \in R].
\]
Hence,
\[
    \cE_{\cH, w}(\psi) = \sum_{T \in \cH} w(T) \cE_{T,M}(\psi) = \sum_{T \in \cH} w(T)\one[x{\vert}_{T} \in R].
\]
Therefore, the fact that $\cE_{\cH, w}(\psi) \in [1-\eps, 1+\eps] \cdot \cE_{\cH}(\psi)$ (see \cref{eq:quantum-sparsifier}) implies \cref{eq:CSP-sparsification}, as desired.

We now show that \cref{item:classical-sparsifier} implies \cref{item:quantum-sparsifier}. By reversing our steps in the previous case, the fact that $w$ is an $\eps$-sparsifier of $\cH$ with respect to $R$ already implies that $\cE_{\cH, w}(\psi) \in [1-\eps, 1+\eps] \cdot \cE_{\cH}(\psi)$ (see \cref{eq:quantum-sparsifier}) holds for every state $\psi \in \C^{2^n}$ of the form ${\vert}x\rangle$ for some $x \in \zo^n$. Now, for an arbitrary state $\psi \in \C^{2^n}$ and an arbitrary hyperedge $T \in \cH$, observe that since $M$ (and thus $M_T$) is a diagonal matrix, we have that 
\begin{align*}
    \cE_{T,M}(\psi) = \langle \psi, M_{T}\psi\rangle &= \sum_{x \in \{0,1\}^n} {\vert}\psi(x){\vert}^2 (M_T)_{x,x}\\
    &= \sum_{x \in \{0,1\}^n} {\vert}\psi(x){\vert}^2 \one[x{\vert}_{T} \in R]\\
    &= \sum_{x \in \{0,1\}^n} {\vert}\psi(x){\vert}^2 \cE_{T,M}({\vert}x\rangle).
\end{align*}
Therefore, 
\[
\cE_{\cH, w}(\psi) = \sum_{T \in \cH} \cE_{T,M}(\psi) = \sum_{T \in \cH}\sum_{x \in \{0,1\}^n} {\vert}\psi(x){\vert}^2 \cE_{T,M}({\vert}x\rangle) = \sum_{x \in \{0,1\}^n} {\vert}\psi(x){\vert}^2 \cE_{\cH,w}({\vert}x\rangle).
\]
That is, the energy of any $\psi \in \C^{2^n}$ is a positive linear combination of the energies of classical states. Thus, since $w$ is an $\eps$-sparsifier of all classical states, $w$ is an $\eps$-sparsifier of all states, as desired.
\end{proof}

\begin{remark}
One may mildly generalize \cref{prop:classical-equivalence} by defining $M \in \C^{2^r \times 2^r}$ to be classical if $M$ is merely a diagonal matrix (with arbitrary nonnegative weights on the diagonal). In this case, the corresponding version of \cref{prop:classical-equivalence} would say that sparsifying $M$-Hamiltonians is equivalent to sparsifying a \emph{valued} constraint satisfaction problem (VCSP) (see \cite{SchiexFV95}), where different assignments to a classical predicate are given different weights. Some special cases of VCSP sparsification are studied in \cite{BrakensiekGP25}.
\end{remark}

\subsection{Code and XOR Sparsification}

Toward the proof of \cref{thm:sparspaulihamIntro} and its corollaries, we make use of \emph{code sparsification} framework of Khanna, Putterman, and Sudan~\cite{KhannaPS24,KhannaPS25}. By a \emph{code}, we mean an $n$-dimensional linear subspace $C \subseteq \F_2^m$. Given a set of weights $w : [m] \to \R_{\ge 0}$, we say that $w' : [m] \to \R_{\ge 0}$ is an $\eps$-sparsifier of $w$ if for every codeword $c \in C$, we have that
\[
    \sum_{i \in [m]} w'(i) \one[c_i = 1] \in [1 - \eps, 1 + \eps] \cdot \sum_{i \in [m]} w(i) \one[c_i = 1].
\]
\begin{theorem}[\cite{KhannaPS25}]\label{thm:codesparsifier}
For any $n$-dimensional linear code $C \subseteq \F_2^m$ and any set of weights $w : [m] \to \R_{\ge 0}$, there exists an $\eps$-sparsifier $w' : [m] \to \R_{\ge 0}$ with support size at most $\widetilde{O}(\eps^{-2}n)$. Furthermore, $w'$ can be computed efficiently via a randomized algorithm in $\poly(n, m)$ time.
\end{theorem}

A key application of \cref{thm:codesparsifier} in the problem of $\mathbf{XOR}$-sparsification \cite{KhannaPS24, KhannaPS25}. We first define the $\neg\mathbf{XOR}$ predicate:
\begin{definition}
    For any $r\geq 1$, define the predicate $\neg\mathbf{XOR}:\zo^r\to\zo$ as 
    \[\neg\mathbf{XOR}(x_1, \ldots, x_r):= \begin{cases}
        1 & \text{if }\sum_i x_i\text{ is even, i.e. }\sum_i x_i =_{\F_2}0, \\
        0 & \text{otherwise }
    \end{cases}\mper\]
\end{definition}
The $\mathbf{XOR}$-sparsification theorem proven in \cite{KhannaPS25} in fact sparsifies weighted instances too, without any loss in parameters. The exact statement is given below:

\begin{theorem}[Theorem 1.3 in \cite{KhannaPS25}]
\label{thm:negxorsparsifier}
    Fix any $r\geq 1$, and consider a weighted $r$-uniform hypergraph $\cH = \{T_i\}\seq[n]^r$ for $i = 1 ,\dots m$ and weights $\nu:[m] \rightarrow \R_{\geq 0}$. Then for any $\eps\in(0, 1)$ there exists a $\eps$-sparsifier $w:\cH\to\R_{\geq 0}$ with ${\vert}\supp(w){\vert}\leq \widetilde{O}(\eps^{-2}n)$, computable in randomized $\poly(n^r)$ time, such that for any $x\in\zo^n$, 
    \[\sum_i w(i)\cdot\1(x\text{ satisfies }(\neg\mathbf{XOR})\vert_{T_i})\in (1 \pm \eps)\sum_i\nu(i)\cdot\1(x\text{ satisfies }(\neg\mathbf{XOR})\vert_{T_i})\mper\]
\end{theorem}
\begin{remark}
\label{remark:negxorspars}
    A few remarks are due:
    \begin{enumerate}[(1)]
        \item In \cref{thm:negxorsparsifier} the sparsifier size upper bound \textbf{does not depend} on the weights $\nu$.
        \item Note that \cref{thm:negxorsparsifier} is optimal (in sparsifier size) upto $\widetilde{O}_r(1)$ factors since it is easy to see that any sparsifier must have $\Omega_r(n)$ elements.
    \end{enumerate}
\end{remark}

 We also need a mild extension of \cref{thm:negxorsparsifier} in which we require that the sparsifier is \emph{unbiased} in the sense that $\sum_i w(i) = \sum \nu(i)$. We now show how to do this by applying \cref{thm:codesparsifier}.

 \begin{corollary}\label{cor:negxorsparsifier-unbiased}
    Fix any $r\geq 1$, and consider a weighted $r$-uniform hypergraph $\cH = \{T_i\}\seq[n]^r$ for $i = 1 ,\dots m$ and weights $\nu:[m] \rightarrow \R_{\geq 0}$. Then for any $\eps\in(0, 1)$ there exists a $\eps$-sparsifier $w:\cH\to\R_{\geq 0}$ with ${\vert}\supp(w){\vert}\leq \widetilde{O}(\eps^{-2}n)$, computable in randomized $\poly(n^r)$ time, such that for any $x\in\zo^n$, 
    \[\sum_i w(i)\cdot\1(x\text{ satisfies }(\neg\mathbf{XOR})\vert_{T_i})\in (1 \pm \eps)\sum_i\nu(i)\cdot\1(x\text{ satisfies }(\neg\mathbf{XOR})\vert_{T_i})\mper\]
    Furthermore $\sum_i w(i) = \sum_i \nu(i)$.
 \end{corollary}
 \begin{proof}
 We construct a matrix $G \in \F_2^{m \times (n+1)}$ as follows For $i \in [m]$ and $j \in [n]$, we let $G_{i,j} = 1$ if $j \in T_i$ and $0$ otherwise. Furthermore, we let $G_{i,n+1} = 1$ for all $i \in [m]$. Let $C := \{Gx : x \in \F_2^{n+1}\}$ be the linear code generated by $G$. Note that $\dim(C) \le n+1$. We weight the coordinates of $C$ by $\nu_1, \hdots, \nu_m$. By \cref{thm:codesparsifier}, we can construct an $(\eps/3)$-sparsifier $w : [m] \to \R_{\ge 0}$ of size $\widetilde{O}(\eps^{-2}n)$. such that for every $x \in \F_2^{n+1}$ we have that
 \begin{align}
    \sum_{i \in [m]} w(i) \one[(Gx)_i = 1] \in [1 - \eps/3, 1 + \eps/3] \cdot \sum_{i \in [m]} \nu(i) \one[(Gx)_i = 1].\label{eq:3}
 \end{align}
 Returning to our $\neg \mathbf{XOR}$ instance, any assignment $x \in \{0,1\}^n$ can be viewed as a vector $y := (x_1, \hdots, x_n, 1)\in \F_2^{n+1}$. Now for all $i \in [m]$, we have that
 \[
    (Gy)_i = 1 \iff \bigoplus_{j \in T_i} x_j \oplus 1 = 0 \iff x\text{ satisfies } (\neg \mathbf{XOR})\vert_{T_i}.
 \]
 Thus, \cref{eq:3} implies $w$ is an $\eps/3$-sparsifier of our $\neg \mathbf{XOR}$ instance.

 To make this sparsifier unbiased, consider $y = (0, \hdots, 0, 1) \in \F_2^{n+1}$ and note $Gy = (1, \hdots, 1)$. Thus, (\ref{eq:3}) implies that
 \[
    \sum_{i \in [m]} w(i) \in [1 - \eps / 3, 1 + \eps / 3]  \cdot \sum_{i \in [m]} \nu(i).
 \]
 Now for all $i \in [m]$ define $w'(i) := w_i \cdot \frac{\sum_{i \in [m]} \nu(i)}{\sum_{i \in [m]} w(i)} \in [1-\eps/3, 1+\eps/3]$. We have that $\sum_{i=1}^m w'(i) = \sum_{i=1}^m \nu(i)$, and for all $y \in \F_2^{n+1}$, we have that
 \begin{align*}
\sum_{i \in [m]} w'(i) \one[(Gx)_i = 1] &= \frac{\sum_{i \in [m]} \nu_i}{\sum_{i \in [m]} w(i)} \cdot \sum_{i \in [m]} w_i \one[(Gx)_i = 1]\\
&\in [1-\eps/3, 1+\eps/3] \cdot [1 - \eps/3, 1 + \eps/3] \cdot \sum_{i \in [m]} \nu_i \one[(Gx)_i = 1]\\
&\subseteq [1-\eps, 1+\eps]\cdot \sum_{i \in [m]} \nu_i \one[(Gx)_i = 1],
 \end{align*}
 as desired. Therefore, the weights $w'(1), \hdots, w'(m)$ form our unbiased $\eps$-sparsifier of our $\neg \mathbf{XOR}$ instance.
 \end{proof}

\subsection{Concentration Inequalities.}

\begin{fact}[Matrix Bernstein Inequality (\cite{Tro15})]
\label{fact:matrix-bernstein}
Let $Y_1, \ldots, Y_m\in\C^{N\times N}$ be independently random mean-zero Hermitian matrices such that $\max_{i\in[m]}{\Vert}Y_i{\Vert}_{\op}\leq R$ with probability $1$ for some $R > 0$. Write $\sigma:= \left\Vert \E\sum_{i\in[m]}Y_i^2\right\Vert_{\op}^{1/2}$, $Y = \sum_{i\in[m]}Y_i$. Then ${\Vert}Y{\Vert}_{\op}\leq O(R\log N + \sigma\sqrt{\log N})$ with probability $\geq 1 - 1/N^{10}$.
\end{fact}

\begin{fact}[Matrix Chernoff]
\label{fact:matrix-bernstein-restated}
Let $Y_1, \ldots, Y_m\in\C^{N\times N}$ be independent, random Hermitian matrices. Let $Y = \sum_{i = 1}^m Y_i$ and let $Z = \E[Y]$. Suppose that $Y_i \preceq R \cdot Z$ with probability $1$, where $R$ is an arbitrary scalar, for every $1 \leq i \leq m$. Then, for all $\eps \in [0,1]$,
\[
\Pr\left [\sum_{i = 1}^m Y_i \preceq (1 - \eps) Z \right] \leq N \cdot \exp(-\eps^2 / 2R)
\]
and 
\[
\Pr\left [\sum_{i = 1}^m Y_i \succeq (1 + \eps) Z \right] \leq N \cdot \exp(-\eps^2 / 3R).
\]
\end{fact}

\subsection{Continuous Distributions and Boundedness Assumptions}

We will often make use of the following fact about evaluating non-zero polynomials:

\begin{claim}[Corollary 10 in \cite{GunningR65}]
\label{clm:zeroEvaluation}
    Let $f:\C^k\to\C$ be a multivariate polynomial on $k$ variables which is \textbf{not} identically zero. Let $Z(f):= \{(x_1, \ldots, x_k)\in\C^k: f(x_1, \ldots, x_k) = 0\}$. Then $Z(f)$ has zero Lebesgue measure in $\C^k$. In particular, for any absolutely continuous distribution $\mu$ on $\C^k$, 
    \[\Pr_{(x_1, \ldots, x_k)\sim\mu}(f(x_1, \ldots, x_k) = 0) = 0\mper\]
\end{claim}

We also need to make some analytical assumptions. Fix any $r\geq 2$. Throughout the paper we will be dealing with arbitrarily large families of PSD matrices $\{M_i\}_{i\in[m]}\subset\C^{2^r\times 2^r}$, and we will need them to satisfy some ``boundedness'' conditions, which we state below. Before we state our conditions, we denote, for any non-zero Hermitian PSD matrix $A$, by $\lambda_{\min,\neq 0}(A)$ the minimum non-zero eigenvalue of $A$.

Our first boundedness criterion restricts the maximum eigenvalue of our collection of PSD matrices.

\begin{definition}[$C$-boundedness]
\label{def:OnlyCbounded}
    Let $\cM:= \{M_i\}_{i\in[m]}\subset\C^{2^r\times 2^r}$ be a collection of non-zero PSD matrices, where $r\geq 2$. For some $C \in(0, \infty)$ we say $\cM$ is $C$-bounded if $\sup_{M\in\cM}{\Vert}M{\Vert}_{\operatorname{op}} = \sup_{M\in\cM}\lambda_{\max}(M)\leq C$.
\end{definition}

Again, we can see that there are absolutely continuous distributions of PSD matrices which are $C$ bounded, for $C \leq O_r(1)$.

\begin{example}\label{example:BasicCollectionOfRandomMatrices}
    Consider the collection of matrices $\cM$ which contains all matrices $M$ which can be written as \[
    M = \sum_{i =1}^{2^r} v_i v_i^{\dagger},
    \]
    where each $v_i \in \C^{2^r}$ is a vector with norm $\leq 1$.
\end{example}

Our second criterion also restricts some minimum nonzero eigenvalues.

\begin{definition}[$(C, \ell_0)$-boundedness]
\label{def:Cbounded}
    Let $\cM:= \{M_i\}_{i\in[m]}\subset\C^{2^r\times 2^r}$ be a family of non-zero PSD matrices, where $r\geq 2$. For some $C \in(0, \infty), \ell_0\in\N$ we say $\cM$ is $(C, \ell_0)$-bounded if:
    \begin{enumerate}[(1)]
        \item $\sup_{M\in\cM}{\Vert}M{\Vert}_{\operatorname{op}} = \sup_{M\in\cM}\lambda_{\max}(M)\leq C$, and, 
        \item Fix any non-empty $\cM'\subset\cM$ with ${\vert}\cM'{\vert}:= \ell\leq \ell_0$. Let $\cH$ be any $r$-uniform hypergraph with $\ell$ hyperedges, and write $V(\cH):= \bigcup_{T\in\cH}T$. Let $\pi:\cM'\to\cH$ be any bijection. Then 
        \[\lambda_{\min,\neq 0}\left(\sum_{M\in\cM'}M\vert_{\pi(M)}\otimes\Id_{V(\cH)\setminus\pi(M)}\right)\geq \frac{1}{C}\mper\]
    \end{enumerate}
\end{definition}
\begin{example}
\label{example:singleMbounded}
To see that this definition makes sense, consider the following example of a bounded matrix family: Fix any non-zero PSD $M\in\C^{2^r\times 2^r}$, and let $M_i = M$ for all $i\in\N$. Also fix any $\ell_0:= \ell_0(r)$. Then there exists a constant $C = C(M, \ell_0)$ such $\cM:= \{M_i\}_{i\in\N}$ is $(C, \ell_0)$-bounded: Indeed, there are finitely many $r$-uniform hypergraph configurations with $\leq\ell_0$ hyperedges, and since a finite set is compact, we can choose $C > 0$ so that $\cM$ is $(C, \ell_0)$-bounded. Note that $C$ is a function of $r, \ell_0, M$. 
\end{example}

\section{Sparsifying Pauli Hamiltonians}\label{sec:pauli}

\subsection{General Framework}

We begin with some definitions and notation:
\begin{definition}[Pauli Matrices]
    Define the Pauli matrices \[X = \begin{bmatrix}
    0 & 1\\
    1 & 0
\end{bmatrix}, Y = \begin{bmatrix}
    0 & -i\\
    i & 0
\end{bmatrix}, Z = \begin{bmatrix}
    1 & 0\\
    0 & -1
\end{bmatrix}\in\C^{2\times 2}.\]
\end{definition}
 The Pauli matrices are unitary, and their eigenvalues are $\pm 1$. For $P\in\{X, Y, Z\}$, fix $\{u^{(-1)}_P, u^{(1)}_P\}\seq\C^2$ to be an orthonormal eigenbasis for $P$, with $Pu^{(\lambda)}_P = \lambda u^{(\lambda)}_P$ for all $\lambda\in\{-1, 1\}$.

We now take tensor products of Pauli operators to build the so-called \emph{Pauli strings}:
\begin{definition}[Pauli Strings]
    A unitary operator $\operatorname{U}$ acting on $\C^{\zo^{[n]}}$ of the form $P_1\otimes\cdots\otimes P_n$, where $P_i\in\{\Id, X, Y, Z\}$ acts on the $i^{\mathrm{th}}$ qubit, is called a \emph{Pauli string}. Furthermore, if $\#\{i: P_i\neq\Id\}\leq r$, then we call $\operatorname{U}$ a $r$-local operator.
\end{definition}
Note that if $\operatorname{U}$ is a Pauli string, then all eigenvalues of $\operatorname{U}$ are $\pm 1$. Furthermore, if $\operatorname{U}$ is a $r$-local Pauli string, then $\operatorname{U} = W\vert_T\otimes\Id_{[n]\setminus T}$ for some $T\in[n]^r$. To fit Pauli strings into our framework of PSD matrices, we simply add the identity matrix to the entire system. This makes all the eigenvalues non-negative, and the resulting system PSD. Thus, we give the following definition:
\begin{definition}[Pauli PSD Matrices]
Let $P_1, \ldots, P_r\in\{X, Y, Z, -X, -Y, -Z\}$ be Pauli matrices (and their negations). We call the matrix $M:= P_1\otimes\cdots\otimes P_r + \Id\in\C^{2^r\times 2^r}$ a \emph{Pauli PSD matrix}.  Equivalently, one could constrain $P_1, \ldots, P_r\in\{X, Y, Z\}$ and specify that $\Id \pm P_1\otimes\cdots\otimes P_r$ is a Pauli PSD matrix.
\end{definition}

We also define \emph{Pauli Hamiltonians}:
\begin{definition}[Pauli Hamiltonians]
    If we have a Hamiltonian $\cH = \{(T_i, M_i)\}_{i\in [m]}$ where each $M_i$ is a Pauli PSD matrix, then we call $\cH$ a \emph{Pauli Hamiltonian}. In general, we can also have a \emph{weighted Pauli Hamiltonian} $\cH = (\{(T_i, M_i)\}, \nu)$, where $\nu$ is a function mapping each index $i$ to a weight in $\R_{\geq 0}$. Note that a weighted Pauli Hamiltonian is naturally associated to the PSD matrix $\sum_i\nu(i)\cdot(M_i)\vert_{T_i}\otimes\Id_{\bar{T_i}}$. 
\end{definition}

\begin{theorem}[Sparsifying Pauli Hamiltonians]
\label{thm:sparspauliham}
    Let $\cH$ be a $r$-local weighted Pauli Hamiltonian with $m$ terms on $n$ qubits. Then there exists a $\eps$-sparsifier $w:[m]\to\R_{\geq 0}$ of size $\leq \widetilde{O}_r(\eps^{-2}n)$, i.e. for any $\psi\in\C^{\zo^n}$, 
    \[\sum_iw(i)\cE_{T_i, M_i}(\psi)\in(1\pm\eps)\sum_i\nu(i)\cE_{T_i, M_i}(\psi)\mper\]
    
    Furthermore, $\sum_i w(i) = \sum_i \nu(i)$ and there exists a randomized classical $\poly(m, n, O_r(1))$-time algorithm to compute such a sparsifier.
\end{theorem}
\begin{proof}
    Firstly, note that there are $3^r$ choices for a Pauli PSD matrix in $\C^{2^r\times 2^r}$. Thus we can split the hypergraph underlying $\cH$ into $k\leq 3^r \leq O_r(1)$ hypergraphs $\cH_1, \ldots, \cH_k$ corresponding to the different choices for a Pauli PSD matrix. Note that we can take the union of $\eps$-sparsifiers of each of these sub-hamiltonians to get an $\eps$-sparsifier for $\cH$, and thus $\SPR(\cH, \eps)\leq\sum_j\SPR(\cH_j, \eps)$. 

    Thus fix a Pauli PSD matrix $W = \Id + P_1\otimes\cdots\otimes P_r$ and consider the $W$-Hamiltonian $\cH_W\seq\cH$, i.e. $\cH_W$ is the collection of all hyperedges in $\cH$ which are being acted upon by $W$. Note that it suffices to prove the claim for $\cH_W$, i.e. show that $\SPR(\cH_W, \eps)\leq \widetilde{O}_r(\eps^{-2}n)$, and such a sparsifier can be found efficiently.

    By \cref{prop:r-partite-WLOG}, we can compute (in $\poly(n^r)$ time) a partition $\cH_W = \bigsqcup_{a\in\mathcal{I}}\cH_{W, a}$ such that ${\vert}\mathcal{I}{\vert}\leq O_r(\log n)$, and $\cH_{W, a}$ is $r$-partite. If we show that $\SPR(\cH_{W, a}, \eps)\leq\widetilde{O}_r(\eps^{-2}n)$, then we obtain $\SPR(\cH_W, \eps)\leq {\vert}\mathcal{I}{\vert}\cdot \widetilde{O}_r(\eps^{-2}n) \leq \widetilde{O}_r(\eps^{-2}n)$, and we're done.

    Thus for the remaining proof assume we have a $r$-partite $W$-Hamiltonian $\cH$ on $n$ qubits,\footnote{if some $r$-partite subhypergraphs have fewer than $n$ qubits that are acted non-trivally upon, we can simply shrink our space of qubits} where $W = P_1\otimes\cdots\otimes P_r$. Let $V_1, \ldots, V_r$ be the parts of $\cH$, i.e. $\cH\seq\bigtimes_{k = 1}^rV_k$. Thus for any $T\in\cH$, we can write $T = (t_1, \ldots, t_r)$, where $t_k\in V_k$ for all $k\in[r]$. 

    Now, consider the orthonormal basis\footnote{note that if $\mathcal{B}_1 = \{u_1, \ldots, u_k\}, \mathcal{B}_2 = \{v_1, \ldots, v_\ell\}$ are two orthonormal bases of vector spaces $V_1, V_2$, then $\mathcal{B}_1\otimes \mathcal{B}_2 = \{u_i\otimes v_j\}_{i\in[k], j\in[\ell]}$ is an orthonormal basis for $V_1\otimes V_2$} $\mathcal{B}_W$ of $\C^{\zo^{[n]}}$ given by 
    \[\mathcal{B}_W:= \left\{\bigotimes_{k = 1}^r\bigotimes_{v\in V_k}u^{(b_v)}_{P_k}\right\}_{b\in\{-1, 1\}^{[n]}}\mper\]

    We claim that $\mathcal{B}_W$ diagonalizes $W_T = W\vert_T\otimes\Id_{\bar{T}}$ for any $T = (t_1, \ldots, t_r)\in\cH$. Firstly, note that it suffices to show that $\mathcal{B}_W$ diagonalizes $\operatorname{U}_T:= (\bigotimes_{k = 1}^rP_k)\otimes\Id_{\bar{T}}$, where $P_k$ acts on $t_k$, since $W_T = \operatorname{U}_T + \Id$. But note that since $P_k$ acts on $t_k$, and since every vector in $\mathcal{B}_W$ has either $u^{(-1)}_{P_k}$ or $u^{(1)}_{P_k}$ in the $t_k^{\mathrm{th}}$ position in the tensor product, we obtain that every vector in $\mathcal{B}_W$ is an eigenvector of $W_T$ for any $T\in\cH$. Furthermore, for any $b\in\{\pm 1\}^{[n]}$, we have 
    \[\operatorname{U}_T\cdot\bigotimes_{k = 1}^r\bigotimes_{v\in V_k}u^{(b_v)}_{P_k} = \left(\prod_{k = 1}^rb_{t_k}\right)\cdot\bigotimes_{k = 1}^r\bigotimes_{v\in V_k}u^{(b_v)}_{P_k}\implies W_T\cdot\bigotimes_{k = 1}^r\bigotimes_{v\in V_k}u^{(b_v)}_{P_k} = \left(1 + \prod_{k = 1}^rb_{t_k}\right)\cdot\bigotimes_{k = 1}^r\bigotimes_{v\in V_k}u^{(b_v)}_{P_k}\mper\]
    Writing $b_v = (-1)^{\tau_v}$ (for $\tau_v\in\F_2$), we obtain 
    \[W_T\cdot\bigotimes_{k = 1}^r\bigotimes_{v\in V_k}u^{(b_v)}_{P_k} = 2\cdot\left(1 - \bigoplus_{k = 1}^r\tau_{t_k}\right)\cdot\bigotimes_{k = 1}^r\bigotimes_{v\in V_k}u^{(b_v)}_{P_k}\mcom\]
    where $\oplus$ is taken over the field $\F_2$. Consequently, in the basis $\mathcal{B}_W$, $W_T/2$ is a diagonal matrix whose diagonal describes the truth-table of $\neg\mathbf{XOR}(\tau_{t_1}, \ldots, \tau_{t_r})$! Since quadratic forms are invariant under an orthonormal change of basis (from $(\braket{0}  + \braket{1})^{\otimes n}$ to $\mathcal{B}_W$), sparsifying the (weighted) $W$-Hamiltonian $\cH$ reduces to sparsifying the (weighted) classical CSP $\neg\mathbf{XOR}$ on the same hypergraph by \cref{prop:classical-equivalence}. At this point, we can conclude by \cref{cor:negxorsparsifier-unbiased}.
\end{proof}
Note that by \cref{remark:negxorspars}, the bound on the sparsifier size in \cref{thm:sparspauliham} does \textbf{not} depend on the weights in $\cH$.

\subsection{Sparsifying Pauli String Decompositions}

We can also prove a somewhat stronger version of \cref{thm:sparspauliham}. To motivate this extension, let $M \in \C^{2^r\times 2^r}$ be a PSD predicate matrix which can be be written as a sum $M := \eta(1) M_1 + \cdots + \eta(k) M_k$, where each of $M_1, \hdots, M_k$ is a Pauli PSD matrix. Now consider an $M$-Hamiltonian $\cH = \{(T_i)\}_{i = 1}^m$, with weights $\nu: [m] \rightarrow \R_{\geq 0}$. Given our Pauli decomposition of $M$, we can decompose $\cH$ as $\cH_1 \cup \cdots \cup \cH_k$, where $\cH_j := (\{(T_i, M_j\}, \eta(j)\nu)$. Applying \cref{thm:sparspauliham}, we can compute sparsifiers $w_1, \hdots, w_k$ of $\cH_1, \hdots, \cH_k$ which each has size $\widetilde{O}_r(\eps^{-2}n)$. Since $k\leq O_r(1)$, the aggregate size of these sparsifiers is also $\widetilde{O}_r(\eps^{-2}n)$. \emph{A priori}, the resulting sparsifier is not a true $\eps$-sparsifier of $\cH$ as each term $M_{T_i}$ is decomposed into $\sum_{j=1}^k w_j(i) (M_j)_{T_i}$, which may not be a scalar multiple of $M_{T_i}$. However, careful inspection of the proof of \cref{thm:sparspauliham} shows that these $k$ different $\cH_j$'s reduce to sparsifying the \emph{same} classical $\neg \mathbf{XOR}$ instances! In other words, the $\eps$-sparsifier weights for any one of these $\cH_i$'s works as an $\eps$-sparsifier for any of the others and thus of $\cH$ itself! We formalize this observation in the following corollary.

\begin{corollary}\label{cor:sparspauliham-decomp}
Let $M \in \C^{2^r\times 2^r}$ be a PSD predicate matrix which can be be written as a sum $M := \eta(1) M_1 + \cdots + \eta(k) M_k$, where each of $M_1, \hdots, M_k$ is a Pauli PSD matrix. Given an $M$-Hamiltonian $\cH = \{(T_i\}_{i \in [m]}$ with weights $\nu: [m] \rightarrow \R_{\geq 0}$, one can in randomized $\poly(n^r)$-time construct an $\eps$-sparsifier $\cH' = \{(T_i)\}_{i \in [m]}$ with weights $w: [m] \rightarrow \R_{\geq 0}$ with support size at most $\widetilde{O}_r(\eps^{-2}n)$. And furthermore,  $\sum_i w(i) = \sum_i \nu(i)$.
\end{corollary}

As many details of the proof of \cref{cor:sparspauliham-decomp} are identical to \cref{thm:sparspauliham}, we only describe the novel aspects of the proof.

\begin{proof}
Using \cref{prop:r-partite-WLOG}, we may assume without loss of generality that our instance $\cH$ is $r$-bipartite. That is, the $n$-qubits have a partition $V_1, \hdots, V_r$ such that $T_i \subseteq V_1 \times \cdots \times V_r$ for all $i \in [m]$.

Since $M = \eta(1) M_1 + \cdots + \eta(k) M_k$, we may construct weighted Pauli Hamiltonians $\cH_j := (\{(T_i, M_j\}_{i \in [m]}, \eta(j)\nu)$ such that $\cE_{\cH} = \sum_{j=1}^k \cE_{\cH_j}$. By the proof of \cref{thm:sparspauliham}, since each $\cH_j$ is bipartite, an $\eps$-sparsifier of $\cH_j$ is equivalent to an $\eps$-sparsifier of the classical CSP $\cH'_j := \{T_i\}_{i \in [m]}$ of $\neg \mathbf{XOR}$ with weights $\eta(j)\nu : [m] \to \R_{\ge 0}$. Let $w^{(1)} : [m] \to \R_{\ge 0}$ be an $\eps$-sparsifier of $\cH'_1$ (which is also an $\eps$-sparsifier of $\cH_1$). Now observe for any $j \in [k]$, we have that $\cH'_j$ is $\cH'_1$ scaled by $\eta(j)/\eta(1)$. Thus, $\frac{\eta(j)}{\eta(1)} \cdot w^{(1)}$ is an $\eps$-sparsifier of $\cH'_j$ and $\cH_j$. Now define $w := \frac{w^{(1)}}{\eta(1)}$. We claim that $w$ is an $\eps$-sparsifier of $\cH$. To see why, for any state $\psi \in \C^{\zo^n}$ we have that
\begin{align*}
    \sum_{i=1}^m w(i) \cE_{T_i,M}(\psi) &= \sum_{i=1}^m\sum_{j=1}^k w(i) \eta(j)\cE_{T_i, M_j}(\psi)\\
    &= \sum_{j=1}^k \sum_{i=1}^m \frac{\eta(j) w^{(1)}(i)}{\eta(i)} \cE_{T_i, M_j}(\psi)\\
    &\in [1-\eps, 1+\eps] \cdot \sum_{j=1}^k \sum_{i=1}^m \eta(j) \nu(i) \cE_{T_i, M_j}(\psi)\\
    &\subseteq [1-\eps, 1+\eps] \cdot \sum_{i=1}^m \nu(i) \cE_{T_i, M}(\psi),
\end{align*}
as desired. Finally, observe that by \cref{thm:sparspauliham}, we have that $\sum_{i=1}^m w^{(1)}(i) = \sum_{i=1}^m \eta(1) \nu(i).$ Thus, $\sum_{i=1}^m w(i) = \sum_{i=1}^m \frac{w^{(1)}(i)}{\eta(1)} = \sum_{i=1}^m \nu(i)$, as desired.
\end{proof}

\section{Sparsifying Random Hamiltonians}\label{sec:sparse-genericLouie}

In this section, we will study the sparsifiability of Hamiltonians using \emph{random} predicate matrices $M \in \C^{2^r \times 2^r}$. 

\begin{remark}[Distribution assumptions]\label{rmk:assumptionsDist}
    Throughout this section, we will consider matrices $M$ which have rank $R \geq 2^{r-1}+1$, and whose span can be written as $\mathrm{span}(M) = \mathrm{span}(v_1, \dots v_{R})$, where the vectors $v_i \in \C^{2^r}$ are chosen at \emph{random} from a distribution with absolutely continuous joint support, and uniform boundedness (see \cref{def:OnlyCbounded}, \cref{clm:zeroEvaluation}). Furthermore, we suppose that this distribution is fixed \emph{only as a function of} $r$. We then let the number $n$ of qubits grow \emph{independently} of this distribution over predicate matrices.
\end{remark}

With this notation in place, our main theorem is the following:

\begin{theorem}\label{thm:EfficientRandomSparsifierLouie}
    Let $M \in \C^{2^r \times 2^r}$ be a random predicate matrix of rank $R \geq 2^{r-1} + 1$ drawn from a distribution in accordance with \cref{rmk:assumptionsDist} and let $\mathcal{H} = \{(T_i, M):{i \in [m]}\}$ be an arbitrary Hamiltonian over $n$ variables and $m$ terms. Then, with probabilty $1-o_r(1)$ over the choice of $M$, for any $\eps > 0$, there is a randomized algorithm running in time $\mathrm{poly}(n,m)$ which computes a weight function $w: \mathcal{H} \rightarrow \R_{\geq 0}$ such that $w$ is an $\eps$-sparsifier of $\mathcal{H}$ and the size of $w$ is bounded by $\widetilde{O}_r(\eps^{-2}n)$ with probability $1 - 2^{-\Omega(n)}$.
\end{theorem}

\begin{remark}
Note that we use $O_r(\cdot)$ to denote that the leading constant of the sparsifier size has a constant which depends on the arity $r$ (in particular, on the exact non-zero eigenvalues of the predicate matrix $M$). Importantly, the arity $r$ and the predicate matrix $M$ are fixed \emph{independently} of the number of qubits $n$, and so as the number of qubits $n$ increases, the size of the sparsifier scales only linearly in $n$ (and the error parameter $\eps$).
\end{remark}

We can also prove a version of \cref{thm:EfficientRandomSparsifierLouie} which holds when the individual predicate matrices are \emph{different}. Formally, we show:

\begin{theorem}\label{thm:EfficientRandomSparsifierDifferentMLouie}
    Let $\mathcal{H} = \{(T_i, M_i):{i \in [m]}\}$ be an arbitrary Hamiltonian over $n$ variables and $m$ terms, where each predicate matrix $M_i \in \C^{2^r \times 2^r}$ is random of rank $R \geq 2^{r-1}+1$, independently drawn from a distribution in accordance with \cref{rmk:assumptionsDist}.
    
    Then, for any $\eps > 0$, there is a randomized algorithm running in time $\mathrm{poly}(n,m)$ which computes a weight function $w: \mathcal{H} \rightarrow \R_{\geq 0}$ such that $w$ is an $\eps$-sparsifier of $\mathcal{H}$ and the size of $w$ is bounded by $\widetilde{O}_r(\eps^{-2}n)$  with probability $1 - 2^{-\Omega(n)}$ over the randomness of the $M_i$'s and the algorithm.
\end{theorem}

\begin{remark}
    Going forward, when we describe matrices as being ``random'', we mean exactly that they are drawn from a distribution as described in \cref{rmk:assumptionsDist}. For brevity of the statements, we omit this qualifier.
\end{remark}

Additionally, going forward, we will use \cref{prop:r-partite-WLOG} to assume that the underlying Hamiltonian is $r$-partite WLOG, which leads to extra logarithmic factors in the sparsifier size.

\subsection{Analyzing Intersecting Terms}\label{sec:intersectingTermGenericLouie}

Towards proving \cref{thm:EfficientRandomSparsifierLouie}, our first key insight is to look at the structure of the underlying terms in the Hamiltonian. In particular, we focus our attention on any terms $T_i \subseteq [n]$ and $T_j \subseteq [n]$ such that $T_i \cap T_j \neq \emptyset$. For such terms, we have the following lemma:

\begin{lemma}\label{lem:nullSpaceLouie}
    Let $M \in \C^{2^r \times 2^r}$ be a random predicate matrix of rank $R \geq 2^{r-1} + 1$, and let $T, T' \subseteq [n]$, ${\vert}T{\vert} = r$ denote two local constraints. Suppose that $T \cap T' \neq \emptyset$. Then, with probability $1$ over $M$, for any state $\psi \in \C^{2^n}$ such that $\cE_{T, M}(\psi) = \cE_{T', M}(\psi) = 0$, it must be the case that for any $z \in \zo^n$ that $\psi_z = 0$.
\end{lemma}

\begin{remark}\label{rmk:genericLouie}
    Note that when a predicate matrix $M$ satisfies the above condition, we say that $M$ is \emph{generic}.
\end{remark}

\begin{proof}
We suppose WLOG that $R = 2^{r-1}+1$, as we shall see below that having higher rank introduces \emph{more constraints} on the kernel, which only makes the proof simpler.

    Indeed, let us consider two terms $T$ and $T'$ such that $T \cap T' \neq \emptyset$. We let $U = T \cup T' \subseteq [n]$ denote the union of these two terms, and observe that ${\vert}U{\vert} \leq 2r -1$. For now, we focus on the case where ${\vert}U{\vert} = 2r-1$ (so $T$ and $T'$ intersect in exactly one qubit).

    Now, let $\psi \in \C^{2^n}$ be any state such that $\cE_{T, M}(\psi) = \cE_{T', M}(\psi) = 0$. In particular, this means that for any choice of $y \in \zo^{[n] - U}$, it must be that 
    \[
    \langle \psi_y, M_{T} \otimes \Id_{U - T} \psi_y \rangle = \langle \psi_y, M_{T'} \otimes \Id_{U - T'} \psi_y \rangle = 0,
    \]
    where $\psi_y \in \C^{2^{{\vert}U{\vert}}}$ is the resulting ``sub-vector'' of $\psi$ which is achieved by keeping only those indices of $\psi$ which equal $y$ in the positions $[n] - U$ (see \cref{def:restrictedState}). We will show that, in fact, for every $y' \in \zo^{U}$ that $\psi_{y \circ y'} = 0$, thereby showing that every entry of $\psi$ is $0$. 

    To see this, recall that $\mathrm{span}(M) = \mathrm{span}(v_1, \dots v_{R})$, where each $v_R \in \C^{2^r}$ is chosen at \emph{random}. Thus, a single constraint of the form
    \[
    \langle \psi_y, M_{T} \otimes \Id_{U - T} \psi_y \rangle = 0
    \]
    is imposing several constraints; indeed, for each choice of $\widehat{y} \in \zo^{U - T}$, it must be the case that $\langle \psi_{y \circ \widehat{y}}, M\psi_{y \circ \widehat{y}} \rangle = 0$, which in turn means that $\langle \psi_{y \circ \widehat{y}}, v_j \rangle = 0$, for every $j \in [R]$. We denote this set of constraints by $A_T \in \C^{(2^{{\vert}U{\vert} - T} \cdot R) \times 2^{{\vert}U{\vert}}}$. The fact that all these linear constraints must equal $0$ implies that $A_T \psi_y = 0$.

    Likewise, we derive analogous constraints involving $T'$, and denote their constraint matrix by $A_{T'} \in\C^{(2^{{\vert}U{\vert} - T} \cdot R) \times 2^{{\vert}U{\vert}}}$. Thus, our set of constraints is given by $A = \begin{bmatrix}
        A_T \\
        A_{T'}
    \end{bmatrix} \in \C^{2 \cdot(2^{{\vert}U{\vert} - T} \cdot R) \times 2^{{\vert}U{\vert}}}$
    
    Recall that $\psi_y \in \C^{2^{{\vert}U{\vert}}}$, while the number of linear equations that we have is \[
    2 \cdot R \cdot 2^{{\vert}U{\vert}-r} \geq (2^{r-1}+1) \cdot 2^{{\vert}U{\vert} - r + 1} > 2^{{\vert}U{\vert}}.
    \]
    Our goal now will be to show that, with high probability over the randomness of the vectors $v_1, \dots v_R$, that this collection of constraints $A$ has rank $2^{{\vert}U{\vert}}$. To do this, we first show that \emph{there exists} a choice of $v_1, \dots v_R$ such that $A$ has rank $2^{{\vert}U{\vert}}$:

    \begin{claim}\label{clm:existVectorsLouie}
    Let $T, T'$ be such that ${\vert}T \cap T'{\vert} = 1$. Then, letting $A \in \C^{2 \cdot(2^{{\vert}U{\vert} - T} \cdot R) \times 2^{{\vert}U{\vert}}}$ denote the constraint matrix above, there exists a choice of $v_1, \dots v_R$ such that $A$ has rank $2^{{\vert}U{\vert}}$.
    \end{claim}

    \begin{proof}[Proof of \cref{clm:existVectorsLouie}.]
    For ease of notation, we assume WLOG that $T = \{1, 2, 3, \dots r\}$ and that $T' = \{1, r+1, r+2 , \dots 2r-1\}$. To define our vectors $v_i: i \in [R]$, we index the entries of $v_i$ as a tuple $(a, u) \in \zo^r$ where $a \in \zo$ and $u \in \zo^{r-1}$, using $(v_i)_{(a, \cdot)} \in \C^{2^{r-1}}$ to refer to the vector which is achieved by considering only indices which start with symbol $a$. Likewise, for a value of $u \in \zo^{r-1}$ we let $e_u \in \C^{2^{r-1}}$ denote the standard basis vector which has a $1$ only in coordinate $u$. Equivalently, for an integer $\ell \in 2^{r-1},$ we use $e_{\ell}$ to refer to $e_{u}$ where $u$ is the binary representation of $\ell$.
    
    Under this notation we have:
    \begin{enumerate}[(1)]
        \item $(v_1)_{(0, \cdot)} = e_{0}$ and $(v_1)_{(1, \cdot)} = 0$.
        \item For $i \in \{2, 3, \dots, R-1\}$, $(v_i)_{(0, \cdot)} = e_{i-1}$ and $(v_i)_{(1, \cdot)} = e_{i-2}$.
        \item $(v_R)_{(0, \cdot)} = 0$ and $(v_R)_{1, \cdot} = e_{R-2}$
    \end{enumerate}

    To show that $A$ has full rank, we will show that $\ker(A) = 0$. So, for any vector $\psi \in \C^{2^{{\vert}U{\vert}}}$, we construct two matrices $M_0, M_1 \in \C^{2^{r-1} \times 2^{r-1}}$. The entries of $M_0, M_1$ are indexed by bit strings of length $r-1$, where $M_0(a, b) = \psi_{0 \circ a \circ b}$ and $M_1({a, b}) = \psi_{1 \circ a \circ b}$. For the vectors $v_i: i \in [R]$, and the term $T$, the requirement that for each choice of $\widehat{y} \in \zo^{U - T}$ $\langle \psi_{\widehat{y}}, v_i\rangle = 0$ is equivalent to enforcing that for each column $b \in \zo^{r-1}$:
    \begin{enumerate}[(1)]
        \item $M_0(0, b) = 0$. 
        \item $M_0(\ell, b) + M_1(\ell-1, b) = 0$ for $\ell \in \{1, 2, \dots R - 2\}$. Equivalently, $M_0(\ell, b) = - M_1(\ell-1, b)$.
        \item $M_1(R-2, b) = 0$.
    \end{enumerate}

    Likewise, the constraints from the term $T'$ imply that, for every row $a \in \zo^{r-1}$ that:
\begin{enumerate}[(1)]
        \item $M_0(a, 0) = 0$. 
        \item $M_0(a, \ell) + M_1(a, \ell-1) = 0$ for $\ell \in \{1, 2, \dots R - 2\}$. Equivalently, $M_0(a, \ell) = - M_1(a, \ell-1)$.
        \item $M_1(a, R-2) = 0$.
    \end{enumerate}

Importantly, together the above relations imply that, for any $a \in \{0, 1, \dots 2^{r-1}-1 \}$ and $b \in \{0, 1, \dots 2^{r-1}-1\}$ that $M_1(a, b) = -M_0(a+1, b)$ and that $M_1(a, b) = -M_0(a, b+1)$. Together, this means that $M_0(a, b+1) = M_0(a+1, b)$, i.e., that the diagonals of $M_0$ must all be equal. 

Finally, we then see that the conditions above imply that, in $M_0$, all entries of the form $M(a, 0) = M(0, b) = 0$, and likewise in $M_1$ that $M_1(a, R-2) = M_1(R-2, b) = 0$ (the right and bottom sides of the matrix $M_1$). This, together with the propagation condition, then imply that $M_0 = M_1 = 0$.

Because the only possible choice of $\psi$ such that $M_0, M_1$ satisfy the conditions imposed by the constraints $\langle \psi_{\widehat{y}}, v_i\rangle = 0$ is $\psi = 0^{2^{{\vert}U{\vert}}}$, it must be the case that the constraint matrix $A$ is rank $2^{{\vert}U{\vert}}$ (otherwise it would have a non-trivial nullspace). This concludes the claim. 
    \end{proof}

From \cref{clm:existVectorsLouie}, we know that there exists a collection of vectors $v_1, \dots v_R$ such that the constraint matrix $A$ has rank $2^{{\vert}U{\vert}}$. In particular, for this choice of vectors $v_1, \dots v_R$, there must be some $2^{{\vert}U{\vert}} \times 2^{{\vert}U{\vert}}$ minor of $A$ which has non-zero determinant. We denote this minor by $A_{L_1 \times \left [2^{{\vert}U{\vert}} \right]}$, where $L_1 \subseteq [2\cdot R \cdot 2^{{\vert}U{\vert}-r}]$, ${\vert}L_1{\vert} = 2^{{\vert}U{\vert}}$.

Now, we can observe that $\det(A_{L_1 \times \left [2^{{\vert}U{\vert}} \right]})$ is a polynomial in the entries of the vectors $v_1, \dots v_R$ (since each entry in $A$ is simply an entry from one of these vectors). The above shows that, as a polynomial, 
\[
\det(A_{L_1 \times \left [2^{{\vert}U{\vert}} \right]}) \neq 0.
\]
Now, we are free to invoke \cref{clm:zeroEvaluation}, as each entry of the vectors $v_i$ is sampled uniformly at random from a continuous distribution. This implies that with probability $1$ over the vectors $v_1, \dots v_R$ that $\det(A_{L_1 \times \left [2^{{\vert}U{\vert}} \right]}) \neq 0$, and so $\rank(A) = 2^{{\vert}U{\vert}}$, and $\ker(A) = 0$. Thus, every entry of $\psi$ must be $0$.

Finally, it remains to consider the case when ${\vert}T \cap T'{\vert} > 1$. We let $U = T \cup T'$, we let $\ell = T \cap T'$, and we assume WLOG that $T = \{1, 2, 3, \dots \ell, \ell +1, \dots r\}$ and $T' = \{1, 2, \dots \ell, r+\ell, \dots 2r-1 \}$. 

In this case, as before, we know that for any choice of $y \in \zo^{[n] - U}$, our requirement that $\psi$ gives $0$ energy enforces that
    \begin{align}\label{eq:zeroEnergyLouie}
    \langle \psi_y, M_{T} \otimes \Id_{U - T} \psi_y \rangle = \langle \psi_y, M_{T'} \otimes \Id_{U - T'} \psi_y \rangle = 0,
    \end{align}
where $\psi_y \in \C^{2^{{\vert}U{\vert}}}$. For clarity, we let $\phi = \psi_y$.

Now, suppose for the sake of contradiction that there existed a non-zero vector $\phi \in \C^{2^{{\vert}U{\vert}}}$ such that \cref{eq:zeroEnergyLouie} holds. We now lift $\phi$ to a new vector $\widehat{\phi} \in \C^{2^{2r-1}}$, such that, for any entry of $\widehat{\phi}$, indexed by $b \in \zo^{2r-1}$, we have:
\[
\widehat{\phi}_b = \begin{cases}
    0 \text{ if } \exists i \in \{2, 3, 4, \dots \ell\}: b_i \neq b_{i + r -1}. \\
    \phi_{b_{1:r} \circ b_{r+\ell:2r-1}}.
\end{cases}
\]

At the same time, we consider a different Hamiltonian, with constraints $S = \{1, 2, 3, \dots r\}$ and $S' = \{1, r+1, r+2, \dots 2r-1\}$. The key observation is that, by construction, 
\[
\langle \phi, M_{T} \otimes \Id_{U - T} \phi \rangle = \langle \widehat{\phi}, M_S \otimes \Id_{[2r-1] - S} \widehat{\phi}\rangle
\]
and 
\[
\langle \phi, M_{T'} \otimes \Id_{U - T'} \phi \rangle = \langle \widehat{\phi}, M_{S'} \otimes \Id_{[2r-1] - S'} \widehat{\phi}\rangle.
\]

But, for these terms $S, S'$, we see that ${\vert}S \cap S'{\vert} = 1$, and thus the only state $\widehat{\phi}$ which gives both $S$ and $S'$ zero energy is when every entry of $\widehat{\phi}$ is $0$. But, this then implies that every entry of $\phi$ itself is also $0$. This concludes the claim. 
\end{proof}

Note that by repeating the above argument, we can also prove a version of \cref{lem:nullSpaceLouie} where $T$ and $T'$ have \emph{different} random operators $M, M'$. The same argument goes through as we can re-use the proof from \cref{lem:nullSpaceLouie} to show that when $M = M'$, the determinant polynomial is non-zero (and thus, when $M, M'$ are distinct, the determinant polynomial is still non-zero). Thus, we formally get the following corollary:

\begin{corollary}\label{lem:nullSpaceDifferentMLouie}
    Let $M, M'\in \C^{2^r \times 2^r}$ be random predicate matrices of rank $R \geq 2^{r-1} + 1$, and let $T, T' \subseteq [n]$, ${\vert}T{\vert} = r$ denote two local constraints. Suppose that $T \cap T' \neq \emptyset$. Then, with probability $1$ over $M, M'$, for any state $\psi \in \C^{2^n}$ such that $\cE_{T, M}(\psi) = \cE_{T', M}(\psi) = 0$, it must be the case that for any $z \in \zo^n$ that $\psi_z = 0$.
\end{corollary}

Now, as a consequence of \cref{lem:nullSpaceLouie}, we also have the following corollary:

\begin{corollary}\label{cor:IntersectingEnergyLouie}
    Let $M \in \C^{2^r \times 2^r}$ be a random predicate matrix of rank $R \geq 2^{r-1} + 1$, and let $T, T', T'' \subseteq [n]$, ${\vert}T{\vert} = {\vert}T'{\vert} = r$ denote two $r$-local constraints. Suppose that $T \cap T' \neq \emptyset$. Then, with probability $1 - o_r(1)$ over $M$,
    \[
    \lambda_{\min}(M{\vert}_{T} \otimes \Id_{[n] - T} + M{\vert}_{T'} \otimes \Id_{[n] - T'}) \geq \Omega_{r}(1),
    \]
    and 
    \[
    \lambda_{\max}(M{\vert}_{T''} \otimes \Id_{[n] - T''})\leq O_r(1).
    \]
\end{corollary}

\begin{proof}
To see the second point above, we can simply observe that 
\[
\lambda_{\max}(M{\vert}_{T''} \otimes \Id_{[n] - T''}) = \lambda_{\max}(M)\leq O_r(1),
\]
by our assumption on boundedness (see \cref{def:OnlyCbounded}, \cref{rmk:assumptionsDist}).

For the first point, we let $U = T \cup T'$. We observe that we can view $M{\vert}_{T} \otimes \Id_{[n] - T} + M{\vert}_{T'} \otimes \Id_{[n] - T'}$ as a local operator on $U = T \cup T'$. Thus, there is a matrix $\hat{M} \in \C^{2^{{\vert}U{\vert}} \times 2^{{\vert}U{\vert}}}$ such that 
\[
M{\vert}_{T} \otimes \Id_{[n] - T} + M{\vert}_{T'} \otimes \Id_{[n] - T'} = \hat{M}_U \otimes \Id_{[n] - U}.
\]
Indeed, in this formulation, we see that $\hat{M} = M{\vert}_T \otimes \Id_{U - T} + M{\vert}_{T'} \otimes \Id_{U - T'}$. Thus the eigenvalues of $\hat{M}$ depend only on the entries of $M$, which in turn depends only on the distribution from which $M$ is sampled and the parameter $r$. Likewise, via \cref{lem:nullSpaceLouie}, we know that all eigenvalues of $\hat{M} = M{\vert}_T \otimes \Id_{U - T} + M{\vert}_{T'} \otimes \Id_{U - T'}$ are positive. Thus, we can see that, with probability (say) $1-o_r(1)$,
\[
    \lambda_{\min}(M{\vert}_{T} \otimes \Id_{[n] - T} + M{\vert}_{T'} \otimes \Id_{[n] - T'}) \geq \Omega_{r}(1),
    \]
as this matrix $\hat{M}$ is drawn from a distribution whose parameters depend only on $r$. Importantly, as discussed in \cref{rmk:assumptionsDist}, the parameters of this distribution are \emph{independent} of $n$.
\end{proof}

Invoking the same argument (but instead using \cref{lem:nullSpaceDifferentMLouie}), we can also show the following:

\begin{corollary}\label{cor:IntersectingEnergyDifferentMLouie}
    Let $M, M', M'' \in \C^{2^r \times 2^r}$ be random predicate matrices of rank $R \geq 2^{r-1} + 1$, and let $T, T', T'' \subseteq [n]$, ${\vert}T{\vert} = {\vert}T'{\vert} = r$ denote three $r$-local constraints, using $M, M', M''$ respectively. Then:
    \begin{enumerate}
        \item With probability $1$ over $M''$,    \[
    \lambda_{\max}(M''{\vert}_{T''} \otimes \Id_{[n] - T''})\leq O_{r}(1).
    \]
    \item If $T \cap T' \neq \emptyset$, then with probability $\geq 1/2$ over $M, M'$,
    \[
    \lambda_{\min}(M{\vert}_{T} \otimes \Id_{[n] - T} + M'{\vert}_{T'} \otimes \Id_{[n] - T'}) \geq \Omega_{r}(1).
    \]
    \end{enumerate}
\end{corollary}

\begin{proof}
    The first point follows from our assumption on the boundedness of our matrices (see \cref{def:OnlyCbounded}).

    For the second point, we can again observe that, letting $U = T \cup T'$, we can write 
    \[
    M{\vert}_{T} \otimes \Id_{[n] - T} + M'{\vert}_{T'} \otimes \Id_{[n] - T'} = \hat{M}{\vert}_U \otimes \Id_{[n] - U}.
    \]
    Here, $U \in \C^{2^{O(r)} \times 2^{O(r)}}$, and $U$'s entries depend only on the entries of $M', M$. From \cref{lem:nullSpaceDifferentMLouie}, we know that with probability $1$ over $M, M'$ that $\lambda_{\min}(M{\vert}_{T} \otimes \Id_{[n] - T} + M'{\vert}_{T'} \otimes \Id_{[n] - T'}) > 0$. 
    
    Importantly, because the distribution from which $M, M'$ are drawn is only parameterized by $r$ (see \cref{rmk:assumptionsDist}), the statistics of this distribution also depend only on $r$. Thus, we can then say that there exists some constant $\gamma_r \geq \Omega_r(1)$ such that with probability $1/2$, 
    \[
    \lambda_{\min}(M{\vert}_{T} \otimes \Id_{[n] - T} + M'{\vert}_{T'} \otimes \Id_{[n] - T'}) = \lambda_{\min}(\hat{M}) \geq \gamma_r \geq \Omega_{r}(1).
    \]
\end{proof}

In the above, we use $\Omega_r(\cdot)$ and $O_r(\cdot)$ to hide (potentially large) constant dependences on a function of $r$. Note that it may be slightly counter-intuitive in \cref{cor:IntersectingEnergyDifferentMLouie} that we only use a weak property that with probability $1/2$ the smallest eigenvalue of the sum of terms is $\Omega_r(1)$. As we shall we see later, this is because when the number of terms increases (scaling with $n$, the number of qubits, independent of $r$), it is possible that certain ``low probability'' tail events may occur which could yield eigenvalues that scale with $n$ as opposed to $r$. \cref{cor:IntersectingEnergyDifferentMLouie} provides a very basic, weak guarantee which nevertheless still suffices for our sparsifiers by using parameters of the distribution of predicate matrices that are \emph{independent} of $n$.

\subsection{Building Sparsifiers}

\paragraph{Sparsifiers for the Shared Predicate Case}

With \cref{cor:IntersectingEnergyLouie} now established, we proceed to our algorithm for building sparsifiers. Roughly speaking, the intuition is that once there are sufficiently many terms which operate on \emph{intersecting} sets of qubits, then this naturally provides a stronger lower bound on the minimum eigenvalue of the entire Hamiltonian $\cH$. Towards proving this intuition formally, we first make the following claim:

\begin{claim}\label{clm:disjointIntersectingTermsLouie}
    Let $\cH$ be a Hamiltonian with terms $H_i = (M, T_i): i \in [m]$. Provided $m \geq 4n$, then there must exist $\geq \frac{m}{4}$ disjoint pairs of terms $(a_j, b_j): j \in [\frac{m}{2}], a_j, b_j \in [m]$ such that $T_{a_j} \cap T_{b_j} \neq \emptyset$.
\end{claim}

\begin{proof}
    Indeed, consider the following basic algorithm: let $S = [m]$ denote the starting set of all terms in the Hamiltonian, and let $j = 1$. Now, while there exists a pair of disjoint terms  $(a_j, b_j), a_j, b_j \in [m]$  such that $T_{a_j} \cap T_{b_j} \neq \emptyset$, we remove this pair of terms (i.e., $S \leftarrow S \setminus \{a_j, b_j\}$), and increment $j$ by $1$.

    We claim that in this process, the final value of $j$ will be at least $\frac{m}{2} +1$ (meaning that at least $\frac{m}{2}$ many pairs of intersecting terms are identified). Indeed, to see why, suppose that ${\vert}S{\vert} > \frac{n}{r}$. Then, we can see that the total sum of the degrees of the underlying qubits is $> \frac{n}{r} \cdot r = n$. Thus, there must be some qubit whose degree is $\geq 2$, which means that there must be some two terms that intersect.

    Thus, in essence, while ${\vert}S{\vert} > \frac{n}{r}$, we can find a pair of intersecting terms. Thus, we can identify 
    \[
    \geq \frac{m - \frac{n}{r}}{2} \geq \frac{m}{4}
    \]
    many intersecting terms in the Hamiltonian.
\end{proof}

With the above claim, we now make the following definition of \emph{importance scores}:

\begin{definition}\label{def:importanceScoreLouie}
    Let $\cH$ be a Hamiltonian with terms $H_i = (M_i, T_i): i \in [m]$. For each term $i \in [m]$, we define the \emph{importance score} of $H_i$ as:
    \[
    \mathrm{Importance}(H_i) = \max_{\psi \in \C^{2^n}} \frac{\langle \psi {\vert} H_i {\vert} \psi \rangle}{\langle \psi {\vert} \cH{\vert} \psi \rangle}.
    \]
\end{definition}

Importantly, we have the following claim regarding importance scores:

\begin{claim}\label{clm:importanceSamplingLouie}
    Let $\cH$ be a Hamiltonian with $r$-local terms $H_i = (M_i, T_i): i \in [m]$ on $n$ qubits, and let $\eps \in (0,1)$. Suppose that for every term $i \in [m]$, 
    \[
    \mathrm{Importance}(H_i) \leq p. 
    \]
    Let $C = \frac{100 n}{\eps^2}$. Now, let $L \subseteq [m]$ be a random sample of terms, where $i$ is sampled in $L$ independently with probability $p \cdot C$. Then, with probability $1 - 2^{- \Omega(n)}$ over the sample $L$, it is the case that, for every state $\psi \in \C^{2^n}$ that
    \[
    \sum_{i \in L} \frac{1}{p \cdot C} \cdot \langle \psi {\vert} H_i{\vert} \psi \rangle \in (1 \pm \eps) \cdot \langle \psi {\vert} \cH{\vert} \psi \rangle.
    \]
    Furthermore, with probability $1 - 2^{- \Omega(n)}$, the number of sampled terms is bounded by $O(\eps^{-2} \cdot m \cdot p \cdot n).$
\end{claim}

\begin{proof}
    We use Matrix Chernoff to prove this (see \cref{fact:matrix-bernstein-restated}). For each $i \in [m]$, we let 
    \[
    Y_i = \begin{cases}
        \frac{1}{p \cdot C} \cdot H_i \text{ w.p. } p \cdot C \\
        0 \quad \quad\text{ otherwise}.
    \end{cases}
    \]

    Letting $Z = \E[\sum_{i = 1}^m Y_i]$, we can then immediately see that $Z = \sum_{i = 1}^m H_i = \cH$. Likewise, because 
\[
    \mathrm{Importance}(H_i) \leq p,
    \]
    we then also know that for every $\psi \in \C^{2^n}$ that
    \[
    \langle \psi {\vert} H_i {\vert} \psi \rangle \leq p \cdot \langle \psi {\vert} \cH{\vert} \psi \rangle = p \cdot \langle \psi {\vert} Z{\vert} \psi \rangle.
    \]
    Because $Y_i \preceq \frac{1}{p \cdot C} \cdot H_i$, we then obtain that 
    \[
    Y_i \preceq \frac{1}{p \cdot C} \cdot p \cdot Z = \frac{1}{C} \cdot Z.
    \]

    Thus, by using the statement of \cref{fact:matrix-bernstein-restated}, we see that with probability $1 - 2 \cdot 2^n \cdot 2^{- \eps / 3C} = 1-2^{- \Omega(n)}$ by our choice of $C$, that 
    \[
    (1 - \eps) Z \preceq \sum_{i = 1}^m Y_i \preceq (1 + \eps) Z.
    \]

    Finally, we can see that $\sum_{i = 1}^m Y_i$ has the same distribution as $\sum_{i \in L} \frac{1}{p \cdot C} \cdot \langle \psi {\vert} H_i{\vert} \psi \rangle$, by construction. 

    Thus, we obtain that with probability $1 - 2^{- \Omega(m)}$ over the sample $L$ that 
        \[
    \sum_{i \in L} \frac{1}{p \cdot C} \cdot \langle \psi {\vert} H_i{\vert} \psi \rangle \in (1 \pm \eps) \cdot \langle \psi {\vert} \cH{\vert} \psi \rangle,
    \]
    as we desire.

    The bound on the number of sampled terms follows from a simple Chernoff bound.
\end{proof}

Now, we use \cref{clm:disjointIntersectingTermsLouie} and \cref{cor:IntersectingEnergyLouie} to show the following:

\begin{claim}\label{clm:lowerBoundHEigenLouie}
    Let $\cH = (M, T_i): i \in [m]$ denote an $r$-local Hamiltonian with $m$ terms over $n$ qubits, and random predicate matrix $M \in \C^{2^r \times 2^r}$. Suppose $m \geq 4n$ of rank $R \geq 2^{r-1} + 1$. Then, with probability $1 - o_r(1)$ over $M$, it is the case that \[
    \lambda_{\min}(\cH) \geq \frac{m}{4} \cdot \Omega_r(1).
    \]
\end{claim}

\begin{proof}
    Indeed, by \cref{clm:disjointIntersectingTermsLouie}, we know that we can find $\geq m/4$ disjoint pairs of terms $H_{a_j} = (M, T_{a_j}), H_{b_j}= (M, T_{b_j})$ such that $T_{a_j} \cap T_{b_j} \neq \emptyset$. Then, by  \cref{cor:IntersectingEnergyLouie}, we know that with probability $0.99$,
    \[
    \lambda_{\min}(H_{a_j} + H_{b_j}) \geq \Omega_r(1).
    \]

    Thus, summing over these disjoint pairs, we see that 
    \[
    \lambda_{\min}(\cH) \geq \sum_{j \in [m/4]}\lambda_{\min}(H_{a_j} + H_{b_j}) = \frac{m}{4} \cdot \Omega_r(1) \geq \Omega_r(m).
    \]
\end{proof}

Finally, by using \cref{clm:lowerBoundHEigenLouie}, we can now bound the importance of every individual Hamiltonian term:

\begin{claim}\label{clm:boundImportanceLouie}
    Let $\cH = (M, T_i): i \in [m]$ denote an $r$-local Hamiltonian with $m$ terms over $n$ qubits, and random predicate matrix $M \in \C^{2^r \times 2^r}$ of rank $R \geq 2^{r-1} + 1$. Suppose $m \geq 4n$. Then, with probability $1 - o_r(1)$ over $M$, for every $i \in [m]$, it is the case that 
    \[
    \mathrm{Importance}(H_i) \leq \frac{1}{\Omega_r(m)}.
    \]
\end{claim}

\begin{proof}
    Fix any individual term $H_i$. We see that 
    \[
    \mathrm{Importance}(H_i) = \max_{\psi \in \C^{2^n}} \frac{\langle \psi {\vert} H_i {\vert} \psi \rangle}{\langle \psi {\vert} \cH{\vert} \psi \rangle} \leq \frac{\lambda_{\max}(H_i)}{\lambda_{\min}(\cH)} \leq \frac{1}{\Omega_r(m)},
    \]
    where the final inequality uses \cref{clm:lowerBoundHEigenLouie} and \cref{cor:IntersectingEnergyLouie}.
\end{proof}

With this established, we are now ready to prove \cref{thm:EfficientRandomSparsifierLouie}:

\begin{proof}[Proof of \cref{thm:EfficientRandomSparsifierLouie}.]
    Indeed, let $\eps$ be given, and suppose that $M$ is generic, as required by \cref{lem:nullSpaceLouie} (which occurs with probability $0.99$ over $M$). Further, let us suppose that $m \geq 4n$.
    
    Then, via \cref{clm:boundImportanceLouie}, we see that for every term $H_i$ in $\cH$ that 
    \[
    \mathrm{Importance}(H_i) \leq \frac{1}{\Omega_r(m)} = p.
    \]
    Now, we can directly invoke \cref{clm:importanceSamplingLouie}; this shows that with probability $1 - 2^{- \Omega(n)}$, we can randomly sample the terms if $\cH$ at rate $p \cdot C$, assigning weight $1 / p \cdot C$ to those sampled terms. Let us use $L \subseteq [m]$ to denote these sampled terms, and let $\widetilde{\cH}$ denote the Hamiltonian with terms from $L$, each weighted by $1 / p \cdot C$. 

    \cref{clm:importanceSamplingLouie} implies that, for every state $\psi \in \C^{2^n}$ that
    \[
    \langle \psi {\vert} \widetilde{\cH}{\vert} \psi \rangle \in (1\pm\eps) \cdot \langle \psi {\vert} \cH{\vert} \psi \rangle.
    \]

    Simultaneously, \cref{clm:importanceSamplingLouie} implies that 
    \[
    {\vert}L{\vert}\leq O(\eps^{-2} \cdot m \cdot p \cdot n)\leq O_M(\eps^{-2}n).
    \]
    Importantly, this quantity is \emph{independent} of the number of qubits in our system, and so the sparsifier size scales only with $n$ and $\eps^2$.

    Finally, recall that our invocation of \cref{prop:r-partite-WLOG} introduces logarithmic factors of $n$ to the sparsifier size, thus yielding our theorem.
\end{proof}

\paragraph{Sparsifiers for the Different Predicate Case}

We can prove \cref{thm:EfficientRandomSparsifierDifferentMLouie} in an identical manner:

\begin{claim}\label{clm:lowerBoundHEigenDifferentLouie}
    Let $\cH = (M_i, T_i): i \in [m]$ denote an $r$-local Hamiltonian with $m$ terms over $n$ qubits, and random predicate matrices $M_i \in \C^{2^r \times 2^r}$ of rank $R \geq 2^{r-1} + 1$. Suppose $m \geq 4n$. Then, with probability $1 - 2^{- \Omega(n)}$ over the random predicate matrices $M_i \in \C^{2^r \times 2^r}$, it is the case that \[
    \lambda_{\min}(\cH) \geq \frac{m}{4} \cdot \Omega_r(1).
    \]
\end{claim}

\begin{proof}
    Indeed, by \cref{clm:disjointIntersectingTermsLouie}, we know that we can find $\geq m/4$ disjoint pairs of terms $H_{a_j} = (M, T_{a_j}), H_{b_j}= (M, T_{b_j})$ such that $T_{a_j} \cap T_{b_j} \neq \emptyset$. Then, by  \cref{cor:IntersectingEnergyDifferentMLouie}, we know that with probability $1/2$ over each pair,
    \[
    \lambda_{\min}(H_{a_j} + H_{b_j}) \geq \Omega_r(1).
    \]

    Thus, summing over these disjoint pairs and taking a Chernoff bound, we see that with probability $1 - 2^{- \Omega(n)}$, at least a $1/4$ fraction of these pairs have $\lambda_{\min}(H_{a_j} + H_{b_j}) \geq \Omega_r(1)$. Thus, with probability $1 - 2^{- \Omega(n)}$,
    \[
    \lambda_{\min}(\cH) \geq \sum_{j \in [m/4]}\lambda_{\min}(H_{a_j} + H_{b_j}) = \frac{m}{4} \cdot \Omega_r(1) \geq \Omega_r(m).
    \]
\end{proof}

With \cref{clm:lowerBoundHEigenDifferentLouie}, we can now bound the importance of every individual Hamiltonian term:

\begin{claim}\label{clm:boundImportanceDifferentLouie}
    Let $\cH = (M, T_i): i \in [m]$ denote an $r$-local Hamiltonian with $m$ terms over $n$ qubits, and random predicate matrix $M \in \C^{2^r \times 2^r}$ of rank $R \geq 2^{r-1}+1$. Suppose $m \geq 4n$. Then, with probability $1 - 2^{- \Omega(n)}$, for every $i \in [m]$, it is the case that 
    \[
    \mathrm{Importance}(H_i) \leq \frac{1}{\Omega_r(m)}.
    \]
\end{claim}

\begin{proof}
    Fix any individual term $H_i$. We see that 
    \[
    \mathrm{Importance}(H_i) = \max_{\psi \in \C^{2^n}} \frac{\langle \psi {\vert} H_i {\vert} \psi \rangle}{\langle \psi {\vert} \cH{\vert} \psi \rangle} \leq \frac{\lambda_{\max}(H_i)}{\lambda_{\min}(\cH)} \leq \frac{1}{\Omega_r(m)},
    \]
    where the final inequality uses \cref{clm:lowerBoundHEigenDifferentLouie} and \cref{cor:IntersectingEnergyDifferentMLouie}.
\end{proof}

With this established, we are now ready to prove \cref{thm:EfficientRandomSparsifierDifferentMLouie}:

\begin{proof}[Proof of \cref{thm:EfficientRandomSparsifierDifferentMLouie}.]
    Indeed, let $\eps$ be given, and suppose that $M$ is generic, as required by \cref{lem:nullSpaceLouie} (which occurs with probability $1$ over $M$). Further, let us suppose that $m \geq 4n$.
    
    Then, via \cref{clm:boundImportanceLouie}, we see that, with probability $1 - 2^{- \Omega(n)}$, for every term $H_i$ in $\cH$ that 
    \[
    \mathrm{Importance}(H_i) \leq \frac{1}{\Omega_r(m)} := p.
    \]
    Now, we can directly invoke \cref{clm:importanceSamplingLouie}; this shows that with probability $1 - 2^{- \Omega(n)}$, we can randomly sample the terms if $\cH$ at rate $p \cdot C$, assigning weight $1 / p \cdot C$ to those sampled terms. Let us use $L \subseteq [m]$ to denote these sampled terms, and let $\widetilde{\cH}$ denote the Hamiltonian with terms from $L$, each weighted by $1 / p \cdot C$. 

    \cref{clm:importanceSamplingLouie} implies that, for every state $\psi \in \C^{2^n}$ that
    \[
    \langle \psi {\vert} \widetilde{\cH}{\vert} \psi \rangle \in (1\pm\eps) \cdot \langle \psi {\vert} \cH{\vert} \psi \rangle.
    \]

    Simultaneously, \cref{clm:importanceSamplingLouie} implies that 
    \[
    {\vert}L{\vert}\leq O(\eps^{-2} \cdot m \cdot p \cdot n)\leq O_r(\eps^{-2} n).
    \]
    Importantly, this quantity is \emph{independent} of the number of qubits in our system, and so the sparsifier size scales only with $n$ and $\eps^2$.

    Finally, recall that our invocation of \cref{prop:r-partite-WLOG} introduces logarithmic factors of $n$ to the sparsifier size, thus yielding our theorem.
\end{proof}

\subsection{Separations Between Sparsifiability of Classical and Quantum Systems}\label{subsec:separation}

So far, we have showed that random arity-$r$ Hamiltonians of rank at least $2^{r-1}+1$ have near-quadratic size sparsifiers. Now we show that this result is in sharp contrast to the classical setting--the case in which our $M \in \C^{2^r \times 2^r}$ is a diagonal matrix with $\{0,1\}$-entries  (see \cref{prop:classical-equivalence}). In particular, we show the following result.

\begin{theorem}\label{thm:random-classical-LB}
Let $R \subseteq \{0,1\}^r$ be a uniformly random predicate with cardinality $k \in [\alpha, 1-\alpha] \cdot 2^r$ where $\alpha \in (0,1/2)$ are constants independent of $r$. For any sufficiently large positive integer $n$ and any $\eps > 0$, with probability $1-o(1)$ (where $o(1)$ hids a function of $r$ approaching zero), we have that
\[
    \SPR(R, n, \eps) \geq \Omega_r(n^{\lceil \log_2 r - \log_2 \log_2(1/\alpha) - 2\rceil}).
\]
\end{theorem}

\begin{remark}
If $k = 2^{r-1}+1$ like in \cref{thm:EfficientRandomSparsifierLouie}, then $\alpha$ can be taken to be $1/2-o(1)$, so the lower bound is at least $\Omega_r(n^{\lceil \log_2 r - 3\rceil})$. Thus, random classical predicates of a fixed cardinality are often much more difficult to sparsify than random Hamiltonians of a fixed rank.
\end{remark}

To build up to \Cref{thm:random-classical-LB}, we need a key tool in the study of classical CSP sparsification and related questions which is the use of rich families of \emph{gadget reductions} to connect the sparsifiability of some predicate with the sparsifiability of other predicates~\cite{FiltserK17,ButtiZ20,ChenJP20,LagerkvistW20,BessiereCK20,Carbonnel22,KhannaPS24,KhannaPS25,BrakensiekG25,BrakensiekGJLW25}. In particular, we use the ``projections'' framework for Boolean CSPs~\cite{FiltserK17,KhannaPS25}.

We give a summary of the projections framework for Boolean CSPs. Given a pool of Boolean variables $x_1, \hdots, x_c$, we define a \emph{literal} to be one of the following:
\begin{itemize}
\item A constant $\ell \in \{0,1\}$.
\item A variable $\ell \in \{x_1, \hdots, x_c\}$.
\item A negated variable $\ell \in \{\bar{x}_1, \hdots, \bar{x}_c\}$.
\end{itemize}
We let $\cL_c = \{0,1,x_1, \hdots, x_c, \bar{x}_1, \hdots, \bar{x}_c\}$ denote this list of literals. 

Given a relation $R \in \{0,1\}^r$ and a sequence of literals $\ell_1, \hdots, \ell_r \in \cL_c$ (possibly with repetition), we define the \emph{projection} of $R$ with respect to $\ell_1, \hdots, \ell_r$ to be a relation $R(\ell_1, \hdots, \ell_r) \in \{0,1\}^c$ such that $(x_1, \hdots, x_c) \in R(\ell_1, \hdots, \ell_r)$ if and only if the evaluation of $\ell_1, \hdots, \ell_r$ with respect to $(x_1, \hdots, x_c)$ lies in $R$.

\begin{remark}\label{rem:example-R}
As a concrete example, consider $r = 3, c = 2$ and $R = \{001,101,111\}$. Then, we have that
\begin{align*}
R(x_1, x_2, 1) &= \{00, 10, 11\},\\
R(x_1, \bar{x}_1, x_2) &= \{11\}.
\end{align*}
\end{remark}

A crucially property of projections is that they can only decrease the sparsifier size.
\begin{proposition}[Implicit in \cite{KhannaPS25}, see also \cite{Carbonnel22}]\label{prop:literal-reduction}
Fix $r \geq c \geq 1$. For any $R \in \{0,1\}^r$ and any $\ell_1, \hdots, \ell_r \in \cL_c$, we have for all $n \geq 1$ and $\eps > 0$ that
\[
    \SPR(R, n, \eps) \geq \Omega_{r}(\SPR(R(\ell_1, \hdots, \ell_r), n, \eps)).
\]
\end{proposition}

To benefit from considering these projections, we also use the following lower bounds on the sparsifier size of the predicate $\AND_c := \{1^c\}.$
\begin{proposition}[\cite{KhannaPS25}, see also \cite{Carbonnel22,ChenKN20,LagerkvistW20}]\label{prop:AND-lb}
For any $c \geq 1, n \geq 1$, and $\eps > 0$, we have that
\[
    \SPR(\AND_c, n, \eps) \geq \binom{n}{c}\geq \Omega_{c}(n^c).
\]
\end{proposition}

We note that a Hamiltonian analogue of \cref{prop:AND-lb} is proved in \cref{sec:tensor-lb}.

\begin{remark}
Continuing from \cref{rem:example-R} the example of $R = \{001,101,111\}$, the projection $R(x_1, \bar{x}_1, x_2) = \{11\}$ proves that $\SPR(R, n, \eps) \geq \Omega(n^2)$, which can be shown to be tight up to a factor of $\widetilde{O}(\eps^{-2})$ using the polynomial method of \cite{KhannaPS25}.
\end{remark}

We also need the following fact about the binomial distribution.
\begin{theorem}[e.g., \cite{Ash1990}]\label{claim:binomial-equal}
Let $p = k/n$ and consider the binomial distribution $B(n, p)$. The probability that $m$ is sampled from $B(n,p)$ is at least
\[
    \frac{1}{\sqrt{8n \cdot \frac{m}{n} \cdot (1 - \frac{m}{n})}} \cdot \exp\left(-m \log\frac{m}{k} + (n-m) \log \frac{n-m}{n-k}\right).
\]
\end{theorem}

We now prove \cref{thm:random-classical-LB}.

\begin{proof}[Proof of \cref{thm:random-classical-LB}]
Let $c$ be a positive integer to be determined. Let $p = k/2^r$, and let $X$ be the distribution over relations $\{0,1\}^r$ such that each element is included independently with probability $p$. By \cref{claim:binomial-equal}, the probability that $X$ has cardinality exactly $k$ is at least
\begin{align}
    \Pr[{\vert}X{\vert}=k] \geq \frac{1}{\sqrt{2^{r+3} \cdot \frac{k}{2^r} \cdot (1 - \frac{k}{2^r})}} \cdot \exp(0) \geq \frac{1}{2^{(r+1)/2}}.\label{eq:k-lb}
\end{align}
If we condition on $X$ having cardinality exactly $k$, then we recover the distribution $R$ of uniformly random relations of $\{0,1\}^r$ of arity $k$. In particular, if we prove that $\NRD(X, n, \eps) > \Omega_r(n^c)$ with probability at least $1-q$, then 
\begin{align}
    \Pr[\NRD(R, n, \eps) > \Omega_r(n^c)] &= \Pr[\NRD(X, n, \eps) > \Omega_r(n^c) {\vert} {\vert}X{\vert} = k]\nonumber\\
    &= \frac{\Pr[\NRD(X, n, \eps) > \Omega_r(n^c) \wedge {\vert}X{\vert} = k]}{\Pr[{\vert}X{\vert} = k]}\nonumber\\
    &\geq \frac{\Pr[\NRD(X, n, \eps) > \Omega_r(n^c) + \Pr[{\vert}X{\vert} = k] - 1}{\Pr[{\vert}X{\vert}=k]}\nonumber\\
    &= 1 - \frac{q}{\Pr[{\vert}X{\vert}=k]}\nonumber\\
    &\geq 1 - 2^{(r+1)/2}q,\label{eq:R-lb}
\end{align}
where the last line uses \cref{eq:k-lb}. Thus, we now focus on giving a lower bound for $\Pr[\NRD(X, n, \eps) > \Omega_r(n^c)]$. By \cref{prop:literal-reduction} and \cref{prop:AND-lb}, it suffices to lower bound the probability that $X$ has a projection to $\AND_c$. In fact, we restrict ourselves to a special family of projections.

For every $b \in \{0,1\}^{r-c}$ and any choice of literals $\ell_1 \in \{x_1, \bar{x}_1\}, \hdots, \ell_c \in \{x_c, \bar{x}_c\}$, consider the projection $X(b_1, \hdots, b_{r-c}, \ell_1, \hdots, \ell_c)$. The probability that $X(b_1, \hdots, b_{r-c}, \ell_1, \hdots, \ell_c) = \AND_c$ is exactly $p(1-p)^{2^c-1}$ as the elements of $X(b_1, \hdots, b_{r-c}, \ell_1, \hdots, \ell_c) \subseteq \{0,1\}^c$ are each sampled with probability $p$. Observe that for any choice of $(\ell'_1, \hdots, \ell'_c) \in \{x_1, \bar{x}_1\} \times \cdots \times \{x_c, \bar{x}_c\} \setminus \{(\ell_1, \hdots, \ell_c)\}$, the events $X(b_1, \hdots, b_{r-c}, \ell_1, \hdots, \ell_c) = \AND_c$ and $X(b_1, \hdots, b_{r-c}, \ell'_1, \hdots, \ell'_c) = \AND_c$ cannot occur simultaneously. Thus, for a fixed choice of $b$, the probability that $X(b_1, \hdots, b_{r-c}, \ell_1, \hdots, \ell_c) = \AND_c$ for some choice of $(\ell_1, \hdots, \ell_c) \in \{x_1, \bar{x}_1\} \times \cdots \times \{x_c, \bar{x}_c\}$ is precisely $2^cp(1-p)^{2^c-1}$. For ease of notation, we call this event $E(b)$.

Now, consider an enumeration $b_1, \hdots, b_{2^{r-c}}$ of the elements of $\{0,1\}^{r-c}$. We can see that the events $E(b_1), \hdots, E(b_{2^{r-c}})$ are independent, as the projections of $X$ considered by these events always consider disjoint elements of $X$.  Therefore, the probability that at least one of $E(b_1), \hdots, E(b_{2^{r-c}})$ occurs is exactly
\begin{align}
    1 - (1 - 2^cp(1-p)^{2^c-1})^{2^{r-c}}.\label{eq:Eb}
\end{align}
By \cref{prop:literal-reduction} and \cref{prop:AND-lb}, the quantity $\Pr[\NRD(X, n, \eps) > \Omega_r(n^c)] = 1-q$ is at least this probability from \cref{eq:Eb}. Therefore,
\[
    q \leq (1 - 2^cp(1-p)^{2^c-1})^{2^{r-c}}
\]
Therefore, by \cref{eq:R-lb}, to prove our theorem it suffices to prove that $2^{(r+1)/2}(1 - 2^cp(1-p)^{2^c-1})^{2^{r-c}} = o(1)$ when $p = \frac{k}{2^r} \in [\alpha,  1-\alpha]$. In particular note that
\begin{align*}
    1 -\Pr[\NRD(R, n, \eps) > \Omega_r(n^c)] &\leq 2^{(r+1)/2}(1 - 2^cp(1-p)^{2^c-1})^{2^{r-c}}\\
    &\leq 2^{(r+1)/2}(1 - 2^c\alpha^{2^c})^{2^{r-c}}\\
    &\leq 2^{(r+1)/2}\exp(-2^{r-c} \cdot 2^c\alpha^{2^c})\\
    &\leq 2^{(r+1)/2}\exp(-2^{r}\alpha^{2^c}).
\end{align*}
If $c = \rceil \log_2 r - \log_2 \log_2 (1/\alpha) - 2\rceil \leq  \log_2 r - \log_2 \log_2 (1/\alpha) - 1$, then
\[
\alpha^{2^c} \leq \alpha^{2^{\log_2 r - \log_2 \log_2 (1/\alpha) - 1}} = \alpha^{\frac{r}{2\log_2(1/\alpha)}} = \alpha^{\frac{r\log_{\alpha}(1/2)}{2}} = 2^{-r/2}.
\]
In such case, we have that $1 -\Pr[\NRD(R, n, \eps) \leq \Omega_r(n^c)] \leq 2^{(r+1)/2} \exp(-2 ^{r/2}) = o(1)$. Thus, $\NRD(R, n, \eps) \geq \Omega_r(n^{\lceil \log_2 r - \log_2 \log_2 (1/\alpha) - 2\rceil})$ with high probability, as desired.
\end{proof}

\section{Sparsifying Nullity $1$ Predicates}\label{sec:nullity-1}
Let $\cH = \{(T_i, M_i)\}_{i\in [m]}$ be a Hamiltonian where $\operatorname{nullity}(M_i) \leq 1$ for all $i$, i.e. $M_i\in\C^{2^r\times 2^r}$ are rank $\geq 2^r - 1$ PSD matrices. Let $\cM:= \{M_i\}_{i\in[m]}\subset\C^{2^r\times 2^r}$ be the family of operators associated to $\cH$. Throughout \cref{sec:nullity-1}, we'll assume that $\cM$ is $(C, 2r)$-bounded for some $C > 0$ (see \cref{def:Cbounded}). \footnote{Note that by \cref{example:singleMbounded}, if all $M_i$ are the same matrix $M$, then we can take $C$ to be a function of $r, M$ alone. In particular, since $M$ can be seen as a (matrix-valued) function of $r$ alone, we can take $C\leq O_r(1)$. Thus $C$ should be seen as a ``constant'' in whatever follows.} In this section, we prove the following theorem about the sparsifiability of nullity $1$ Hamiltonians:
\begin{theorem}\label{thm:Nullity1Sparsifier}
    Let $\mathcal{H} = \{(T_i, M_i):{i \in [m]}\}$ be an arbitrary nullity $1$ Hamiltonian over $n$ variables and $m$ terms. Suppose $\cM:= \{M_i\}_{i\in[m]}\subset\C^{2^r\times 2^r}$ is $(C, 2r)$-bounded. Then for any $\eps > 0$, there is a randomized algorithm running in time $\mathrm{poly}(n,m)$ which computes a weight function $w: \mathcal{H} \rightarrow \R_{\geq 0}$ such that $w$ is an $\eps$-sparsifier of $\mathcal{H}$ and the size of $w$ is bounded by $\widetilde{O}_r(\eps^{-2} C^2n^2)$ with probability $1 - 1 / \mathrm{poly}(n)$.
\end{theorem}

\subsection{Sparsifiers for Hamiltonians of Full-Rank Predicates}\label{subsec:full-rank}
As a warm-up to \cref{thm:Nullity1Sparsifier},  we first consider the case in which each matrix has full rank. Suppose $M\in\C^{2^r\times 2^r}$ is a full-rank matrix. Then there exist constants $0 < \lambda\leq \Lambda$ such that $\Lambda\cdot\Id\succeq M\succeq\lambda\cdot\Id$. In particular, for any hyperedge $T\in\cH$, $\Lambda\cdot\Id\succeq M_T\succeq\lambda\cdot\Id$. 

Note that this immediately implies that $N(\alpha; M, n)\leq 2$ for $\alpha:= \lambda/\Lambda\geq\Omega_M(1)$: Indeed, for any $T, T'\in\cH$, we have $\alpha M_T\preceq M_{T'}$. Consequently, \cref{thm:nstarsparsifier} immediately implies that for any $\eps\in(0, 1)$ there exists a $\eps$-sparsifier $w:\cH\to\R_{\geq 0}$ with ${\vert}\supp(w){\vert}\leq O_M(\eps^{-2}n\log n)$.

However, it is possible to construct a $\eps$-sparsifier more directly, just using the ``well-conditioned'' property, along with a simple Matrix Bernstein bound, to get a sparsifier size of $O_M(\eps^{-2}n)$ for any $M$-Hamiltonian $\cH$. 
\begin{theorem}
    If $M\in\C^{2^r\times 2^r}$ is a full-rank matrix, then for any $M$-Hamiltonian $\cH$ on $n$ qubits, there exists a sparsifier $w:\cH\to\R_{\geq 0}$ with ${\vert}\supp(w){\vert}\leq O_M(\eps^{-2}n)$. Furthermore, sampling hyperedges randomly yields such a $\eps$-sparsifier with high probability, i.e. if we define $w(T):= \1(T\in\cG)/p$, where $p:= \Theta_M(n/(\eps^2{\vert}\cH{\vert}))$, and every hyperedge in $\cH$ is retained in $\cG$ with probability $p$ independently, then $w$ is an $\eps$-sparsifier of size $\leq O(\eps^{-2}n)$ with probability $\geq 1 - 1/\poly(n)$.
\end{theorem}
\begin{proof}
    Write $p:= cn/(\eps^2{\vert}\cH{\vert})$ for some large enough constant $c > 0$, and define $w(T):= \1(T\in\cG)/p$. Write $A = \sum_{T\in\cH}M_T, A':= \sum_{T\in\cH}w(T)M_T$, and note that $\E[A'] = A$. $w$ is an $\eps$-sparsifier only if $(1 - \eps)A\preceq A'\preceq(1 + \eps)A$. But note that 
    \[(1 - \eps)A\preceq A'\preceq(1 + \eps)A\iff -\eps A\preceq A' - A\preceq \eps A\iff \] 
    \[-\eps\Id\preceq\sum_{T\in\cH}\left(\frac{\1(T\in\cG)}{p} - 1\right)A^{-1/2}M_TA^{-1/2}\preceq\eps\Id\]  
    \begin{align}
    \label{eq:lev-spars}
        \iff\left\Vert\sum_{T\in\cH}\left(\frac{\1(T\in\cG)}{p} - 1\right)A^{-1/2}M_TA^{-1/2}\right\Vert_{\operatorname{op}}\leq\eps\mper
    \end{align}
    We'll prove \cref{eq:lev-spars} using \cref{fact:matrix-bernstein}. Towards that, note that if $\lambda > 0$ is the smallest eigenvalue of $M$, then $M\succeq\lambda\Id$, and thus $M_T\succeq\lambda\Id$, and thus $A\succeq\lambda\cdot{\vert}\cH{\vert}\cdot\Id$. Similarly, if $\Lambda$ is the largest eigenvalue of $M$, then $M_T\preceq \Lambda\Id$. Consequently, ${\Vert}A^{-1/2}M_TA^{-1/2}{\Vert}_{\op}\leq(\Lambda/\lambda)\cdot (1/{\vert}\cH{\vert})\leq O_M(1/{\vert}\cH{\vert})$. Thus if we write $Y_T = \left(\frac{\1(T\in\cG)}{p} - 1\right)A^{-1/2}M_TA^{-1/2}$ for all $T\in\cH$, then $\max_{T\in\cH}{\Vert}Y_T{\Vert}_{\op}\leq O(\eps^2/n)$ with probability $1$ when $c$ is chosen sufficiently large w.r.t. $M$. Furthermore, $\E[Y_T^2] = \left(\frac{\1(T\in\cG)}{p} - 1\right)(A^{-1/2}M_TA^{-1/2})^2\preceq (A^{-1/2}M_TA^{-1/2})^2/p\preceq O(\eps^2/n)\cdot A^{-1/2}M_TA^{-1/2}$, where the last inequality follows on plugging the value of $p$. Consequently, $\sum_{T\in\cH}\E[Y_T^2]\preceq O(\eps^2/n)\cdot \sum_{T\in\cH}A^{-1/2}M_TA^{-1/2}\preceq O(\eps^2/n)\cdot\Id$. Invoking \cref{fact:matrix-bernstein} now proves the result, and a simple Chernoff bound yields an upper bound on the size of the sparsifier.
\end{proof}
\begin{remark}
\label{rem:different-null0-spars}
    Note that the proof goes through verbatim even if the matrices $M$ acting on different tuples were different, provided they were all full-rank. In that case, if our family of local operators is $\cM\subset\C^{2^r\times 2^r}$, we get an $\eps$-sparsifier of size
    \[O\left(\frac{\sup_{M\in\cM}\lambda_{\max}(M)}{\inf_{M\in\cM}\lambda_{\min}(M)}\cdot\frac{n}{\eps^2}\right)\mper\]
\end{remark}

\subsection{The Nullity $1$ Sparsifier}

We now return to proving \cref{thm:Nullity1Sparsifier}. Since $\cM$ is $(C, 2r)$-bounded, note that $\frac{\sup_{M\in\cM}\lambda_{\max}(M)}{\inf_{M\in\cM}\lambda_{\min}(M)}\leq C^2$, and thus if we write $\cH':= \{(T_i, M_i): M_i\text{ is full-rank}\}$, then \cref{rem:different-null0-spars} allows us to obtain a $O(\eps^{-2} C^2n)$-sized $\eps$-sparsifier for $\cH'$. Since this is lesser than the bound we're shooting for in \cref{thm:Nullity1Sparsifier} anyways, WLOG we assume every matrix in $\cM$ has nullity exactly $1$. \footnote{since after we have sparsified the ``exactly-nullity-$1$'' sub-instance of $\cH$, we can take a union with the full-rank sparsifier without hurting our asymptotics}

Write $\ker(M_i) = \spn\{\phi_i\}$, and let $S_i:= \supp(\phi_i)\subset\{0, 1\}^r$ be the support of $\phi_i$. We call $M_i$ the ``operator acting on $T_i$'', and $\phi_i$ the ``kernel vector associated to $T_i$''.

Recall that we denote the ground state of the Hamiltonian $\cH$ by $\Psi_\cH := \bigcap_{T_i\in\cH}\ker((M_i)_{T_i})$, where $(M_i)_{T_i} = M_i\vert_{T_i}\otimes\Id_{\bar{T_i}}$.

We first prove a statement about the ground states of such nullity $1$ Hamiltonians. We assume that $\cH$ has $n$ qubits, and WLOG assume $\bigcup_{T\in\cH} T = [n]$.
\begin{lemma}
\label{lem:uniquepsi}
    $\dim_\C(\Psi_\cH)\leq 1$, i.e. the dimension of $\Psi_\cH$ as a complex vector space is $\leq 1$.
\end{lemma}
\begin{proof}
    Define 
    \[G_{\cH}:= \{x\in\zo^n: x\vert_{T_i}\in S_i\text{ for all }T_i\in\cH\}\mper\]
    Now suppose $\dim_\C(\Psi_\cH)\geq 1$ (otherwise we're already done), and let $\psi:\zo^n\to\C$ be a non-zero vector in $\Psi_\cH$. Since $\psi\in\Psi_\cH$, $\psi\in\ker((M_i)_{T_i})$ for all $i$, and thus by \cref{eq:etmexpression} we have $\langle\psi_z, M_i\psi_z\rangle = 0$ for all $z\in\zo^{[n]\setminus T_i}, (T_i, M_i)\in\cH$. But $\langle\psi_z, M_i\psi_z\rangle = 0\iff \psi_z\in\ker(M_i)\iff\psi_z\in\spn\{\phi_i\}$.

    Consequently, if for some $x\in\zo^n$ there exists $(T_i, M_i)\in\cH$ such that $x\vert_{T_i}\notin\supp(\phi_i)$, then $\psi(x) = 0$. Consequently, $\supp(\psi):= \{x\in\zo^n:\psi(x)\neq 0\}\seq G_{\cH}$, and thus if ${\vert}G_{\cH}{\vert}\leq 1$, then we're done.

    Now consider the following graph on $G_{\cH}$: For any $x, x'\in G_{\cH}$ for which there exists $T\in\cH$ such that $x\vert_{[n]\setminus T} = x'\vert_{[n]\setminus T}$, we connect $x, x'$ with an edge. 
    
    Now, let $x^{(0)}\in G_{\cH}$ be such that $\psi(x^{(0)})\neq 0$, and let $x^{(1)}\in G_{\cH}\setminus\{x^{(0)}\}$ be any other element. Since $\bigcup_{T\in\cH}T = [n]$, there exists $T\in\cH$ such that $x^{(0)}\vert_{T}\neq x^{(1)}\vert_T$. Fix $S:= \supp\{\phi\}$ where $\phi$ is the kernel vector associated to $T$.
    
    Define $x'$ as $x'\vert_T:= x^{(1)}\vert_T, x'\vert_{[n]\setminus T} := x^{(0)}\vert_{[n]\setminus T}$, and note that $\Delta(x', x^{(1)}) < \Delta(x^{(0)}, x^{(1)})$, where $\Delta(\ast, \ast)$ denotes Hamming distance.

    Now, since $x^{(1)}\in G_{\cH}$, $x'\vert_T = x^{(1)}\vert_T\in S$. Furthermore, since $\psi\in\ker(M_T)$, $\psi_{x^{(0)}\vert_{[n]\setminus T}}\in\C^{2^r}$ is parallel to $\phi$. Consequently, 
    \[\psi(x') = \psi_{x'\vert_{[n]\setminus T}}(x'\vert_{T}) = \psi_{x^{(0)}\vert_{[n]\setminus T}}(x^{(1)}\vert_{T}) = \frac{\phi(x^{(1)}\vert_{T})}{\phi(x^{(0)}\vert_{T})}\cdot\psi_{x^{(0)}\vert_{[n]\setminus T}}(x^{(0)}\vert_{T}) = \frac{\phi(x^{(1)}\vert_{T})}{\phi(x^{(0)}\vert_{T})}\cdot\psi(x^{(0)})\neq 0\mper\]
    Since $\psi(x')\neq 0$, $x'\in G_{\cH}$. Consequently, for any $x^{(0)}\neq x^{(1)}\in G_{\cH}$, there exists $x'\in G_{\cH}$ such that $\Delta(x', x^{(1)}) < \Delta(x^{(0)}, x^{(1)})$. Consequently, $G_{\cH}$ is connected, since Hamming distance can't decrease indefinitely.

    Finally, we prove the uniqueness of $\psi$ (upto scaling), which implies $\dim(\Psi_\cH) = 1$: Indeed, consider any $x\in G_{\cH}\setminus\{x^{(0)}\}$. Then there exists a path $x^{(0)}, x^{(1)}, \ldots, x^{(\ell)}, x^{(\ell + 1)}:= x$ in $G_{\cH}$, where all points in this path are distinct. Let $T^{(i)}\in\cH$ be the tuple associated to the edge $x^{(i)}, x^{(i + 1)}$, i.e. $x^{(i)}\vert_{[n]\setminus T^{(i)}} = x^{(i + 1)}\vert_{[n]\setminus T^{(i)}}$. Let $\phi^{(i)}$ be the kernel vector associated to $T^{(i)}$.
    
    Then note that 
    \[\psi(x) = \psi(x^{(0)})\cdot\prod_{i = 0}^\ell\frac{\phi^{(i)}(x^{(i + 1)}\vert_{T^{(i)}})}{\phi^{(i)}(x^{(i)}\vert_{T^{(i)}})}\mcom\]
    i.e.\ the value of $\psi(x^{(0)})$ uniquely determines $\psi$ on $G_{\cH}$, and thus uniquely determines $\psi$, as desired.
\end{proof}

Using \cref{lem:uniquepsi} above, we now give an algorithm (\cref{alg:buildnullity1Sparsifier}) to construct a sparsifier for nullity-$1$ Hamiltonians. 

Before that, we need to set up some notations and definitions though. We first define the notion of a spanning forest of $\mathcal{H}$:

\begin{definition}\label{def:spanningForest}
    Let $\mathcal{H} = \{T_i:{i \in [m]}\}$ be a hypergraph over $n$ qubits. We say that a subset $F \subseteq [m]$ of the terms of $\mathcal{H}$ is a spanning forest if, for every qubit $v \in [n]$, if there is a term $T_i: i \in [m]$ such that $v \in T_i$, then there is also a term $T_j: j \in F$ such that $v \in T_j$.
\end{definition}

We now also make the following basic observation:

\begin{proposition}\label{prop:buildSpanningForest}
    Let $\mathcal{H} = \{T_i:{i \in [m]}\}$ be a hypergraph over $n$ qubits. Then, $\mathcal{H}$ has a spanning forest $F \subseteq [m]$ with ${\vert}F{\vert} \leq n$.
\end{proposition}

\begin{proof}
    Consider a simple greedy algorithm: we initialize $F = \emptyset$, and iterate through the terms $\{(T_i, M):{i \in [m]}\}$. For each term $T_i$, we check if there is some qubit $v \in [n]$ in $T_i$ which is not yet participating in any term in $\{(T_j, M):{j \in F}\}$. If so, then we add index $i$ to $F$.

    We can observe that this can only happen $n$ times, as each qubit must only be covered a single time.
\end{proof}

We also need the following definitions:
\begin{definition}[Clauses containing a variable]
    Let $\cH$ be a hypergraph on a vertex set $V$. For any $v\in V$, write $\cH_v:= \{T\in\cH:v\in T\}$.
\end{definition}

\begin{definition}
    Let $\cH = \{(T_i, M_i)\}_{i\in [m]}$ be a Hamiltonian on $[n]$, and let $\cH'\seq\cH$ be a subset of $\cH$. We say $(T, M)\in\cH\setminus\cH'$ \emph{dominates} $\cH'$ if there exists $\psi\in\C^{\zo^n}$ such that $\cE_{T, M}(\psi) > 0 = \sum_{(T', M')\in\cH'}\cE_{T', M'}(\psi)$.
\end{definition}

If $\cH = \{(T_i, M_i)\}_{i\in [m]}$ is a Hamiltonian, we will abuse notation and let $\cH_v$ also be a Hamiltonian, corresponding to the clauses containing $v$.

\begin{algorithm}
\caption{ExtractDominatingCover$(\mathcal{H} = \{(T_i, M_i)\}_{i\in [m]})$}\label{alg:extractDominatingCover}
Let $D$ be a spanning forest of $\mathcal{H}$, extracted as in \cref{prop:buildSpanningForest}\;
Let $V = \bigcup_{T\in F}T = \bigcup_{T\in\mathcal{H}}T$\;

\For{$v\in V$}{
    \If{$\exists (T, M)\in\mathcal{H}_v\setminus F$ such that $(T, M)$ dominates $F$}{
        $F\leftarrow F\cup\{(T, M)\}$ \tcp*[r]{If multiple such $(T, M)$ exist, add only one}
    }
}
\Return{$F$.}
\end{algorithm}
For any $\cH$, we call (any) output of \cref{alg:extractDominatingCover} a \emph{Dominating Cover} of $\cH$.

Note that in \cref{alg:extractDominatingCover} since the \texttt{For} loop can go on for at most $n$ steps, and since a spanning forest has $\leq n$ elements (by \cref{prop:buildSpanningForest}), we obtain that a dominating cover has $\leq 2n\leq O(n)$ many elements.

Let $\cH = \{(T_i, M_i)\}_{i\in [m]}$ be a nullity $1$ Hamiltonian on $[n]$.  

\begin{lemma}
\label{lem:dominatingcoverenergy}
    Let $\cH = \{(T_i, M_i)\}_{i\in [m]}$ be a nullity $1$ Hamiltonian on $[n]$, and let $D$ be a dominating cover of $\cH$. Suppose $\cM:= \{M_i\}\subset\C^{2^r\times 2^r}$ is $(C, 2r)$-bounded. Then for any $(T, M)\in\cH\setminus D$, $M_T\preceq C^2\cdot\sum_{(T', M')\in D}M'_{T'}$.
\end{lemma}
\begin{proof}
    By construction (see \cref{alg:extractDominatingCover}), we can write $D = F\cup \bigcup_{v\in[n]}D^{(v)}$, where $F$ is a spanning forest of $\cH$, and $D^{(v)}$ is a set of size $\leq 1$, which, if non-empty, contains a hyperedge $T_v$ containing $v$.

    Now, for every $t\in T$, choose a hyperedge $T^{(t)}\in F$ containing $t$. Such a hyperedge exists because $F$ is a spanning forest. Finally, write $E:= \{T^{(t)}: t\in T\}\cup\bigcup_{t\in T} D^{(t)}$. Note that ${\vert}E{\vert}\leq 2r$, and thus ${\vert}V_E{\vert}\leq 2r^2$, where $V_E := \bigcup_{T\in E}T$. 

    We shall prove the stronger claim that $M_T\preceq C^2\cdot\sum_{(T', M')\in E}M'_{T'} =: A_E$. WLOG we can restrict our attention to the qubits in $V_E$. Note that $M_T, M'_{T'}$ act as $\Id$ on the qubits in $[n]\setminus V_E$, and thus it suffices to prove $\cE_{T, M}(\psi)\leq C^2\cdot\sum_{(T', M')\in E}\cE_{T', M'}(\psi)$ only for unit vectors $\psi\in\ker(M_T)^\perp\cap\C^{\zo^{V_E}}$. 

    We now make cases: 
    
    \parhead{${\vert}\bigcup_{t\in T}D^{(t)}{\vert} = 0$} In this case, for all $(T, M)\in\cH\setminus F$, and for any $\psi\in\C^{\zo^{V_E}}$, if $\cE_{T, M}(\psi) > 0$, then $\sum_{(T', M')\in F}\cE_{T', M'}(\psi) > 0$. Consequently, $\sum_{(T', M')\in F}M'_{T'}$ is a strictly positive operator on $\ker(M_T)^\perp\cap\C^{\zo^{V_E}}$, and thus by $(C, 2r)$-boundedness, for any unit vector $\psi\in\ker(M_T)^\perp\cap\C^{\zo^{V_E}}$, $\sum_{(T', M')\in F}\cE_{T', M'}(\psi)\geq 1/C$. On the other hand, also by $(C, 2r)$-boundedness, for any unit vector $\psi\in\ker(M_T)^\perp\cap\C^{\zo^{V_E}}$, $\cE_{T, M}(\psi)\leq\lambda_{\max}(M)\leq C$, and we're done.

    \parhead{${\vert}\bigcup_{t\in T}D^{(t)}{\vert} \geq 1$} In this case $D^{(t)} = \{R^{(t)}\}$ for some $t\in T$. We claim that $A_E$ is invertible: Indeed, if $\sum_{(T', M')\in F}M'_{T'}$ is invertible, then we're already done. Otherwise, by \cref{lem:uniquepsi}, $(A_E)\vert_{\C^{\zo^{V_E}}}$ has a dimension $1$ null-space, which we denote as $\spn\{\tau\}$. Equivalently, if $\sum_{(T', M')\in F}\cE_{T', M'}(\psi) = 0$, then $\psi\in\spn\{\tau\}$. \footnote{since $\ker(A_E)\seq\ker\left(\sum_{(T', M')\in F}M'_{T'}\right)$. By \cref{lem:uniquepsi} we have $\dim\ker\left(\sum_{(T', M')\in F}M'_{T'}\right)\leq 1$, but we also know that $\dim\ker(A_E) = 1$, and thus $\dim\ker\left(\sum_{(T', M')\in F}M'_{T'}\right) = 1\implies \ker\left(\sum_{(T', M')\in F}M'_{T'}\right) = \ker(A_E)$} On the other hand, since $R^{(t)}$ dominates $F$ (by construction), there exists a unit vector $\nu\in\C^{\zo^{V_E}}$ such that $\cE_{R^{(t)}, M}(\nu) > 0 = \sum_{(T', M')\in F}\cE_{T', M'}(\nu)$, where $M$ is the local operator associated to $R^{(t)}$. But $0 = \sum_{(T', M')\in F}\cE_{T', M'}(\nu)\implies\nu\in\spn\{\tau\} = \ker(A_E)\seq\ker(M_T)\implies\nu\in\ker(M_T)$, leading to a contradiction.
    
    Thus $A_E$ is invertible, and thus by $(C, 2r)$-boundedness, for any unit vector $\psi\in\ker(M_T)^\perp\cap\C^{\zo^{V_E}}$, $\langle\psi, A\psi\rangle\geq 1/C$. At this point, we're done as in the previous case.\qedhere
\end{proof}

With \cref{lem:dominatingcoverenergy} now established, we proceed to our algorithm for building sparsifiers. Roughly speaking, the intuition is that we will recover ``dominating covers''. If we denote such a dominating cover by $D$, the key insight is that, for every other term $T \in \mathcal{H}$, it must be the case (by \cref{lem:dominatingcoverenergy}), that an $1/C^2$ fraction of $T$'s energy is contributed by $D$.

We then repeat this process, peeling off more spanning forests of $\mathcal{H}$, roughly $O_r(\eps^{-2}C^2n)$ many times. At this point, if we denote the collection of spanning forests by $D_1, \dots D_k$, we will be able to guarantee that, for any $\psi \in \C^{2^n}$ and any $T \in \mathcal{H} - D_1 - \dots - D_k$ that $\cE_{D_1 \cup \dots \cup D_k}(\psi) \geq \eps^{-2}n \cdot \cE_{T}(\psi)$. Morally, this means that the ``importance'' of $T$ is quite small, and we will then leverage this to show that random sampling produces a sparsifier of the Hamiltonian $\mathcal{H}$.

With this, we are now ready to present our sparsification algorithm. 

\begin{algorithm}
\caption{Sparsify$(\mathcal{H} = \{(T_i, M_i):{i \in [m]}\}, n, \eps)$, $\cM:= \{M_i\}$ is $(C, 2r)$-bounded}\label{alg:buildnullity1Sparsifier}
\If{$m \leq 100\eps^{-2} C^2 \cdot n^2$}{
 \Return{$w: [m] \rightarrow \{1\}$.}
}
\Else{
$\mathrm{DominatingCovers} = \emptyset$. \\
\For{$j \in [100C^2 \cdot \eps^{-2} n]$}{
Let $D_j$ be a dominating cover of $\{(T_i, M_i):{i \in [m] - \bigcup_{D \in \mathrm{DominatingCovers}} D}\}$, as per \cref{alg:extractDominatingCover}. \\
$\mathrm{DominatingCovers} \leftarrow \mathrm{DominatingCovers} \cup \{D_j\}$. \\
}
Let $S$ be a uniformly random sample (at rate $1/2$) of $[m] - \bigcup_{D \in \mathrm{DominatingCovers}} D$. \label{line:samplenull1}\\
Let $w: \bigcup_{D \in \mathrm{DominatingCovers}} D \rightarrow \{1\}$. \\
Let $w' = 2 \cdot \mathrm{Sparsify}(\{(T_i, M_i):{i \in S}\}, n, \eps)$. \\
\Return{$w \sqcup w'$.}
}
\end{algorithm}

\subsection{Bounding Sparsifier Size}

We first bound the size of the sparsifier returned by \cref{alg:buildnullity1Sparsifier}.

To do this, we make the following basic claim about the number of rounds of sampling that are performed in \cref{alg:buildnullity1Sparsifier}:

\begin{claim}\label{clm:boundNumberRoundsnull1}
    Let $\mathcal{H} = \{(T_i, M_i):{i \in [m]}\}$ be a Hamiltonian over $n$ qubits. Then, when \cref{alg:buildnullity1Sparsifier} is called on $\mathcal{H}$ with parameter $\eps$, the number of recursive calls to $\mathrm{Sparsify}$ is bounded by $O(r \log(n))$.
\end{claim}

\begin{proof}
    First, we can observe that initially $m \leq n^r$, as each term is only of arity $r$. Now, \cref{alg:buildnullity1Sparsifier} operates recursively, where in each recursive round, some set of terms is explicitly removed (in $\mathrm{DominatingCovers}$), and the remaining terms in $[m] - \bigcup_{D \in \mathrm{DominatingCovers}} D$ are sampled at rate $1/2$.

    Now, if we let $S^{(j)}$ denote the argument to the $j^{\mathrm{th}}$ recursive call to $\mathrm{Sparsify}$, the probability that a fixed term $T$ is in $S^{(j)}$ is bounded by $1 / 2^j$ (as any such term must have been sampled in the $j$ rounds leading to this recursive call). 
    
    Thus, after $20r \log(n)$ rounds, the probability that any term $T$ is in the recursive call to sparsify is bounded by $1 / n^{20r}$. Thus, by a simple union bound over all of the terms, we see that 
    \[
    \Pr[{\vert}S^{(2r \log(n))}{\vert} > 0] \leq n^{-19r},
    \]
    and thus the procedure terminates within $O_r(\log(n))$ iterative rounds.
\end{proof}

\begin{claim}\label{clm:sparsifierSizenull1}
    Let $\mathcal{H} = \{(T_i, M_i):{i \in [m]}\}$ be a Hamiltonian over $n$ qubits, and let $w: [m] \rightarrow \R_{\geq 0}$ be the weight function returned by \cref{alg:buildnullity1Sparsifier}, when called on $\mathcal{H}$ with parameter $\eps$. Then, the size of $w$ is bounded by $O_r(\eps^{-2} C^2n^2 \log(n))$.
\end{claim}

\begin{proof}
    By \cref{clm:boundNumberRoundsnull1}, we see that the number of recursive calls made in \cref{alg:buildnullity1Sparsifier} is bounded by $O_r(\log(n))$.

    In each such recursive call, the support size of the sparsifier increases by exactly $\bigcup_{F \in \mathrm{DominatingCovers}} F$. We know that any dominating cover has size $\leq O(n)$, and we store $\leq 100\eps^{-2}C^2 \cdot n$ many such dominating covers. Thus, in each recursive call, the size of the sparsifier increases by $O(\eps^{-2}C^2n^2)$, and taking the union over all of the recursive calls, the total size of the sparsifier is bounded by $O_r(\eps^{-2}C^2n^2 \log(n))$.
\end{proof}

\subsubsection{Sparsifier Correctness}

In this section, we show that \cref{alg:buildnullity1Sparsifier} actually produces \emph{good sparsifiers}. Towards this end, we first prove the following:

\begin{claim}\label{clm:oneRoundSparsifynull1}
    Let $\mathcal{H} = \{(T_i, M_i):{i \in [m]}\}$ be a nullity $1$ Hamiltonian over $n$ qubits. Suppose $\cM:= \{M_i\}\subset\C^{2^r\times 2^r}$ is $(C, 2r)$-bounded. Then, letting $\mathrm{DominatingCovers}$ denote a disjoint collection of $100\eps^{-2}C^2 \cdot n$ many spanning forests, letting $S$ be a uniformly random sample (at rate $1/2$) of $[m] - \bigcup_{D \in \mathrm{DominatingCovers}} F$, letting $w: \bigcup_{D \in \mathrm{DominatingCovers}} D \rightarrow \{1\}$ and $w': S \rightarrow \{2\}$,
    then with probability $1 - 2^{-n}$ over the sample $S$, it is the case that $w \sqcup w'$ is an $\eps$-sparsifier of $\mathcal{H}$.
\end{claim}

\begin{proof}
    First, for each term $(T_i, M_i): i \in \bigcup_{D \in \mathrm{DominatingCovers}} D$, we create multiple random variables, which we denote by $Y_{i, 1}, \dots Y_{i, 100\eps^{-2}C^2 \cdot n}$, where each $Y_{i,j} = \frac{\eps^2 \cdot (M_i)_{T_i} \otimes \Id_{\overline{T_i}}}{100C^2 \cdot n }$ with probability $1$.

    Otherwise, for the terms $(T_i, M_i): i \in [m] - \bigcup_{D \in \mathrm{DominatingCovers}} D$, we define the random variables $Y_i = 2 \cdot (M_i)_{T_i} \otimes \Id_{\overline{T_i}}$ with probability $1/2$, and otherwise is $0$.

    We can observe that 
    \[
    \E\left [\sum_{i \in \bigcup_{D \in \mathrm{DominatingCovers}} D}\sum_{j \in [100\eps^{-2} C^2 \cdot n]}Y_{i,j} + \sum_{i \in [m] - \bigcup_{D \in \mathrm{DominatingCovers}} D} Y_i \right ] = \sum_{i \in [m]} (M_i)_{T_i} \otimes \Id_{\overline{T_i}} = Z.
    \]
    
    Next, we now bound the \emph{importance} of each of our random variables:
    \begin{enumerate}[(1)]
        \item For each $Y_{i,j}$ for $i \in \bigcup_{D \in \mathrm{DominatingCovers}} D$ and $j \in [100\eps^{-2} C^2 \cdot n]$, we see that \[
        Y_{i,j} = \frac{\eps^2 \cdot (M_i)_{T_i} \otimes \Id_{\overline{T_i}}}{100C^2 \cdot n } \preceq \frac{\eps^2}{100C^2 \cdot n} \cdot Z,
        \]
        as $Z$ contains the term $ (M_i)_{T_i} \otimes \Id_{\overline{T_i}}$.
        \item Next, we consider the terms $Y_i$ for $i \in [m] - \bigcup_{D \in \mathrm{DominatingCovers}} D$. Here, we see that 
        \[
        Z \succeq \sum_{D \in \mathrm{DominatingCovers}} \sum_{i \in D} (M_i)_{T_i} \otimes \Id_{\overline{T_i}}.
        \]
        Importantly, because each $D$ is a dominating cover, this means that for each term $T_i: i \in [m] - \bigcup_{D \in \mathrm{DominatingCovers}} D$, by \cref{lem:dominatingcoverenergy}, 
        \begin{align}\label{eq:dominatenull1}
        (M_i)_{T_i} \otimes \Id_{\overline    T_i} \preceq O\left(C^2\cdot\sum_{(M_j, T_j)\in D}(M_j)_{T_j} \otimes \Id_{\overline {T_j}}\right).
        \end{align}
        Thus, we see that 
        \[
        Z \succeq \sum_{D \in \mathrm{DominatingCovers}} \sum_{(M_j, T_j)\in D}(M_j)_{T_j} \otimes \Id_{\overline {T_j}} \succeq {\vert}\mathrm{DominatingCovers}{\vert} \cdot \frac{1}{C^2} \cdot  (M_i)_{T_i} \otimes \Id_{\overline{T_i}}\] 
        \[= \frac{100n}{2\eps^2} \cdot (2(M_i)_{T_i} \otimes \Id_{\overline{T_i}}) = \frac{100n}{\eps^2} \cdot Y_i.
        \]
    \end{enumerate}

In particular, in the language of \cref{fact:matrix-bernstein-restated}, we can set the constant $R = \frac{\eps^2}{100n}$. Thus, \cref{fact:matrix-bernstein-restated} guarantees that with probability $1 - 2 \cdot (2^n \cdot \exp(-100n\eps^2 / 3\eps^2)) \geq 1 - 2^{-n}$ that 
\begin{align}\label{eq:PSDSandwichnull1}
(1 - \eps) Z \preceq \sum_{i \in \bigcup_{D \in \mathrm{DominatingCovers}} D}\sum_{j \in [100\eps^{-2} C^2 \cdot n]}Y_{i,j} + \sum_{i \in [m] - \bigcup_{D \in \mathrm{DominatingCovers}} D} Y_i \preceq (1 + \eps) Z.
\end{align}

In particular, for $\psi \in \C^{2^n}$, we can see that $\langle \psi, Z \psi \rangle = \cE_{\mathcal{H}}(\psi)$.

Next, we can see that our sampling procedure over the terms $Y_i$ is exactly equivalent to giving weight $1$ to all terms in $\bigcup_{D \in \mathrm{DominatingCovers}} D$, and sampling the remaining terms at rate $1/2$ and giving weight $2$ to those that are sampled. 

Thus, letting $S$ be a uniformly random sample (at rate $1/2$) of $[m] - \bigcup_{D \in \mathrm{DominatingCovers}} D$, letting $w: \bigcup_{D \in \mathrm{DominatingCovers}} D \rightarrow \{1\}$ and $w': S \rightarrow \{2\}$,
    then with probability $\geq 1 - 2^{-n}$ over the sample $S$, it is the case that $w \sqcup w'$ is an $\eps$-sparsifier of $\mathcal{H}$, as \cref{eq:PSDSandwichnull1} implies that 
    \[
    \cE_{\mathcal{H},w }(\psi) \in [1\pm\eps] \cE_{\mathcal{H}}(\psi)\mcom
    \]
    as desired.
\end{proof}

The preceding claim shows that \emph{in one single round of sampling}, the accuracy of the sparsifier is within the tolerable $(1 \pm \eps)$ regime. However, we then repeat this procedure, sparsifying part of the sparsifier. 

Thus, we get the following claim:

\begin{claim}\label{clm:accuracynull1}
    Let $\mathcal{H} = \{(T_i, M_i):{i \in [m]}\}$ be a nullity $1$ Hamiltonian over $n$ qubits. Then, letting $w$ denote the returned sparsifier from \cref{alg:buildnullity1Sparsifier}, it is the case that, for every $\psi \in \C^{2^n}$
    \[
    \cE_{\mathcal{H}, w}(\psi) \in (1 \pm \eps)^{100 r \log(n)} \cdot \cE_{\mathcal{H}}(\psi).
    \]
\end{claim}

\begin{proof}
    We consider the execution of \cref{alg:buildnullity1Sparsifier}. We let $D^{(1)}$ denote the collection of terms in the dominating covers in round $1$, and let $S^{(1)}$ denote the resulting sample of the remaining terms in round $1$. In general, we will consider this across all iterations of the above procedure, using $D^{(\ell)}$ and $S^{(\ell)}$ to denote the recovered forests and sampled terms from round $\ell$.

    Likewise, we let $w^{(\ell)}$ denote the weight function which gives weight $2^j$ to terms in $D^{(j)}$ for $j \in [\ell]$, and weight $2^{\ell+1}$ to terms in $S^{(\ell)}$ (and weight $0$ to all other terms). We also let $\widehat{w}^{(\ell)}$ denote the weight function which gives weight $1$ to every term in $F^{(\ell)}$ and weight $2$ to every term in $S^{(\ell)}$ (and weight $0$ to every other term).

    If we let $\mathbf{1}[A]$ for $A \subseteq [m]$ denote the weight function which gives weight $1$ to terms with index in $A$, then we can interpret \cref{clm:oneRoundSparsifynull1} as saying that (with probability $1 - 2^{-n}$), for every $\psi \in \C^{2^n}$,
    \[
     \cE_{\mathcal{H}, \widehat{w}^{(\ell)}}(\psi)\in (1 \pm \eps) \cE_{\mathcal{H}, \mathbf{1}[S^{(\ell-1)}]}(\psi) .
    \]

    Now, we claim that, by induction, 
    \[
    \cE_{\mathcal{H}, w^{(\ell)}}(\psi) \in (1 \pm \eps)^{\ell}\cE_{\mathcal{H}}(\psi).
    \]
    Indeed, the base case follows trivially via \cref{clm:oneRoundSparsifynull1}. For the inductive step, let us suppose the hypothesis holds for $\ell-1$. Then, we know that 
    \[
    \cE_{\mathcal{H}, w^{(\ell-1)}}(\psi) \in (1 \pm \eps)^{\ell-1}\cE_{\mathcal{H}}(\psi).
    \]
    We can now write 
    \begin{align}\label{eq:expandEnergynull1}
    \cE_{\mathcal{H}, w^{(\ell)}}(\psi) = \sum_{j = 1}^{\ell-1} 2^j \cdot \cE_{\mathcal{H}, \mathbf{1}[F^{(j)}]}(\psi) + \left (2^{\ell} \cdot \cE_{\mathcal{H}, \mathbf{1}[F^{(\ell)}]}(\psi) + 2^{\ell+1} \cdot \cE_{\mathcal{H}, \mathbf{1}[S^{(\ell)}]}(\psi) \right ).
    \end{align}
    By \cref{clm:oneRoundSparsifynull1}, we know that 
    \begin{align}\label{eq:singleTermSubnull1}
    \left (2^{\ell} \cdot \cE_{\mathcal{H}, \mathbf{1}[F^{(\ell)}]}(\psi) + 2^{\ell+1} \cdot \cE_{\mathcal{H}, \mathbf{1}[S^{(\ell)}]}(\psi) \right ) = 2^{\ell} \cdot \cE_{\mathcal{H}, \widehat{w}^{(\ell}}(\psi) \in 2^{\ell} \cdot (1 \pm \eps) \cdot \cE_{\mathcal{H}, \mathbf{1}[S^{(\ell-1)}]}(\psi).
    \end{align}
    Plugging \cref{eq:singleTermSubnull1} into \cref{eq:expandEnergynull1}, we see that 
    \[
    \cE_{\mathcal{H}, w^{(\ell)}}(\psi) \in \sum_{j = 1}^{\ell-1} 2^j \cdot \cE_{\mathcal{H}, \mathbf{1}[F^{(j)}]}(\psi) + 2^{\ell} \cdot (1 \pm \eps) \cdot \cE_{\mathcal{H}, \mathbf{1}[S^{(\ell-1)}]}(\psi)
    \]
    \[
    \in (1 \pm \eps) \cdot \left ( \sum_{j = 1}^{\ell-1} 2^j \cdot \cE_{\mathcal{H}, \mathbf{1}[F^{(j)}]}(\psi) + 2^{\ell} \cdot \cE_{\mathcal{H}, \mathbf{1}[S^{(\ell-1)}]}(\psi) \right ) \in (1\pm\eps) \cdot \cE_{\mathcal{H}, w^{(\ell-1)}}(\psi) \in (1\pm \eps)^{\ell} \cdot \cE_{\mathcal{H}}(\psi),
    \]
    as we desire.

    Plugging in \cref{clm:boundNumberRoundsnull1} (and taking a union bound over all of the failure probabilities of the sparsification rounds), we see that with probability $1 - \frac{1}{n^{10}}$, that for every $\psi \in \C^{2^n}$,
    \[
    \cE_{\mathcal{H}, w}(\psi) \in (1 \pm \eps)^{100 r \log(n)} \cdot \cE_{\mathcal{H}}(\psi),
    \]
    as we desire.
\end{proof}

With these claims in place, we can now prove our main theorem.

\begin{proof}[Proof of \cref{thm:Nullity1Sparsifier}]
    The algorithm for computing the sparsifier is exactly \cref{alg:buildnullity1Sparsifier}. By \cref{alg:extractDominatingCover} and \cref{clm:boundNumberRoundsnull1}, we can see that \cref{alg:buildnullity1Sparsifier} runs in time $\mathrm{poly}(m, n)$.

    Now, to get a $(1 \pm \eps)$ sparsifier, we simply invoke \cref{alg:buildnullity1Sparsifier} with accuracy parameter $\frac{\eps'}{200 r \log(n)}$. From \cref{clm:accuracynull1}, this then ensures that with high probability, the resulting weight function does constitute a $(1 \pm \eps)$ sparsifier of $\mathcal{H}$. Lastly, the size bound follows from \cref{clm:sparsifierSizenull1} invoked with $\eps' = \eps / 200r \log(n)$. This concludes the theorem. 
\end{proof}

\section{Sparsifying Quantum MAX-CSPs}\label{sec:spar-qmcsp}

\subsection{Sparsifying Quantum MAX-CUT}

Recall that an instance of quantum MAX-CUT is specified by a \emph{graph} $G = (V, E, w)$ on $n$ vertices (qubits), where each edge $e = (u,v) \in E$ corresponds with a Hamiltonian term operating on qubits $u, v$. In quantum MAX-CUT, each Hamiltonian term is of the form $M{\vert}_{u,v} \otimes \Id_{[n] - \{u,v\}}$, where 
\begin{align}\label{eq:pauliDecomp}
M = \frac{1}{2}\left(\Id{\vert}_1 \otimes \Id{\vert}_2 - X{\vert}_1 \otimes X{\vert}_2 - Y{\vert}_1 \otimes Y{\vert}_2 - Z{\vert}_1 \otimes Z{\vert}_2\right) = v v^{\dagger},
\end{align}
where $v = \braket{01} - \braket{10} \in \C^{2^2}$.

Going forward, we use our hamiltonian $\cH$ to denote 
\[
\cH(G) = \sum_{e=(u,v)\in E} w(e) \cdot  M{\vert}_{\{u,v\}} \otimes \Id{\vert}_{[n] - \{u,v\}}.
\]
We use $\mathrm{OPT}(G)$ to denote the \emph{maximum energy} of the above expression, across all states. I.e., 
\[
\mathrm{OPT}(G)= \max_{\braket{\psi} \in \C^{2^n}: \Vert \psi \Vert_2 = 1} \langle \psi {\vert} \cH {\vert} \psi \rangle.
\]
We have the following proposition: 
\begin{proposition}[See, for instance, \cite{KallaugherP22}]\label{prop:maxCutOPT}
    $\mathrm{OPT}(G) \geq \frac{m}{16}$, where $m$ is the total sum of edge weights in $G$.
\end{proposition}

Importantly, we can make the following observation:
\begin{align}\label{eq:PSDify}
    \hat{M} := M +\Id \succeq \Id,
\end{align}
where now each eigenvalue in $\hat{M}$ is at least $1$. 

Now, we can define 
\[
\widehat{\cH}(G) = \sum_{e=(u,v)\in E} (\hat{M}){\vert}_{\{u,v\}} \otimes \Id{\vert}_{[n] - \{u,v\}}.
\]
We refer to each term in the above expression as $\hat{H}_e$.
We now claim the following:

\begin{claim}\label{clm:hatSparsifier}
    Let $\widehat{\cH} = \sum_{e=(u,v)\in E} w(e) \cdot (\hat{H}_e \otimes \Id{\vert}_{[n] - \{u,v\}})$ over $n$ qubits. Then, for any $\eps > 0$ there exists a subset $L \subseteq E$, ${\vert}L{\vert}\leq O(\eps^{-2} n)$, along with new weights $\tilde{w}: L \rightarrow R_{\geq 0}$ such that for every state $\braket{\psi} \in \C^{2^n}$,
    \[
    \sum_{e = (u,v) \in L} \tilde{w}(e) \cdot\langle \psi {\vert} \hat{H}_e{\vert} \psi \rangle \in (1 \pm \eps) \cdot \langle \psi {\vert} \widehat{\cH} {\vert} \psi \rangle.
    \]
    Moreover, this set $L$ can be found in time polynomial in $m, n$, and $\sum_{e \in L}\tilde{w}(e) = \sum_{e \in E}w(e)$.
\end{claim}

\begin{proof}
Again, we use the notion of the \emph{importance} of each term. We see that 
\[
\mathrm{Importance}(\hat{H}_e) = \max_{\braket{\psi}} \frac{\langle \psi {\vert} \hat{H}_e {\vert} \psi \rangle} {\langle \psi {\vert} \hat{\cH} {\vert} \psi \rangle} \leq \frac{\lambda_{\max}(\hat{H}_e)}{\lambda_{\min}(\hat{\cH})} = p_e\leq O(w_e /m ),
\]
where we have used that, by \cref{eq:PSDify}, $\hat{\cH} \succeq m \cdot \Id$.

By \cref{clm:importanceSamplingLouie}, we then see that if we sample each term with probability $\min(1, p_e \cdot \frac{100n}{\eps^2})$, and give weight to each sampled term which is $\frac{w_e}{\min(1, p_e \cdot \frac{100n}{\eps^2})}$, then with probability $1 - 2^{- \Omega(n)}$, the resulting set of terms $L$, with weights $\tilde{w}$, is a $(1 \pm \eps)$ sparsifier of $\cH$. Likewise, by a simple Chernoff bound, the total weight of sampled terms ${\vert}L{\vert}$ is $\in (1 \pm \eps) \cdot m $.

So, to derive the final condition in the claim statement, we simply run the above procedure with $\eps' = \eps / 10$. This finds $L, \tilde{w}$ which is a $(1 \pm \eps')$ sparsifier such that $\sum_{e \in L}\tilde{w}(e) \in (1 \pm \eps') \cdot m$. We can now re-scale the weights by a constant factor $\gamma$ such that $\sum_{e \in L}\tilde{w}(e)  = m$, while only shifting the energy by a factor of $(1 \pm \eps')$. Thus, our resulting sparsifier has accuracy $(1 \pm \eps')^2 \in (1 \pm \eps)$, as we desire.\qedhere
\end{proof}

Finally, we are now able to prove the following theorem:

\begin{theorem}\label{thm:maxcutSparsifier}
    Let $G = (V, E,w )$ be an instance of quantum MAX-CUT with corresponding hamiltonian $\cH$. Then, for any parameter $\eps > 0$, there exists a set $L \subseteq E$, along with weights $\tilde{w}: L \rightarrow \R_{\geq 0}$ such that, letting $\tilde{G} = (V, L, \tilde{w})$, for any state $\braket{\psi}$ which satisfies 
    \[
    \langle \psi {\vert} \cH(\tilde{G}) {\vert} \psi \rangle \geq (1 - \eps) \cdot \mathrm{OPT}(\tilde{G}),
    \]
    then it must also be the case that 
    \[
    \langle \psi {\vert} \cH(G) {\vert} \psi \rangle \geq (1 - 2\eps) \cdot \mathrm{OPT}(G).
    \]
    Moreover, this sparsifier $L, \tilde{w}$ can be found efficiently, in time polynomial in $m, n$.
\end{theorem}

\begin{proof}
    We let the set $L$ and weights $\tilde{w}$ be exactly those weights returned by \cref{clm:hatSparsifier} when invoked with parameter $\eps' = \eps / 100$. The size of this set $L$ already satisfies the conditions set forth above, and the algorithmic efficiency likewise follows.
    
    All that remains is to establish the correctness. 
    For this, denoting $\tilde{G} = (V, L, \tilde{w})$, we see that for every state $\psi \in \C^{2^n}$ that 
    \[
    \langle \psi {\vert}\hat{H}(\tilde{G}){\vert} \psi \rangle \in (1 \pm \eps') \cdot \langle \psi {\vert}\hat{H}(G){\vert} \psi \rangle.
    \]
    Importantly, we know that
    \begin{align}\label{eq:translateSparsifier}
    \hat{H}(G) = H(G) + 2m \cdot \Id, \quad \quad \hat{H}(\widetilde{G}) \in H(\widetilde{G}) + 2m \cdot \Id,
    \end{align}
    where we have used the fact that $\sum_{e \in L}\tilde{w}(e) = \sum_{e \in E}w(e)$ in \cref{clm:hatSparsifier}.
    Plugging this into the above expression, we see that 
    \[
     \langle \psi {\vert}\hat{H}(\tilde{G}){\vert} \psi\rangle =  \langle \psi {\vert}H(\tilde{G}){\vert} \psi \rangle + 2m.
    \]
    Thus, we obtain that 
    \[
    (1 - \eps') \cdot \langle \psi {\vert}H(G){\vert} \psi \rangle + 2m - 2\eps'm \leq \langle \psi {\vert}H(\tilde{G}){\vert} \psi \rangle + 2m \leq (1 + \eps') \cdot \langle \psi {\vert}H(G){\vert} \psi \rangle + 2m + 2\eps' m.
    \]
    Simplifying, this means that 
    \begin{align}\label{eq:IntermediatePreservation}
    (1 - \eps') \cdot \langle \psi {\vert}H(G){\vert} \psi \rangle - 2\eps'm \leq \langle \psi {\vert}H(\tilde{G}){\vert} \psi \rangle \leq (1 + \eps') \cdot \langle \psi {\vert}H(G){\vert} \psi \rangle +2\eps' m.
    \end{align}

    Now, we have two cases:
    \begin{enumerate}
        \item Consider any state $\braket{\psi}$ such that $\langle \psi {\vert} \cH(\tilde{G}) {\vert} \psi \rangle  \geq \frac{m}{32}$. Then, $2 \eps'm \leq 64 \eps' \langle \psi {\vert} \cH(\tilde{G}) {\vert} \psi \rangle$.

    Thus, by re-writing \cref{eq:IntermediatePreservation}, we see that 
    \[
     (1 - \eps') \cdot \langle \psi {\vert}H(G){\vert} \psi \rangle\leq \langle \psi {\vert}H(\tilde{G}){\vert} \psi \rangle (1 + 64 \eps'),
    \]
    and 
    \[
    \langle \psi {\vert}H(\tilde{G}){\vert} \psi \rangle (1 - 64 \eps') \leq (1 + \eps') \cdot \langle \psi {\vert}H(G){\vert} \psi \rangle.
    \]
    Re-writing the above expressions, we obtain that 
    \[
    \langle \psi {\vert}H(\tilde{G}){\vert} \psi \rangle \in (1 \pm 66 \eps') \cdot \langle \psi {\vert}H(G){\vert} \psi \rangle.
    \]
    \item Consider any state $\braket{\psi}$ such that $\langle \psi {\vert} \cH(G) {\vert} \psi \rangle  \geq \frac{m}{32}$. Then, $2 \eps'm \leq 64 \eps' \langle \psi {\vert} \cH(G) {\vert} \psi \rangle$.

    Thus, by re-writing \cref{eq:IntermediatePreservation}, we see that 
    \[
     (1 - 65\eps') \cdot \langle \psi {\vert}H(G){\vert} \psi \rangle\leq \langle \psi {\vert}H(\tilde{G}){\vert} \psi \rangle,
    \]
    and 
    \[
    \langle \psi {\vert}H(\tilde{G}){\vert} \psi \rangle \leq (1 + 65\eps') \cdot \langle \psi {\vert}H(G){\vert} \psi \rangle.
    \]
    Together, the above expressions imply that 
    \[
    \langle \psi {\vert}H(\tilde{G}){\vert} \psi \rangle \in (1 \pm 66 \eps') \cdot \langle \psi {\vert}H(G){\vert} \psi \rangle.
    \]
    \end{enumerate}

Now, we prove our desired theorem statement: indeed, by \cref{prop:maxCutOPT}, we know that $\mathrm{OPT}(G) \geq \frac{m}{16}$. Thus, by the above expressions, it must also be the case that $\mathrm{OPT}(\tilde{G}) \geq \mathrm{OPT}(G) \cdot (1 - 66 \eps')$. Now, let us consider any state $\braket{\psi}$ which satisfies 
    \[
    \langle \psi {\vert} \cH(\tilde{G}) {\vert} \psi \rangle \geq (1 - \eps) \cdot \mathrm{OPT}(\tilde{G}).
    \]
For this state $\braket{\psi}$, we then know that 
\[
 \langle \psi {\vert} \cH(\tilde{G}) {\vert} \psi \rangle \geq (1 - \eps ) \cdot \mathrm{OPT}(G) \cdot (1 - 66 \eps') \geq \frac{m}{32}.
 \]
Again, using our multiplicative error bound in this regime, we then know that 
\[
\langle \psi {\vert} \cH(G) {\vert} \psi \rangle \geq \frac{1}{1 + 66 \eps'} \cdot (1 - \eps ) \cdot \mathrm{OPT}(G) \cdot (1 - 66 \eps').
\]
By choosing $\eps'$ to be sufficiently small (say $\eps / 200$), we then obtain our desired bound that 
\[
\langle \psi {\vert} \cH(G) {\vert} \psi \rangle \geq (1 - 2\eps)  \mathrm{OPT}(G).
\]
    This concludes the theorem.
\end{proof}

In particular, observe that the sparsification procedure in \cref{thm:maxcutSparsifier} is simply a uniformly random sample of the terms. This, in turn, is immediately implementable in insertion-only streams via reservoir sampling (see, for instance \cite{li1994reservoir}). 

In this insertion-only streaming setting, the algorithms operate on a stream of edges, where in each time step, a single new edge is inserted. The resulting graph (the accumulation of all the edges) is then denoted by $G$.

In this setting, we get the following corollary:

\begin{corollary}\label{cor:streamingQuantumMaxCut}
    There is an algorithm, which, for any $\eps > 0$, and for any stream of edges defined over a graph of $n$ vertices, produces a set $L \subseteq E$, ${\vert}L{\vert}\leq O(\eps^{-2} n)$, along with weights $\tilde{w}: L \rightarrow \R_{\geq 0}$ such that, for any state $\braket{\psi}$ which satisfies 
    \[
    \langle \psi {\vert} \cH(\tilde{G}) {\vert} \psi \rangle \geq (1 - \eps) \cdot \mathrm{OPT}(\tilde{G}),
    \]
    then it must also be the case that 
    \[
    \langle \psi {\vert} \cH(G) {\vert} \psi \rangle \geq (1 - 2\eps) \cdot \mathrm{OPT}(G).
    \]
    Moreover, this sparsifier $L, \tilde{w}$ can be found efficiently, in time $\widetilde{O}(\max(m, n))$, and space $\widetilde{O}(\eps^{-2} n)$.
\end{corollary}

\begin{remark}
    This settles the ``information-theoretic'' complexity of QMC in the insertion-only streaming setting, answering an open question of \cite{KallaugherP22}. Indeed, one must only construct the aforementioned sparsifier of the underlying graph, and then solve QMC on this sparsifier. The total space complexity of the streaming algorithm for building this sparsifier is bounded by $\widetilde{O}(\eps^{-2} n)$ (classical) bits of space.
\end{remark}

\subsection{Sparsifying Quantum MAX-CSP}

The previous section can be generalized to \emph{any} quantum MAX-CSP problem, provided it satisfies two conditions. In particular, we assume we are working with an $r$-local Hamiltonian with interaction hypergraph $E$:
\[
\cH = \sum_{e \in E} H_e,
\]
where each $H_e = M_e \otimes \Id_{[n] - e}$, where $M_e$ is a PSD matrix with maximum eigenvalue $O(1)$.

We use $\mathrm{OPT}(G)$ to denote the \emph{maximum energy} of the above expression, across all states. I.e., 
\[
\mathrm{OPT}(G)= \max_{\braket{\psi} \in \C^{2^n}: \Vert \psi \Vert_2 = 1} \langle \psi {\vert} \cH {\vert} \psi \rangle.
\]
Our first assumption is the following:
\begin{proposition}\label{prop:maxSATOPT}
    $\mathrm{OPT}(\cH) \geq \Omega(m)$, where $m$ is the total number of terms in $\cH$.
\end{proposition}

Importantly, we can make the following observation:
\begin{align}\label{eq:PSDifySAT}
    \hat{M}_e := M_e +\Id \succeq \Id,
\end{align}
where now each eigenvalue in $\hat{M}_e$ is at least $1$. 

Now, we can define 
\[
\widehat{\cH}(G) = \sum_{e=(u,v)\in E} (\hat{M}_e) \otimes \Id{\vert}_{[n] - \{u,v\}}.
\]
We refer to each term in the above expression as $\hat{H}_e$.
We now claim the following:

\begin{claim}\label{clm:hatSparsifierSAT}
    Let $\widehat{\cH} = \sum_{e=(u,v)\in E} w(e) \cdot (\hat{H}_e \otimes \Id{\vert}_{[n] - \{u,v\}})$ over $n$ qubits. Then, for any $\eps > 0$ there exists a subset $L \subseteq E$, ${\vert}L{\vert}\leq O(\eps^{-2} n)$, along with new weights $\tilde{w}: L \rightarrow \R_{\geq 0}$ such that for every state $\braket{\psi} \in \C^{2^n}$,
    \[
    \sum_{e = (u,v) \in L} \tilde{w}(e) \cdot\langle \psi {\vert} \hat{H}_e{\vert} \psi \rangle \in (1 \pm \eps) \cdot \langle \psi {\vert} \widehat{\cH} {\vert} \psi \rangle.
    \]
    Moreover, this set $L$ can be found in time polynomial in $m, n$, and $\sum_{e \in L}\tilde{w}(e) = m$.
\end{claim}

\begin{proof}
Again, we use the notion of the \emph{importance} of each term. We see that 
\[
\mathrm{Importance}(\hat{H}_e) = \max_{\braket{\psi}} \frac{\langle \psi {\vert} \hat{H}_e {\vert} \psi \rangle} {\langle \psi {\vert} \hat{\cH} {\vert} \psi \rangle} \leq \frac{\lambda_{\max}(\hat{H}_e)}{\lambda_{\min}(\hat{\cH})} = p_e\leq O(w_e /m ),
\]
where we have used that, by \cref{eq:PSDifySAT}, $\hat{\cH} \succeq m \cdot \Id$.

By \cref{clm:importanceSamplingLouie}, we then see that if we sample each term with probability $\min(1, p_e \cdot \frac{100n}{\eps^2})$, and give weight to each sampled term which is $\frac{w_e}{\min(1, p_e \cdot \frac{100n}{\eps^2})}$, then with probability $1 - 2^{- \Omega(n)}$, the resulting set of terms $L$, with weights $\tilde{w}$, is a $(1 \pm \eps)$ sparsifier of $\cH$. Likewise, by a simple Chernoff bound, the total weight of sampled terms ${\vert}L{\vert}$ is $\in (1 \pm \eps) \cdot m $.

So, to derive the final condition in the claim statement, we simply run the above procedure with $\eps' = \eps / 10$. This finds $L, \tilde{w}$ which is a $(1 \pm \eps')$ sparsifier such that $\sum_{e \in L}\tilde{w}(e) \in (1 \pm \eps') \cdot m$. We can now re-scale the weights by a constant factor $\gamma$ such that $\sum_{e \in L}\tilde{w}(e)  = m$, while only shifting the energy by a factor of $(1 \pm \eps')$. Thus, our resulting sparsifier has accuracy $(1 \pm \eps')^2 \in (1 \pm \eps)$, as we desire.
\end{proof}

Finally, we are now able to prove the following theorem:

\begin{theorem}
    Let $\cH$ be an instance of quantum MAX-CSP with $m$ terms, with $\mathrm{OPT}(\cH) \geq \Omega(m)$. Then, for any parameter $\eps > 0$, there exists a set $L \subseteq E$, along with weights $\tilde{w}: L \rightarrow \R_{\geq 0}$ such that, letting $\tilde{\cH}$ denote the instance $\sum_{e \in L} \tilde{w}(e) \cdot H_e$,  for any state $\braket{\psi}$ which satisfies 
    \[
    \langle \psi {\vert} \tilde{\cH} {\vert} \psi \rangle \geq (1 - \eps) \cdot \mathrm{OPT}(\tilde{\cH}),
    \]
    then it must also be the case that 
    \[
    \langle \psi {\vert} \cH {\vert} \psi \rangle \geq (1 - 2\eps) \cdot \mathrm{OPT}(\cH).
    \]
    Moreover, this sparsifier $L, \tilde{w}$ can be found efficiently, in time polynomial in $m, n$.
\end{theorem}

\begin{proof}
    We let the set $L$ and weights $\tilde{w}$ be exactly those weights returned by \cref{clm:hatSparsifierSAT} when invoked with parameter $\eps' = \eps / \gamma$, where $\gamma$ will be a sufficiently large constant (depending on the constant in \cref{prop:maxSATOPT}). In particular, $\gamma$ will be sufficiently large such that $\mathrm{OPT}(\cH) \geq \frac{10m}{\gamma}$. The size of this set $L$ already satisfies the conditions set forth above, and the algorithmic efficiency likewise follows.
    
    All that remains is to establish the correctness. 
    For this, we see that for every state $\psi \in \C^{2^n}$ that 
    \[
    \langle \psi {\vert}\widetilde{\hat{\cH}}{\vert} \psi \rangle \in (1 \pm \eps') \cdot \langle \psi {\vert}\hat{\cH}{\vert} \psi \rangle.
    \]
    Importantly, we know that
    \begin{align}\label{eq:translateSparsifierSAT}
    \hat{\cH} = \cH + m \cdot \Id, \quad \quad \widetilde{\hat{\cH}} = \tilde{\cH} + 2\cdot m \cdot \Id,
    \end{align}
    where we have used the fact that $\sum_{e \in L}\tilde{w}(e) = \sum_{e \in E}w(e)$ in \cref{clm:hatSparsifierSAT}.
    Plugging this into the above expression, we see that 
    \[
     \langle \psi {\vert}\widetilde{\hat{\cH}}{\vert} \psi\rangle =  \langle \psi {\vert}\widetilde{\cH}{\vert} \psi \rangle + m.
    \]
    Thus, we obtain that 
    \[
    (1 - \eps') \cdot \langle \psi {\vert}\cH{\vert} \psi \rangle + 2m - 2\eps'm \leq \langle \psi {\vert}\tilde{\cH}{\vert} \psi \rangle + 2m \leq (1 + \eps') \cdot \langle \psi {\vert}\cH{\vert} \psi \rangle + 2m + 2\eps' m.
    \]
    Simplifying, this means that 
    \begin{align}\label{eq:IntermediatePreservationSAT}
    (1 - \eps') \cdot \langle \psi {\vert}\cH{\vert} \psi \rangle - 2\eps'm \leq \langle \psi {\vert}\tilde{\cH}{\vert} \psi \rangle \leq (1 + \eps') \cdot \langle \psi {\vert}\cH{\vert} \psi \rangle +2\eps' m.
    \end{align}

    Now, we have two cases:
    \begin{enumerate}
        \item Consider any state $\braket{\psi}$ such that $\langle \psi {\vert} \tilde\cH {\vert} \psi \rangle  \geq \frac{m}{\gamma}$. Then, $2 \eps'm \leq 2\gamma \eps' \langle \psi {\vert} \tilde \cH {\vert} \psi \rangle$.

    Thus, by re-writing \cref{eq:IntermediatePreservationSAT}, we see that 
    \[
     (1 - \eps') \cdot \langle \psi {\vert}\cH{\vert} \psi \rangle\leq \langle \psi {\vert}\tilde \cH{\vert} \psi \rangle (1 + 2\eps'\gamma ),
    \]
    and 
    \[
    \langle \psi {\vert}\tilde \cH{\vert} \psi \rangle (1 - 2\eps'\gamma) \leq (1 + \eps') \cdot \langle \psi {\vert}\cH{\vert} \psi \rangle.
    \]
    Re-writing the above expressions, we obtain that 
    \[
    \langle \psi {\vert}\tilde \cH{\vert} \psi \rangle \in (1 \pm \eps')(1 \pm 3\eps' \gamma) \cdot \langle \psi {\vert}H(G){\vert} \psi \rangle.
    \]
    \item Consider any state $\braket{\psi}$ such that $\langle \psi {\vert} \cH {\vert} \psi \rangle  \geq \frac{m}{\gamma}$. Then, $2 \eps'm \leq 2\gamma \eps' \langle \psi {\vert}\cH{\vert} \psi \rangle$.

    Thus, by re-writing \cref{eq:IntermediatePreservationSAT}, we see that 
    \[
     (1 - \eps' - 2\eps' \gamma) \cdot \langle \psi {\vert}\cH{\vert} \psi \rangle\leq \langle \psi {\vert}\tilde \cH{\vert} \psi \rangle,
    \]
    and 
    \[
    \langle \psi {\vert}\tilde \cH{\vert} \psi \rangle \leq (1 + \eps' + 2\eps' \gamma ) \cdot \langle \psi {\vert}\cH{\vert} \psi \rangle.
    \]
    Re-writing the above expressions, we obtain that 
    \[
    \langle \psi {\vert}\tilde \cH{\vert} \psi \rangle \in (1 \pm \eps')(1 \pm 3\eps' \gamma) \cdot \langle \psi {\vert}H(G){\vert} \psi \rangle.
    \]
    \end{enumerate}

Now, we prove our desired theorem statement: indeed, by \cref{prop:maxSATOPT} and our choice of $\gamma$, we know that $\mathrm{OPT}(G) \geq \frac{10m}{\gamma}$. Thus, by the above expressions, it must also be the case that $\mathrm{OPT}(\tilde \cH) \geq \mathrm{OPT}(\cH) \cdot (1 - \eps')(1 - 3 \eps' \gamma)$. Now, let us consider any state $\braket{\psi}$ which satisfies 
    \[
    \langle \psi {\vert} \tilde \cH {\vert} \psi \rangle \geq (1 - \eps) \cdot \mathrm{OPT}(\tilde \cH).
    \]
For this state $\braket{\psi}$, we then know that 
\[
 \langle \psi {\vert} \tilde \cH {\vert} \psi \rangle \geq (1 - \eps ) \cdot \mathrm{OPT}(\cH) \cdot (1 - \eps')(1 - 3 \eps' \gamma) \geq \frac{5m}{\gamma}.
 \]
Again, using our multiplicative error bound in this regime, we then know that 
\[
\langle \psi {\vert} \cH {\vert} \psi \rangle \geq (1 - \eps ) \cdot \mathrm{OPT}(\cH) \cdot (1 - \eps')(1 - 3 \eps' \gamma) \cdot \frac{1}{(1 + \eps')(1 + 3 \eps' \gamma)}.
\]
By choosing $\eps'$ to be sufficiently small (say $\eps / 200\gamma$), we then obtain our desired bound that 
\[
\langle \psi {\vert} \cH {\vert} \psi \rangle \geq (1 - 2\eps)  \mathrm{OPT}(\cH).
\]
    This concludes the theorem.
\end{proof}

\section{Tensor Product Lower Bounds}\label{sec:tensor-lb}
Consider predicate matrices $M = M_1\otimes M_2\otimes \cdots\otimes M_r$, where $M_i\in\C^{2\times 2}$ are non-zero PSD matrices for all $i\in[r]$, to be viewed as single-qubit operators. Consequently, if $\cH$ is a $r$-partite $r$-uniform hypergraph with the parts being $V_1, \ldots, V_r$, and if we have a hyperedge $T = (t_1, \ldots, t_r)\in\cH$ with $t_i\in V_i$ for all $i\in[r]$, then in the Hamiltonian $M_T$, the matrix $M_i$ acts on the qubit represented by $t_i$ for all $i\in[r]$. 

\begin{claim}
\label{claim:tensor-NRD-construction}
    Let $M = M_1\otimes M_2\otimes \cdots\otimes M_r$ be a predicate matrix where $M_1, \ldots, M_r\in\C^{2\times 2}$ are non-zero PSD matrices that are \textbf{not} full-rank. Let $\cH$ be a $r$-partite $r$-uniform hypergraph. Then for every hyperedge $T\in\cH$, there exists a state $\psi_T$ such that $\cE_{T, M}(\psi_T) > 0$, but $\cE_{T', M}(\psi_{T}) = 0$ for all $T'\in\cH\setminus\{T\}$. 
\end{claim}
\begin{proof}
    For all $i \in [r]$, since $M_i$ is a non-zero PSD matrix which is not full-rank, there exist unit vectors $\tau_i, \pi_i\in\C^2$ such that $\langle\tau_i,M_i\tau_i\rangle = {\Vert}M_i{\Vert}_{\operatorname{op}} > 0$, while $\langle\pi_i,M_i\pi_i\rangle = 0$. Now, let the parts of $\cH$ be $V_1, \ldots, V_r$, and for any $T = (t_1, \ldots, t_r)\in\cH\seq\bigtimes_{i = 1}^rV_i$, consider the vector $\psi_T := \bigotimes_{t\in V}\phi_t$, where for any $t\in V = \bigsqcup_{i = 1}^r V_i$, we define
    \[\phi_t:= \begin{cases}
        \tau_p & \text{if }t = t_p,\\
        \pi_p & \text{if }t\in V_p\setminus\{t_p\}
    \end{cases}\mcom\]
    where $p\in[r]$ is the index for which $t\in V_p$. Then note that for any $T' = (t'_1, \ldots, t'_r)\in\cH$, we have
    \[\cE_{T', M}(\psi_T) = \langle\psi_T, (M\vert_{T'}\otimes\Id_{\bar{T'}})\psi_T\rangle = \prod_{i = 1}^r\langle\phi_{t'_i}, M_i\phi_{t'_i}\rangle = \prod_{i = 1}^r{\Vert}M_i{\Vert}_{\operatorname{op}}\cdot\prod_{i = 1}^r\1(t'_i = t_i)\] 
    \[= \prod_{i = 1}^r{\Vert}M_i{\Vert}_{\operatorname{op}}\cdot\1(T' = T)\mcom\]
    as desired.
\end{proof}
\begin{theorem}
\label{claim:tensor-SPR}
    Let $M = M_1\otimes M_2\otimes \cdots\otimes M_r$ be a predicate matrix where $M_1, \ldots, M_r\in\C^{2\times 2}$ are non-zero PSD matrices that are \textbf{not} full-rank. Then $\SPR(M, n, \eps)\geq\NRD(M, n) \geq (\lfloor n/r\rfloor)^r\geq\Omega_r(n^r)$ for any $\eps\in(0, 1)$.
\end{theorem}

\begin{proof}
    Let $\cH$ be any $r$-partite $r$-uniform hypergraph. By \cref{claim:tensor-NRD-construction}, for every $T\in\cH$, there exists $\psi_T$ such that $\cE_{T, M}(\psi_T) > 0$, while $\cE_{T', M}(\psi_T) = 0$ for all $T'\in\cH\setminus\{T\}$. Thus, we claim that for any $\eps$-sparsifier $w:\cH\to\R_{\geq 0}$, we must have $\supp(w) = \cH$: Indeed, if $w:\cH\to\R_{\geq 0}$ is an $\eps$-sparsifier with $\cH\ni T\notin\supp(w)$, then we have $\cE_{\cH}(\psi_T) = \cE_{T, M}(\psi_T) > 0$, while 
    \[\cE_{\cH, w}(\psi_T) = \sum_{T'\in\cH}w(T', M)\cE_{T', M}(\psi_T) = \sum_{T'\in\cH\setminus\{T\}}w(T', M)\cE_{T', M}(\psi_T) = 0\mcom\]
    contradicting the definition of an $\eps$-sparsifier. Thus ${\vert}w^{-1}(\R_{>0}){\vert} = {\vert}\supp(w){\vert} = {\vert}\cH{\vert}$. Now, if $\cH$ is constrained to have $n$ vertices, then by choosing $\cH$ to have $r$ parts, each of size $\lfloor n/r\rfloor$ or $\lceil n/r\rceil$, we can ensure ${\vert}\cH{\vert}\geq(\lfloor n/r\rfloor)^r$, as desired.
\end{proof}

\begin{example}
As an example of a predicate covered by \cref{claim:tensor-SPR}, consider $\AND$ on $r$ bits Suppose we have the $\AND$ predicate on $r$ bits, i.e. $\AND(x_1, \ldots, x_r) := \bigwedge_{i = 1}^r\ell_i$, where $\ell_i\in\{x_i, \neg x_i\}$. By abuse of notation, let $\AND\in\C^{2^r\times 2^r}$ also denote the matrix representing this predicate. Then note that $\AND = \bigotimes_{i = 1}^r M_i$, where $M_i = \begin{bmatrix}
            0 & 0\\
            0 & 1
        \end{bmatrix}$ if $\ell_i = x_i$, and $\begin{bmatrix}
            1 & 0\\
            0 & 0
        \end{bmatrix}$ otherwise. Here we assume $\begin{bmatrix}
            1\\ 
            0
        \end{bmatrix}\in\C^2$ denotes the qubit $\braket{0}$, and $\begin{bmatrix}
            0\\ 
            1
        \end{bmatrix}\in\C^2$ denotes $\braket{1}$. Consequently, \cref{claim:tensor-SPR} implies that $\SPR(\AND, n, \eps)\geq\Omega_r(n^r)$, a fact well-known from the classical theory of CSP sparsification~\cite{FiltserK17,KhannaPS25}.
\end{example}

\section{Characterizing the Non-Redundancy of All $2$-Qubit Hamiltonians}\label{sec:2-qubit}

In this section, we characterize the non-redundancy of all nonzero 2-qubit predicate matrices $M \in \C^{4 \times 4}$. The characterization is surprisingly simple.

\begin{theorem}\label{thm:2-qubit}
Let $M \in \C^{4 \times 4}$ be a nonzero predicate matrix. Then $\NRD(M, n) = \Theta(n)$ unless $M = M_1 \otimes M_2$ where unary predicate matrices $M_1, M_2 \in \C^{2\times 2}$ have rank one, in which case $\NRD(M, n) = \Theta(n^2)$.
\end{theorem}

Toward proving \Cref{thm:2-qubit}, we first handle a few easy cases.

\begin{enumerate}[(1)]
\item We always have that $\NRD(M, n) \leq 2 \binom{n}{2} \leq O(n^2)$, as a (simple) directed graph on $n$ vertices has at most $2 \binom{n}{2}$ edges.

\item We always have that $\NRD(M, n) \geq \Omega(n)$ by \cref{lem:nalphalowerbound}, as a Hamiltonian $\cH$ consisting of $\lfloor n/2\rfloor$ disjoint edges is always non-redundant (as our assignment $\psi \in \mathbb C^{2^n}$ uses separate qubits to satisfy / not-satisfy each constraint). 

\item If $M = M_1 \otimes M_2$, where $M_1$ and $M_2$ have rank 1, by applying \cref{claim:tensor-SPR} in the case that $r=2$, we get that $\NRD(M, n) \geq \Omega(n^2)$. 
\end{enumerate}

Thus, to prove \Cref{thm:2-qubit}, it suffices to show that $\NRD(M, n) \leq O(n)$ whenever $M$ cannot be expressed as the tensor product of two rank 1 predicate matrices. To do this, we study the \emph{automorphisms} of states $\psi$ in the ground state $\Psi_{\cH}$ of any sufficiently connected Hamiltonian $\cH$.

\paragraph{Automorphisms.} Given a permutation $\pi \in S_n$ (the symmetric group acting on $[n]$), we can define a unitary matrix $U_{\pi} \in \C^{2^n \times 2^n}$ which maps any qubit ${\vert}x\rangle$ with $x \in \{0,1\}^n$ to ${\vert}\pi(x)\rangle$, where $\pi(x)_{\pi(i)} := x_{i}$ for all $i \in [n]$. For example, if $x$ is the indicator vector of the $i^{\mathrm{th}}$ bit, then $\pi(x)$ is the indicator vector of the $\pi(i)^{\mathrm{th}}$ bit. One can verify for any $\pi, \tau \in S_n$, we have that $U_{\pi} \circ U_{\tau} = U_{\pi \circ \tau}$. Given a state $\psi \in \C^{2^n}$, we define its automorphism group $\Aut(\psi)$ as follows:
\[
    \Aut(\psi) = \{\pi \in S_n : U_{\pi} \psi = \psi\}.
\]
We note that $\Aut(\psi)$ is always a subgroup of $S_n$. Given a Hamiltonian $\cH$, we define $\Aut(\cH)$ to be the greatest common automorphism group of all states $\psi \in \Psi_{\cH}$. In other words,
\[
    \Aut(\cH) := \bigcap_{\psi \in \Psi_{\cH}} \Aut(\psi).
\]
The key technical lemma we need toward proving \cref{thm:2-qubit} is understanding how the structure of $\Span(M) \subseteq \C^{4}$ impacts $\Aut(\cH)$.

\subsection{The Nonsingular Case}

In this section, we handle the (general) case in which $\Span(M)$ contains what we call a \emph{nonsingular} vector.

\begin{definition}
We say that $u \in \Span(M)$ are \emph{nonsingular} if we have that
\[
    \det \begin{pmatrix}
    u(00) & u(01)\\
    u(10) & u(11)
    \end{pmatrix} \neq 0.
\]
Otherwise, we say that $u$ is singular.
\end{definition}

We now prove our key technical lemma.

\begin{lemma}\label{lem:nonsingular-aut}
Let $M \in \C^{4 \times 4}$ be a predicate matrix with $u \in \Span(M)$ non-singular, and let $\cH$ be an $M$-Hamiltonian. Let $i,j,k \in [n]$ be distinct indices. If $(i,j), (i,k) \in \cH$, then the permutation $(jk)$ is an element of $\Aut(\cH)$.
\end{lemma}

\begin{proof}
Without loss of generality, we may assume that $T_{12} := (1,2), T_{13} := (1,3) \in \cH$, and we seek to prove that $(23) \in \Aut(\cH)$. Fix an arbitrary state $\psi \in \Psi_{\cH}$. To show that $(23) \in \Aut(\psi)$, it suffices to show $(23) \in \Aut(\psi_x)$ for all $x \in \{0,1\}^T$ where $T = \{4, \hdots, n\}$. This is equivalent to showing that $\psi_x(001) = \psi_x(010)$ and $\psi_x(101) = \psi_x(110)$.

Since $\psi \in \Psi_{\cH}$, we have that the following equations are true:
\begin{align*}
(M{\vert}_{T_{12}} \otimes \Id{\vert}_3) \psi_x &= 0\\
(M{\vert}_{T_{13}} \otimes \Id{\vert}_2) \psi_x &= 0.
\end{align*}
Since $u \in \Span(M)$, we then have the following equations.
\begin{align}
u(00)\psi_x(000) + u(01)\psi_x(010) + u(10)\psi_x(100) + u(11)\psi_x(110) &= 0\label{eq:u0}\\
u(00)\psi_x(001) + u(01)\psi_x(011) + u(10)\psi_x(101) + u(11)\psi_x(111) &= 0\label{eq:u1}\\
u(00)\psi_x(000) + u(01)\psi_x(001) + u(10)\psi_x(100) + u(11)\psi_x(101) &= 0\label{eq:u2}\\
u(00)\psi_x(010) + u(01)\psi_x(011) + u(10)\psi_x(110) + u(11)\psi_x(111) &= 0\label{eq:u3}
\end{align}
Now, if we consider the expression $u(10)(\text{\cref{eq:u0}}) + u(11)(\text{\cref{eq:u1}}) - u(10)(\text{\cref{eq:u2}}) - u(11)(\text{\cref{eq:u3}})$, we can deduce that
\begin{align*}
\det \begin{pmatrix} u(00) & u(01)\\
u(10) & u(10)
\end{pmatrix}\left(\psi_x(001) - \psi_x(010)\right) = 0 \implies \psi_x(001) = \psi_x(010).
\end{align*}
Likewise, if we consider $u(01)(\text{\cref{eq:u0}}) + u(00)(\text{\cref{eq:u1}}) - u(01)(\text{\cref{eq:u2}}) - u(00)(\text{\cref{eq:u3}})$, we can deduce that
\begin{align*}
\det \begin{pmatrix} u(00) & u(01)\\
u(10) & u(10)
\end{pmatrix}\left(\psi_x(101) - \psi_x(110)\right) = 0 \implies \psi_x(101) = \psi_x(110).
\end{align*}
Thus, $(23) \in \Aut(\psi_x)$ so $(23) \in \Aut(\cH)$.
\end{proof}

Given a graph $\cG$ on vertex set $[n]$, we say that a sequence of edges $e_1, \hdots, e_{2k} \in \cG$ is a \emph{bipartite cycle} if there exist (distinct) indices $i_1, \hdots, i_k, j_1, \hdots, j_k \in [n]$ such that
\begin{align}
e_1 = (i_1, j_1), e_2 = (i_1, j_2), e_3 = (i_2, j_2), \hdots, e_{2k-2} = (i_{k-1}, j_k), e_{2k-1} = (i_k, j_k), e_{2k} = (i_k, j_1).\label{eq:bc}
\end{align}

\begin{lemma}\label{lem:bipartite-cycle-redundant}
Assume that $u \in \Span(M)$ is nonsingular. If an $M$-Hamiltonian $\cH$ has a bipartite cycle, then $\cH$ is redundant.
\end{lemma}

\begin{proof}
Assume without loss of generality that \cref{eq:bc} are edges of $\cH$. We claim that $\Psi_{\cH} = \Psi_{\cH \setminus e_{2k}}$. Clearly $\Psi_{\cH} \subseteq \Psi_{\cH \setminus e_{2k}}$, so it suffices to prove that $\Psi_{\cH \setminus e_{2k}} \subseteq \Psi_{\cH}$. Consider any $\psi \in \Psi_{\cH \setminus e_{2k}}$. By invoking \cref{lem:nonsingular-aut} on the pairs of edges $(e_1, e_2), \hdots, (e_{2k-3}, e_{2k-2})$, we have that $(j_1j_2), \hdots, (j_{k-1}j_k) \in \Aut(\psi)$. By taking a suitable product of these automorphisms, we may deduce that $(j_1j_k) \in \Aut(\psi)$. Since $\psi \in \Psi_{e_{2k-1}}$, we have that
\[
    \psi = U_{(j_1j_k)}\psi \in U_{(j_1j_k)}\Psi_{e_{2k-1}} = \Psi_{e_{2k}}.
\]
Thus, $\psi \in \Psi_{\cH \setminus e_{2k}} \cap \Psi_{e_{2k}} = \Psi_{\cH}$, as desired.
\end{proof}

\begin{lemma}\label{lem:nonsingular-NRD}
If $u \in \Span(M)$ is nonsingular, then $\NRD(M, n) \leq O(n)$.
\end{lemma}

\begin{proof}
Randomly color every vertex of $\cH$ uniformly at random either red or blue. Let $\cH'$ be the subset of edges that direct from a red vertex to a blue vertex. By \cref{lem:bipartite-cycle-redundant}, we have that $\cH'$ lacks any cycles, so $\cH'$ has at most $n-1$ edges. Since each edge of $\cH$ is included in $\cH'$ with probability $1/4$, we have that there exists a choice of $\cH'$ with at least ${\vert}\cH{\vert}/4$ edges. Thus, ${\vert}\cH{\vert} \leq 4(n-1) \leq O(n)$, as desired.
\end{proof}

\subsection{The Singular Case}

We still need to classify all predicate matrices $M$ such that every $u \in \Span(M)$ is singular. We start with the following algebraic description of such $M$.

\begin{claim}
If every $u \in \Span(M)$ is singular, then $M = M_1 \otimes M_2$, where $M_1, M_2 \in \C^{2\times 2}$ are unary predicate matrices with at least one of $M_1$ or $M_2$ lacking full rank.
\end{claim}
\begin{proof}
First consider the case that $\rank M = 1$. Then, $M = u u^*$ with $u$ singular. Since $u$ is singular, there exists $u_1, u_2 \in \mathbb C^2$ such that $u = u_1 \otimes u_2$. Thus, if we let $M_1 = u_1 u_1^*$ and $M_2 = u_2 u_2^*$, we have that $M = M_1 \otimes M_2$, with $\rank M_1, \rank M_2 = 1$.

Now assume $\rank M \geq 2$. Pick linearly independent $u, v \in \Span(M)$. For any scalar $\lambda \in \mathbb C$, by assumption we have that
\begin{align*}
    \det \begin{pmatrix}
    u(00) + \lambda v(00) & u(01) + \lambda v(01)\\
    u(10) + \lambda v(10) & u(11) + \lambda v(11)
    \end{pmatrix} &= 0. \implies\\
    \lambda^2 (v(00)v(11) - v(01)v(10)) + \lambda (v(00)u(11)+u(00)v(11) &- u(01)v(10) - u(10)v(01))\\
    + (u(00)u(11) - u(01)u(10)) &= 0.
\end{align*}
In particular, the coefficient of $\lambda$ must be zero:
\begin{align}
v(00)u(11)+u(00)v(11) - u(01)v(10) - u(10)v(01) = 0\label{eq:uv}
\end{align}
Since $u$ and $v$ are singular, we can write $u = u_1 \otimes u_2$ and $v = v_1 \otimes v_2$ with $u_1, u_2, v_1, v_2 \in \C^2$. Then \cref{eq:uv} factors as
\[
(v_1(0) u_1(1) - v_1(1) u_1(0))(v_2(0) u_2(1) - v_2(1) u_2(0)) = 0.
\]
Thus, either $u_1$ and $v_1$ are related by a scalar multiple or $u_2$ and $v_2$ are related by a scalar multiple. If $\rank M = 2$, then we have that either $\Span(M) = \langle u_1\rangle \otimes \langle u_2, v_2\rangle$ or $\Span(M) = \langle u_1, v_1\rangle \otimes \langle u_2\rangle.$ In either case, we have that $M = M_1 \otimes M_2$ with one of $M_1$ or $M_2$ lacking full rank.

Now assume $\rank M \geq 3$. Consider linearly independent $u, v, w \in \Span(M)$ and assume they are all singular, so $u = u_1 \otimes u_2$, $v = v_1 \otimes v_2$, $w = w_1 \otimes w_2$ with $u_1, u_2, v_1, v_2, w_1, w_2 \in \C^2$. By the previous discussion, we have each of the following three logical facts:
\begin{align*}
\langle u_1\rangle = \langle v_1\rangle &\vee \langle u_2\rangle = \langle v_2\rangle\\
\langle u_1\rangle = \langle w_1\rangle &\vee \langle u_2\rangle = \langle w_2\rangle\\
\langle v_1\rangle = \langle w_1\rangle &\vee \langle v_2\rangle = \langle w_2\rangle.
\end{align*}
By pigeonhole, either two equalities are true on the left, or two equalities are true on the right, so by transitivity either $\langle u_1\rangle = \langle v_1\rangle = \langle w_1\rangle$ or $\langle u_2\rangle = \langle v_2\rangle = \langle w_2\rangle$. Since $u_1, u_2, v_1, v_2, w_1, w_2 \in \C^2$, in either case we cannot have $u, v, w$ be linearly independent. Thus, we have completed our classification.
\end{proof}

Thus, the cases we have not classified yet are when the predicate matrix $M$ has the form $M_1 \otimes M_2$ where $(\rank M_1, \rank M_2) \in \{(1,2), (2,1)\}.$ Since swapping the order of the qubits of $M$ leaves $\NRD(M, n)$ unchanged, we may WLOG assume that $(\rank M_1, \rank M_2) = (2,1)$. We now prove an analogue of \cref{lem:nonsingular-aut}.

\begin{lemma}\label{lem:nonsingular-aut-2}
Let $M \in \C^{4 \times 4}$ be a predicate matrix such that $M = M_1 \otimes M_2$ with $M_1, M_2 \in \C^{2\times 2}$ and $(\rank M_1, \rank M_2) = (2,1)$. Let $\cH$ be an $M$-Hamiltonian and $i,j,k \in [n]$ be distinct indices. If $(i,j), (i,k) \in \cH$, then the permutation $(jk)$ is an element of $\Aut(\cH)$.
\end{lemma}

\begin{proof}
We may assume that $\Span(M) = \mathbb C^2 \otimes \langle v\rangle$ for some nonzero $v \in \C^2$. Without loss of generality, we may assume that $T_{12} := (1,2), T_{13} := (1,3) \in \cH$, and we seek to prove that $(23) \in \Aut(\cH)$. Fix an arbitrary state $\psi \in \Psi_{\cH}$. To show that $(23) \in \Aut(\psi)$, it suffices to show $(23) \in \Aut(\psi_x)$ for all $x \in \{0,1\}^T$ where $T = \{4, \hdots, n\}$. This is equivalent to showing that $\psi_x(001) = \psi_x(010)$ and $\psi_x(101) = \psi_x(110)$.

Since $\psi \in \Psi_{\cH}$, we have that the following equations are true:
\begin{align*}
(M{\vert}_{T_{12}} \otimes \Id{\vert}_3) \psi_x &= 0\\
(M{\vert}_{T_{13}} \otimes \Id{\vert}_2) \psi_x &= 0.
\end{align*}
Thus, for all bits $y,z \in \{0,1\}$ we have that
\begin{align*}
v(0)\psi_x(y0z) + v(1)\psi_x(y1z) &= 0\\
v(0)\psi_x(yz0) + v(1)\psi_x(yz1) &= 0.
\end{align*}
If we set $z = 0$ and subtract the above pair of equations, we get that
\[
    v(1)(\psi_x(y10)-\psi_x(y01)) = 0.
\]
Otherwise, if we set $z=1$ and subtract the pair of equations,we get that
\[
    v(0)(\psi_x(y01) - \psi_x(y10)) = 0.
\]
Since $v$ is nonzero, at least one of $v(0) = v(1)$, so we can deduce that $\psi_x(y01) = \psi_x(y10)$. That is, $(23) \in \Aut(\psi_x)$, as desired.
\end{proof}

By suitably adapting the proofs of \cref{lem:bipartite-cycle-redundant} and \cref{lem:nonsingular-NRD}, we may deduce the following

\begin{lemma}
Let $M \in \C^{4 \times 4}$ be a predicate matrix such that $M = M_1 \otimes M_2$ with $M_1, M_2 \in \C^{2\times 2}$ and $(\rank M_1, \rank M_2) \in \{(1,2), (2,1)\}$. Then, $\NRD(M, n) \leq O(n)$.
\end{lemma}

This completes the proof of \cref{thm:2-qubit}.

\section{A Theory of Non-Redundancy for Generic Hamiltonians}\label{sec:nrd-generic}

In this section, we provide more general upper bounds on the non-redundancy of random Hamiltonians. In particular, in \cref{thm:EfficientRandomSparsifierLouie}, we saw that when the predicate matrix $M \in \C^{2^{r} \times 2^r}$ is a random PSD matrix of rank $2^{r-1}+1$, then Hamiltonians using $M$ admit sparsifiers of size $\widetilde{O}(\eps^{-2}n^2)$. As we shall see, we can even show that when the rank falls \emph{below} $2^{r-1}$, that the non-redundancy of these Hamiltonians is still bounded.

We show the following theorem:

\begin{theorem}\label{thm:GenericNRDUpperBound}
    Let $M \in \C^{2^r \times 2^r}$ be a random predicate matrix of rank $R = 2^{r-1}-1$. Then, with probability $1$ over $M$, $\NRD(M, n) \leq \widetilde{O}(n)$.
\end{theorem}

To prove \cref{thm:GenericNRDUpperBound}, we will suppose that $\mathcal{H} = (T_i, M)_{i = 1, \dots m}$ is a non-redundant Hamiltonian. By \cref{prop:r-partite-WLOG}, we can also assume that $\mathcal{H}$ is $r$-partite. 

\subsection{Analyzing Intersecting Terms}

Towards proving \cref{thm:GenericNRDUpperBound}, we first prove the following key lemma:

\begin{lemma}\label{lem:deriveAutomorphismGeneric}
    Let $M \in \C^{2^r \times 2^r}$ be a random predicate matrix of rank $R = 2^{r-1} - 1$, and let $T, T' \subseteq [n]$, ${\vert}T{\vert} = r$ denote two local constraints. Suppose that $T \cap T' \neq \emptyset$, let $a \in T \cap T'$, and let $S = T - \{a \} \subseteq [n]$ and $S' = T' - \{a\} \subseteq [n]$.
    Then, with probability $1$ over $M$, for any state $\psi \in \C^{2^n}$ such that $\cE_{T, M}(\psi) = \cE_{T', M}(\psi) = 0$, it must be the case that the automorphism $(S, S') \in \mathrm{Aut}(\psi)$.
\end{lemma}

When a predicate matrix $M$ satisfies the above requirement, we call it \emph{generic}.

\begin{proof}[Proof of \cref{lem:deriveAutomorphismGeneric}]
Without loss of generality, we will assume that $T = \{1, 2, 3, 4, \dots r\}$ and $T' = \{1, r+1, r+2, \dots 2r-1 \}$. As in the proof of \cref{thm:EfficientRandomSparsifierLouie}, we let $M$ be the span of (random) vectors $v_1, \dots v_R$, and we let $U = T \cup T'$. In this setting, our goal is to show that $\psi$ satisfies the automorphism $(\{2, 3, 4, \dots r\}, \{r+1, \dots 2r-1\}\})$.

Now, for simplicity, we restrict our attention to the sub-state $\phi = \psi_y \in \C^{2^{U}}$ for any $y \in \zo^{[n] - U}$. We construct two matrices $M_0, M_1 \in \C^{2^{r-1} \times 2^{r-1}}$, where the entries of $M_0, M_1$ are indexed by bit strings of length $r-1$, where $M_0(a, b) = \phi_{0 \circ a \circ b}$ and $M_1({a, b}) = \phi_{1 \circ a \circ b}$. To show that $\phi$ satisfies our desired automorphism, we equivalently must only show that $M_0 = M_0^{\top}$ and $M_1 = M_1^{\top}$, as this immediately implies that \[
\phi_{0 \circ a \circ b} = \phi_{0 \circ b \circ a}, \quad \quad \quad \phi_{1 \circ a \circ b} = \phi_{1 \circ b \circ a},
\]
and thus $(\{2, 3, 4, \dots r\}, \{r+1, \dots 2r-1\}\}) \in \Aut(\phi)$ (and correspondingly, of $\psi$). Thus, for the remainder of this proof, we focus on showing that $M_0$ and $M_1$ are symmetric. 

Now, because $\cE_{T, M}(\psi) = \cE_{T', M}(\psi) = 0$, it must be the case that for every $\widehat{y} \in \zo^{U - T}$ and for every $j \in [R]$, $\langle \phi_{\widehat{y}}, v_j \rangle = 0$, and likewise for every $\widehat{y} \in \zo^{U - T'}$, $\langle \phi_{\widehat{y}}, v_j \rangle = 0$. We can now re-write these constraints: for $j \in [R]$, we define $\alpha_j$ to be a length $2^{r-1}$ vector (indexed by strings $a \in \zo^{r-1}$) such that $\alpha_j(a) = (v_j)_{0 \circ a}$. Likewise, we define $\beta_j$ such that $\beta_j(a) = (v_j)_{1 \circ a}$. From these, we construct matrices $A, B \in \C^{R \times 2^{r-1}}$, where $A_{j, a} = \alpha_j(a)$.

Importantly, the constraints that $\langle \phi_{\widehat{y}}, v_j \rangle = 0$ for $\widehat{y} \in \zo^{U - T}$ are in exact correspondence with the property that 
\begin{align}\label{eq:equalZeroFirst}
AM_0 + BM_1 = 0.
\end{align}
Similarly, the fact that $\langle \phi_{\widehat{y}}, v_j \rangle = 0$ for $\widehat{y} \in \zo^{U - T'}$ implies that 
\begin{align}\label{eq:equalZeroSecond}
A(M_0)^{\top} + B (M_1)^{\top} = 0.
\end{align}

Now, for $M_0, M_1$, we decompose these matrices into their symmetric and anti-symmetric parts. I.e., we write
\[
M_0 = Q_0 + K_0, \quad \quad M_1 = Q_1 + K_1,
\]
where $Q_0 = Q_0^{\top}, Q_1 = Q_1^{\top}$, while $K_0 = - K_0^{\top}$ and $K_1 = K_1^{\top}$. Plugging in this decomposition into \cref{eq:equalZeroFirst} and \cref{eq:equalZeroSecond} (and subtracting one from the other), we derive that 
\[
AQ_0 + BQ_1 = 0
\]
and that 
\[
AK_0 + BK_1 = 0.
\]

Because our goal was to show that, in fact, $M_0 = M_0^{\top}$ and $M_1 = M_1^{\top}$, it remains only for us to show that $K_0 = K_1 = 0$, as this then directly implies that our matrices $M_0, M_1$ are symmetric, and we derive our desired automorphism.

Indeed, let us define the function $F_{A, B}(K_0, K_1) = AK_0 + AK_1$, where $K_0, K_1$ are anti-symmetric matrices in $\C^{2^{r-1} \times 2^{r-1}}$, and $A, B \in \C^{R \times 2^{r-1}}$. Because the space of anti-symmetric matrices of size $2^{r-1} \times 2^{r-1}$ is of dimension $2^{r-2} \cdot (2^{r-1}-1)$, we see that $F_{A, B}$ is a linear map with domain of dimension $2 \cdot 2^{r-2} \cdot (2^{r-1}-1) = 2^{r-1} \cdot (2^{r-1}-1)$. Likewise, the co-domain of $F$ is simply $\C^{R \times 2^{r-1}}$, which is of dimension $R \cdot 2^{r-1} = 2^{r-1} \cdot (2^{r-1}-1)$.

Thus, $F_{A, B}$ is a $2^{r-1} \cdot (2^{r-1}-1) \times 2^{r-1} \cdot (2^{r-1}-1)$ dimensional linear map. We let its determinant be $p(A, B) = \det(F_{A, B})$. Observe that this is a polynomial in $A, B$, and because the entries of the matrices $A, B$ are exactly the entries of the vectors $v_1, \dots v_R$, we see $p(A, B)$ is simply a polynomial in the entries of the vectors $v_1, \dots v_R$.

Now, we have the following claim:

\begin{claim}\label{clm:pNonZero}
    The polynomial $p(A, B)$ as defined above is not the zero-polynomial. 
\end{claim}
\begin{proof}[Proof of \cref{clm:pNonZero}.]
    It suffices to find a single choice of $A, B$ such that $F_{A, B}$ is invertible. Indeed, consider 
    \[
    A = \begin{bmatrix}
        \Id_R & 0
    \end{bmatrix}, \quad B = \begin{bmatrix}
        0 & \Id_R
    \end{bmatrix}.
    \]

Now, consider any anti-symmetric matrices 
\[
K_0 = \begin{bmatrix}
0 & x_{1, 2} & x_{1, 3} & \dots & x_{1, 2^{r-1}} \\
-x_{1, 2} & 0 & x_{2, 3} & \dots & x_{2, 2^{r-1}} \\
-x_{1, 3} & -x_{2, 3} & 0 & \dots & x_{3, 2^{r-1}} \\
\vdots & \vdots & \vdots & \ddots & \vdots \\
-x_{1, 2^{r-1}} & -x_{2, 2^{r-1}} & -x_{3, 2^{r-1}} & \dots & 0 \\
\end{bmatrix}\]
\[
K_1 = \begin{bmatrix}
0 & y_{1, 2} & y_{1, 3} & \dots & y_{1, 2^{r-1}} \\
-y_{1, 2} & 0 & y_{2, 3} & \dots & y_{2, 2^{r-1}} \\
-y_{1, 3} & -y_{2, 3} & 0 & \dots & y_{3, 2^{r-1}} \\
\vdots & \vdots & \vdots & \ddots & \vdots \\
-y_{1, 2^{r-1}} & -y_{2, 2^{r-1}} & -y_{3, 2^{r-1}} & \dots & 0 \\
\end{bmatrix}.
\]
Note that $AK_0$ is just $K_0$ with the last row deleted and $BK_1$ is just $K_1$ with the first row deleted. Therefore, we can then see that $F_{A, B}(K_0, K_1) =$
\begingroup
\setlength\arraycolsep{2.5pt}
\[
 \begin{bmatrix}
    -y_{1, 2} & x_{1, 2} & x_{1,3} + y_{2,3} & x_{1, 4} + y_{2, 4} & \dots & x_{1, 2^{r-1}} + y_{2, 2^{r-1}} \\
    -x_{1, 2} - y_{1,3} & -y_{2, 3} & x_{2,3} & x_{2, 4} + y_{3, 4} & \dots & x_{2, 2^{r-1}} + y_{3, 2^{r-1}} \\
    \vdots & \vdots & \vdots & \vdots & \ddots & \vdots \\
    -x_{1, 2^{r-1} -2} - y_{1, 2^{r-1}-1}& -x_{2, 2^{r-1} -2} - y_{2, 2^{r-1}-1} & \dots & \dots & \dots & x_{2^{r-1}-2, 2^{r-1}} + y_{2^{r-1} - 1, 2^{r-1}} \\
    -x_{1, 2^{r-1} -1} - y_{1, 2^{r-1}}& -x_{2, 2^{r-1} -1} - y_{2, 2^{r-1}} & \dots & \dots & \dots & x_{2^{r-1} - 1, 2^{r-1}}
\end{bmatrix}.
\]
\endgroup

To see that this map is invertible, we suppose that we are given a matrix of the form $F_{A, B}(K_0, K_1)$, as prescribed above. It is trivial to see then that one can read off the values of $x_{i, i+1}$ and $y_{i, i+1}$ for $i \in [R]$, as these are ``isolated'' along the diagonal of $F_{A, B}(K_0, K_1)$. Once these values of $x_{i, i+1}, y_{i, i+1}$ are known, we can then calculate $F_{A, B}((K_0){\vert}_{x_{i,i+1} = 0}, (K_1){\vert}_{y_{i, i+1} = 0})$, where $(K_0){\vert}_{x_{i,i+1} = 0}$ refers to the matrix where the variables $x_{i,i+1}: i \in [R]$ are set to $0$.

Again by calculation, we can see that 
\begin{align*}
&F_{A, B}((K_0){\vert}_{x_{i,i+1} = 0}, (K_1){\vert}_{y_{i, i+1} = 0})\\ &= \begin{bmatrix}
    0 & 0 & x_{1,3}  & x_{1, 4} + y_{2, 4} & \dots & x_{1, 2^{r-1}} + y_{2, 2^{r-1}} \\
    y_{1,3} & 0 & 0 & x_{2, 4}  & \dots & x_{2, 2^{r-1}} + y_{3, 2^{r-1}} \\
    \vdots & \vdots & \vdots & \vdots & \ddots & \vdots \\
    -x_{1, 2^{r-1} -2} - y_{1, 2^{r-1}-1}& -x_{2, 2^{r-1} -2} - y_{2, 2^{r-1}-1} & \dots & \dots & \dots & x_{2^{r-1}-2, 2^{r-1}} \\
    -x_{1, 2^{r-1} -1} - y_{1, 2^{r-1}}& -x_{2, 2^{r-1} -1} - y_{2, 2^{r-1}} & \dots & \dots & \dots &0
\end{bmatrix}.
\end{align*}
From here, we can then compute the values of $x_{i, i+2}, y_{i, i+2}: i \in [R-1]$, as these are now ``isolated'' entries along the off-diagonal. 

We then repeat this argument iteratively, where in the $j$th round, where we use $K_0^{(j)} = (K_0){\vert}_{x_{i,i+1} = 0, \dots x_{i, i+ j -1} = 0}$, and $K_1^{(j)} = (K_1){\vert}_{y_{i,i+1} = 0, \dots y_{i, i+ j -1} = 0}$. For these matrices, we will see that in $F_{A, B}(K_0^{(j)}, K_1^{(j)})$, the entries $x_{i,i+j}$ and $y_{i, i+j}$ will be isolated on the off-diagonal, and can therefore be explicitly computed.

Thus, by induction, after $R$ rounds of this procedure, we will have computed $x_{i, p}$ for $i, p \in [2^{r-1}]$, given only $F_{A, B}(K_0, K_1)$. Thus, given $F_{A, B}(K_0, K_1)$, we are able to compute $K_0, K_1$ exactly, and thus $F_{A, B}(K_0, K_1)$ is invertible. 

Thus, for this choice of $A, B$ it must be the case that $p(A, B) = \det(F_{A, B}) \neq 0$ (as a polynomial), as there is a choice of $A, B$ which makes $F_{A, B}$ invertible.
\end{proof}

To conclude, \cref{clm:pNonZero} establishes that the polynomial $p(A, B) \neq 0$. But, at the same time, we know that the variables in $p(A, B)$ are exactly the entries of the vectors $v_1, \dots v_R$ which have continuous joint density. Thus, by \cref{clm:zeroEvaluation}, we know that with probability $1$ over $v_1, \dots v_R$, $p(A, B) \neq 0$. This in turn implies that $K_0 = K_1 = 0$, and so $M_0 = M_0^\top $ and $M_1 = M_1^\top$, thereby deriving our desired automorphism.
\end{proof}

\subsection{Bounding the Non-Redundancy}

Now, we use \cref{lem:deriveAutomorphismGeneric} to prove \cref{thm:GenericNRDUpperBound}.

\begin{proof}[Proof of \cref{thm:GenericNRDUpperBound}.]
    Indeed, suppose that $M$ is generic (as in \cref{lem:deriveAutomorphismGeneric}). Let $\mathcal{H} = (T_i, M)_{i = 1, \dots m}$ denote a non-redundant Hamiltonian with predicate matrix $M$. Note that, as per \cref{prop:r-partite-WLOG}, we can assume that $\mathcal{H}$ is $r$-partite, while only paying a $\widetilde{O}(1)$ factor in the resulting bound on $\NRD$.

    We now make the following claim:
    \begin{claim}\label{clm:RedundantAutomorphismNullSpace}
        Let $\mathcal{H} = (T_i, M)_{i = 1, \dots m}$ denote a Hamiltonian with generic predicate matrix $M$. Fix a term $T_j$ from $\mathcal{H}$, and let $\psi$ be a state such that $\cE_{(T_i, M)_{i = 1, \dots m}, i \neq j}(\psi) = 0$. Suppose further that there is some term $T_{\ell}: \ell \in [m] - i$ such that:
        \begin{enumerate}[(1)]
            \item $T_{\ell} \cap T_j \neq \emptyset$.
            \item For $S_j = T_j - T_j \cap T_{\ell}$ and $S_{\ell} = T_{\ell} - T_j \cap T_{\ell}$, it is the case that $(S_j, S_{\ell}) \in \Aut(\psi)$.
        \end{enumerate}
        Then, it must be the case that $\cE_{(T_j, M)}(\psi) = 0$.
    \end{claim}

    \begin{proof}
        Indeed, because $\cE_{(T_i, M)_{i = 1, \dots m}, i \neq j}(\psi) = 0$, it must be the case that for the term $T_{\ell}$, that for every $y \in \zo^{[n] - T_{\ell}}$ that 
        \[
        \langle \psi_y, M \psi_y \rangle = 0.
        \]

        Now, let us construct the bit string $y' \in \zo^{[n] - T_j}$ such that:
        \begin{enumerate}[(1)]
            \item For indices $g \in [n] - T_j$, $y'_g = y_g$.
            \item For indices $g \in T_j \cap T_{\ell}$, $y'_g = y_g$.
            \item If we let $ S_j = \{ a_1, a_2, \dots \}$ and $S_{\ell} = \{ a'_1, a'_2, \dots \}$, then $y'_{a'_1} = y_{a_1}$.
        \end{enumerate}

        The key observation is exactly that, because of the automorphism $(S_j, S_{\ell}) \in \Aut(\psi)$, it must be the case that $\psi_y = \psi_{y'}$, and thus, 
        \[
        \langle \psi_{y'}, M \psi_{y'} \rangle = 0.
        \]
        Because the strings $y$ and $y'$ as constructed above are in bijection with one another, this means that for every $y' \in \zo^{[n] - T_j}$, $\langle \psi_{y'}, M \psi_{y'} \rangle = 0$, which in turn means that $\cE_{(T_j, M)}(\psi) = 0$, as we desire.
    \end{proof}

    Next, we let $F \subseteq [m]$ denote a subset of $[m]$ of size $\leq n$ which is a spanning forest of $\mathcal{H}$ (see \cref{def:spanningForest} and \cref{prop:buildSpanningForest}).

    Now, we consider adding each other term $T_i: i \in [m] - F$ one at a time. We let $\mathcal{H}^{(p)}$ denote the result of adding the first $p$ terms from $T_i: i \in [m] - F$ (in addition to $F$). Likewise, we let $\Aut^{(p)}(\psi)$ denote the automorphisms that any state $\psi$ must satisfy in order for $\cE_{\mathcal{H}^{(p)}}(\psi) = 0$. Now, when we add the $p+1$st term (call this term $T$) to get $\mathcal{H}^{(p+1)}$, we know that there must be some term $T_j: j \in F$ such that $T_i \cap T_j \neq \emptyset$. We let $S = T - T \cap T_j$ and we let $S_j = T_j - T \cap T_j$. We have have two cases:
    \begin{enumerate}[(1)]
        \item If $(S, S_j) \in \Aut^{(p)}(\psi)$, then we immediately reach a contradiction. Indeed, by \cref{clm:RedundantAutomorphismNullSpace}, this means that for any state $\psi$ such that $\cE_{\mathcal{H}^{(p)}}(\psi) = 0$, it must also be the case that $\cE_{(T, M)}(\psi) = 0$. However, this violates the definition of non-redundancy; indeed, in order for $\mathcal{H}$ to be a non-redundant hypergraph, there must be some choice of $\psi \in \C^{2^n}$ such that $\cE_{\mathcal{H} - (T, M)}(\psi) = 0$ (and thus $\cE_{\mathcal{H}^{(p)}}(\psi) = 0$), while at the same time $\cE_{(T, M)}(\psi) > 0$.
        \item So, it must be the case that $(S, S_j) \notin \Aut^{(p)}(\psi)$. But, by \cref{lem:deriveAutomorphismGeneric}, we know that any state $\psi$ which receives $0$ energy on $T, T_j$ must introduce an automorphism $(S, S_j)$ on $\psi$. This means that $\Aut^{(p+1)}(\psi) = \langle \Aut^{(p)}(\psi) \cup (S, S_j)\rangle$ is a strictly larger subgroup than $\Aut^{(p)}(\psi)$. Thus, it must be the case that
        \[
        2 \leq {\vert} \Aut^{(p+1)}(\psi) : \Aut^{(p)}(\psi) {\vert} = \frac{{\vert}\Aut^{(p+1)}(\psi){\vert}}{ {\vert}\Aut^{(p)}(\psi){\vert}}.
        \]
    \end{enumerate}

    Thus, after $p$ terms are added (beyond $F$), it must be the case that ${\vert}\Aut^{(p)}(\psi){\vert} \geq 2^p$. Once $p = n \log(n)$, this means that ${\vert}\Aut^{(p)}(\psi){\vert} \geq n! = {\vert}S_n{\vert}$, and so the automorphism group is the entire symmetric group on $n$ elements. At this point, ${\vert}\Aut^{(p)}(\psi){\vert}$ cannot grow any larger (i.e., case 2 above cannot occur), and so for any additional term it must be that case $1$ occurs, which in turn violates the non-redundancy of $\mathcal{H}$. Thus, it must be that the number of terms in any non-redundant hypergraph with generic $M$ is $m\leq {\vert}F{\vert} + n \log(n)\leq O(n \log(n))$, as we desire.
\end{proof}

\begin{remark}
We remark that one $\log n$ term in the bound of \cref{thm:GenericNRDUpperBound} can be removed by more carefully analyzing the structure of the symmetric group. More precisely, let $t_1, \hdots, t_\ell$ be the automorphisms identified in the non-redundant instance considered in the proof of \cref{thm:GenericNRDUpperBound}. These automorphisms have the property that $t_i$ does not reside in the subgroup  of $S_n$ generated by $\{t_1, \hdots, t_\ell\} \setminus \{t_i\}$. By a result of Whiston~\cite{Whiston00} any such list of group elements must satisfy $\ell \leq n-1$ (rather than $O(n \log n)$). We omit further details. 
\end{remark}

\section{Conclusion}\label{sec:conclusion}

In this paper, we systematically studied the sparsifiability of Hamiltonians composed of PSD local terms. In comparison to Hamiltonians with negative eigenvalues~\cite{AZ19}, a very rich family of Hamiltonians admit near-linear sparsifiers. A crucial tool we used for constructing these sparsifiers was a quantum analogue of the concept of non-redundancy which captures the interactions between the ground states of various terms in the Hamiltonian. In addition to constructing a few families of sparsifiers, we also systematically studied the non-redundancy of a few types of Hamiltonians, showing the potential for the further construction of Hamiltonian sparsifiers.

Overall, we find that the study of hamiltonian sparsification has been largely unexplored, despite the promise it holds. We hope that others will continue to study questions in this direction, and so we conclude the paper with a few crucial open questions:

\begin{question}\label{question:SPR-NRD}
For an arbitrary predicate matrix $M \in \C^{2^r \times 2^r}$, what is the relationship between $\SPR(M, n, \eps)$ and $\NRD(M, n)$? In particular, when is $\SPR(M, n, \eps) = \widetilde{O}_M(\NRD(M, n))$?
\end{question}

By a result of Brakensiek and Guruswami~\cite{BrakensiekG25} in combination with \cref{prop:classical-equivalence}, we know that the answer to \cref{question:SPR-NRD} is yes whenever $M$ is a classical predicate matrix. A natural quantum extension of \cite{BrakensiekG25} would be to first study \cref{question:SPR-NRD} for \emph{projection matrices} $M$, i.e. matrices $M$ all of whose eigenvalues are $0$ or $1$.

We also highlight a particularly well-studied special case of \cref{question:SPR-NRD}: Quantum Cut.

\begin{question}
What is the sparsifiability of Quantum Cut: $M_{\mathrm{QC}} = ({\vert}01\rangle - {\vert}10\rangle)(\langle 01{\vert} - \langle 10{\vert})$?
\end{question}

Although $M_{\mathrm{QC}}$ has rank $1$, $M_{\mathrm{QC}}$ cannot be decomposed as the product of two rank-$1$ tensors. Thus, by \cref{thm:2-qubit}, we know that the non-redundancy of $M$ is linear, suggesting the the sparsifiability of $M$ could also be near-linear.

\section*{Acknowledgments}

A.B. thanks Alex Lombardi and Siddhant Midha for listening to our results carefully and suggesting Pauli sparsification as a use-case. A.B. also thanks Siddhant Midha for providing comments on a draft of the manuscript.  J.B. thanks Venkatesan Guruswami and Uri Zwick for valuable feedback on an earlier version of the manuscript. A.P. thanks Madhu Sudan, Sanjeev Khanna, and Harald Putterman for helpful comments and discussion.

The authors thank Pravesh Kothari for his involvement in the early stages of this project and for helpful conversations throughout.

The authors used ChatGPT 5.2-pro to assist with the development of concrete examples of vectors which yield non-zero determinants in \cref{sec:sparse-genericLouie} and \cref{sec:nrd-generic}. The authors verified the correctness and originality of all content including references.

\bibliographystyle{alphaurl}
\bibliography{quantum}

@inproceedings{AnshuBN23,
  author       = {Anurag Anshu and
                  Nikolas P. Breuckmann and
                  Chinmay Nirkhe},
  editor       = {Barna Saha and
                  Rocco A. Servedio},
  title        = {{NLTS} Hamiltonians from Good Quantum Codes},
  booktitle    = {Proceedings of the 55th Annual {ACM} Symposium on Theory of Computing,
                  {STOC} 2023, Orlando, FL, USA, June 20-23, 2023},
  pages        = {1090--1096},
  publisher    = {{ACM}},
  year         = {2023},
  url          = {https://doi.org/10.1145/3564246.3585114},
  doi          = {10.1145/3564246.3585114},
  bibsource    = {dblp computer science bibliography, https://dblp.org}
}

@inproceedings{AZ19,
  author       = {Dorit Aharonov and
                  Leo Zhou},
  editor       = {Avrim Blum},
  title        = {Hamiltonian Sparsification and Gap-Simulation},
  booktitle    = {10th Innovations in Theoretical Computer Science Conference, {ITCS}
                  2019, San Diego, California, USA, January 10-12, 2019},
  series       = {LIPIcs},
  pages        = {2:1--2:21},
  publisher    = {Schloss Dagstuhl - Leibniz-Zentrum f{\"{u}}r Informatik},
  year         = {2019},
  url          = {https://doi.org/10.4230/LIPIcs.ITCS.2019.2},
  doi          = {10.4230/LIPIcs.ITCS.2019.2},
  bibsource    = {dblp computer science bibliography, https://dblp.org}
}

@article{Gharibian2019complexityof,
  doi = {10.22331/q-2019-09-30-189},
  url = {https://doi.org/10.22331/q-2019-09-30-189},
  title = {The complexity of simulating local measurements on quantum systems},
  author = {Gharibian, Sevag and Yirka, Justin},
  journal = {{Quantum}},
  issn = {2521-327X},
  publisher = {{Verein zur F{\"{o}}rderung des Open Access Publizierens in den Quantenwissenschaften}},
  volume = {3},
  pages = {189},
  month = sep,
  year = {2019}
}

@article{Kohler2020TranslationallyIU,
  title={Translationally Invariant Universal Quantum Hamiltonians in 1D},
  author={Tamara Kohler and Stephen Piddock and Johannes Bausch and Toby S. Cubitt},
  journal={Annales Henri Poincar{\'e}},
  year={2020},
  volume={23},
  pages={223 - 254},
  url={https://api.semanticscholar.org/CorpusID:214727962}
}

@incollection{parekh2021application,
  author       = {Ojas Parekh and
                  Kevin Thompson},
  editor       = {Nikhil Bansal and
                  Emanuela Merelli and
                  James Worrell},
  title        = {Application of the Level-2 Quantum Lasserre Hierarchy in Quantum Approximation
                  Algorithms},
  booktitle    = {48th International Colloquium on Automata, Languages, and Programming,
                  {ICALP} 2021, Glasgow, Scotland (Virtual Conference), July 12-16,
                  2021},
  series       = {LIPIcs},
  pages        = {102:1--102:20},
  publisher    = {Schloss Dagstuhl - Leibniz-Zentrum f{\"{u}}r Informatik},
  year         = {2021},
  url          = {https://doi.org/10.4230/LIPIcs.ICALP.2021.102},
  doi          = {10.4230/LIPICS.ICALP.2021.102},
  bibsource    = {dblp computer science bibliography, https://dblp.org}
}

@article{anshu2020beyond,
  title={Beyond product state approximations for a quantum analogue of max cut},
  author={Anshu, Anurag and Gosset, David and Morenz, Karen},
  journal={arXiv preprint arXiv:2003.14394},
    url={https://arxiv.org/abs/2003.14394},
  year={2020}
}

@article{arad2010note,
  author       = {Itai Arad},
  title        = {A note about a partial no-go theorem for quantum {PCP}},
  journal      = {Quantum Inf. Comput.},
  volume       = {11},
  number       = {11-12},
  pages        = {1019--1027},
  year         = {2011},
  url          = {https://doi.org/10.26421/QIC11.11-12-10},
  doi          = {10.26421/QIC11.11-12-10},
  bibsource    = {dblp computer science bibliography, https://dblp.org}
}

@article{hastings2012trivial,
  author       = {Matthew B. Hastings},
  title        = {Trivial low energy states for commuting Hamiltonians, and the quantum
                  {PCP} conjecture},  
  editor       = {Nikhil Bansal and
                  Emanuela Merelli and
                  James Worrell},
  booktitle    = {48th International Colloquium on Automata, Languages, and Programming,
                  {ICALP} 2021, Glasgow, Scotland (Virtual Conference), July 12-16,
                  2021},
  series       = {LIPIcs},
  pages        = {102:1--102:20},
  publisher    = {Schloss Dagstuhl - Leibniz-Zentrum f{\"{u}}r Informatik},
  year         = {2021},
  url          = {https://doi.org/10.4230/LIPIcs.ICALP.2021.102},
  doi          = {10.4230/LIPICS.ICALP.2021.102},
  bibsource    = {dblp computer science bibliography, https://dblp.org}
}

@article{aharonov2013guest,
  author       = {Dorit Aharonov and
                  Itai Arad and
                  Thomas Vidick},
  title        = {Guest column: the quantum {PCP} conjecture},
  journal      = {{SIGACT} News},
  volume       = {44},
  number       = {2},
  pages        = {47--79},
  year         = {2013},
  url          = {https://doi.org/10.1145/2491533.2491549},
  doi          = {10.1145/2491533.2491549},
  bibsource    = {dblp computer science bibliography, https://dblp.org}
}

@article{csahinouglu2021hamiltonian,
  title={Hamiltonian simulation in the low-energy subspace},
  author={{\c{S}}ahino{\u{g}}lu, Burak and Somma, Rolando D},
  journal={npj Quantum Information},
  volume={7},
  number={1},
  pages={119},
  year={2021},
  publisher={Nature Publishing Group UK London}
}

@article{zlokapa2024hamiltonian,
  title={Hamiltonian simulation for low-energy states with optimal time dependence},
  author={Zlokapa, Alexander and Somma, Rolando D},
  journal={Quantum},
  volume={8},
  pages={1449},
  year={2024},
  publisher={Verein zur F{\"o}rderung des Open Access Publizierens in den Quantenwissenschaften}
}

@article{zhou2021strongly,
  title={Strongly universal Hamiltonian simulators},
  author={Zhou, Leo and Aharonov, Dorit},
     archivePrefix={arXiv},
  url={https://arxiv.org/abs/2102.02991},
  journal={arXiv preprint arXiv:2102.02991},
  year={2021}
}

@article{li1994reservoir,
  title={Reservoir-sampling algorithms of time complexity o (n (1+ log ({N}/n)))},
  author={Li, Kim-Hung},
  journal={ACM Transactions on Mathematical Software (TOMS)},
  volume={20},
  number={4},
  pages={481--493},
  year={1994},
  publisher={ACM New York, NY, USA},
  doi={https://doi.org/10.1145/198429.198435}
}

@article{jordan2010quantum,
  title={Quantum-Merlin-Arthur--complete problems for stoquastic Hamiltonians and Markov matrices},
  author={Jordan, Stephen P and Gosset, David and Love, Peter J},
  journal={Physical Review A—Atomic, Molecular, and Optical Physics},
  volume={81},
  number={3},
  pages={032331},
  year={2010},
  publisher={APS},
  doi={https://doi.org/10.1103/physreva.81.032331}
}

@inproceedings{KallaugherP22,
  author       = {John Kallaugher and
                  Ojas Parekh},
  title        = {The Quantum and Classical Streaming Complexity of Quantum and Classical
                  Max-Cut},
  booktitle    = {63rd {IEEE} Annual Symposium on Foundations of Computer Science, {FOCS}
                  2022, Denver, CO, USA, October 31 - November 3, 2022},
  pages        = {498--506},
  publisher    = {{IEEE}},
  year         = {2022},
  url          = {https://doi.org/10.1109/FOCS54457.2022.00054},
  doi          = {10.1109/FOCS54457.2022.00054},
  bibsource    = {dblp computer science bibliography, https://dblp.org}
}

@article{Tro15,
url = {https://doi.org/10.1561/2200000048},
year = {2015},
volume = {8},
journal = {Foundations and Trends® in Machine Learning},
title = {An Introduction to Matrix Concentration Inequalities},
doi = {10.1561/2200000048},
issn = {1935-8237},
number = {1-2},
pages = {1-230},
author = {Joel A. Tropp}
}

@inproceedings{BenczurK96,
author = {Bencz\'{u}r, Andr\'{a}s A. and Karger, David R.},
title = {Approximating s-t minimum cuts in \~{O}(n2) time},
year = {1996},
isbn = {0897917855},
publisher = {Association for Computing Machinery},
address = {New York, NY, USA},
url = {https://doi.org/10.1145/237814.237827},
doi = {10.1145/237814.237827},
booktitle = {Proceedings of the Twenty-Eighth Annual ACM Symposium on Theory of Computing},
pages = {47–55},
numpages = {9},
location = {Philadelphia, Pennsylvania, USA},
series = {STOC '96}
}

@article{SpielmanT11,
author = {Spielman, Daniel A. and Teng, Shang-Hua},
title = {Spectral Sparsification of Graphs},
journal = {SIAM Journal on Computing},
volume = {40},
number = {4},
pages = {981-1025},
year = {2011},
doi = {10.1137/08074489X},
URL = {https://doi.org/10.1137/08074489X},
}

@article{SpielmanS11,
author = {Spielman, Daniel A. and Srivastava, Nikhil},
title = {Graph Sparsification by Effective Resistances},
journal = {SIAM Journal on Computing},
volume = {40},
number = {6},
pages = {1913-1926},
year = {2011},
doi = {10.1137/080734029},
}

@article{BatsonSS14,
author = {Batson, Joshua and Spielman, Daniel A. and Srivastava, Nikhil},
title = {Twice-Ramanujan Sparsifiers},
journal = {SIAM Review},
volume = {56},
number = {2},
pages = {315-334},
year = {2014},
doi = {10.1137/130949117},
URL = {https://doi.org/10.1137/130949117},
}

@inproceedings{KhannaPS24,
author={Khanna, Sanjeev and Putterman, Aaron and Sudan, Madhu},
  editor       = {David P. Woodruff},
  title        = {Code sparsification and its applications},
  booktitle    = {Proceedings of the 2024 {ACM-SIAM} Symposium on Discrete Algorithms,
                  {SODA} 2024, Alexandria, VA, USA, January 7-10, 2024},
  pages        = {5145--5168},
  publisher    = {{SIAM}},
  year         = {2024},
  url          = {https://doi.org/10.1137/1.9781611977912.185},
  doi          = {10.1137/1.9781611977912.185},
  bibsource    = {dblp computer science bibliography, https://dblp.org}
}

@BOOK{Ash1990,
  title     = "Information Theory",
  author    = "Ash, Robert",
  publisher = "Dover Publications",
  series    = "Dover Books on Mathematics",
  month     =  nov,
  year      =  1990,
  address   = "Mineola, NY",
  language  = "en"
}

@inproceedings{KhannaPS25,
author = {Khanna, Sanjeev and Putterman, Aaron and Sudan, Madhu},
title = {Efficient Algorithms and New Characterizations for CSP Sparsification},
year = {2025},
isbn = {9798400715105},
publisher = {Association for Computing Machinery},
address = {New York, NY, USA},
url = {https://doi.org/10.1145/3717823.3718205},
doi = {10.1145/3717823.3718205},
booktitle = {Proceedings of the 57th Annual ACM Symposium on Theory of Computing},
pages = {407–416},
numpages = {10},
location = {Prague, Czechia},
series = {STOC '25}
}

@inproceedings{BrakensiekG25,
author = {Brakensiek, Joshua and Guruswami, Venkatesan},
title = {Redundancy Is All You Need},
year = {2025},
isbn = {9798400715105},
publisher = {Association for Computing Machinery},
address = {New York, NY, USA},
url = {https://doi.org/10.1145/3717823.3718212},
doi = {10.1145/3717823.3718212},
booktitle = {Proceedings of the 57th Annual ACM Symposium on Theory of Computing},
pages = {1614–1625},
numpages = {12},
location = {Prague, Czechia},
series = {STOC '25}
}

@inbook{BasuKLM26,
author = {Arpon Basu and Pravesh K. Kothari and Yang P. Liu and Raghu Meka},
title = {Sparsifying Sums of Positive Semidefinite Matrices},
booktitle = {Proceedings of the 2026 Annual ACM-SIAM Symposium on Discrete Algorithms (SODA)},
pages = {6042-6064},
doi = {10.1137/1.9781611978971.216},
URL = {https://doi.org/10.1137/1.9781611978971.216},
year = {2026}
}

@article{FiltserK17,
  author       = {Arnold Filtser and
                  Robert Krauthgamer},
  title        = {Sparsification of Two-Variable Valued Constraint Satisfaction Problems},
  journal      = {{SIAM} J. Discret. Math.},
  volume       = {31},
  number       = {2},
  pages        = {1263--1276},
  year         = {2017},
  url          = {https://doi.org/10.1137/15M1046186},
  doi          = {10.1137/15M1046186},
  bibsource    = {dblp computer science bibliography, https://dblp.org}
}

@inproceedings{KoganK15,
  author       = {Dmitry Kogan and
                  Robert Krauthgamer},
  editor       = {Tim Roughgarden},
  title        = {Sketching Cuts in Graphs and Hypergraphs},
  booktitle    = {Proceedings of the 2015 Conference on Innovations in Theoretical Computer
                  Science, {ITCS} 2015, Rehovot, Israel, January 11-13, 2015},
  pages        = {367--376},
  publisher    = {{ACM}},
  year         = {2015},
  url          = {https://doi.org/10.1145/2688073.2688093},
  doi          = {10.1145/2688073.2688093},
  bibsource    = {dblp computer science bibliography, https://dblp.org}
}

@inproceedings{ChenKN20,
  author       = {Yu Chen and
                  Sanjeev Khanna and
                  Ansh Nagda},
  editor       = {Sandy Irani},
  title        = {Near-linear Size Hypergraph Cut Sparsifiers},
  booktitle    = {61st {IEEE} Annual Symposium on Foundations of Computer Science, {FOCS}
                  2020, Durham, NC, USA, November 16-19, 2020},
  pages        = {61--72},
  publisher    = {{IEEE}},
  year         = {2020},
  url          = {https://doi.org/10.1109/FOCS46700.2020.00015},
  doi          = {10.1109/FOCS46700.2020.00015},
  bibsource    = {dblp computer science bibliography, https://dblp.org}
}

@article{BrakensiekGP25,
  author       = {Joshua Brakensiek and
                  Venkatesan Guruswami and
                  Aaron Putterman},
  title        = {Tight Bounds for Sparsifying Random CSPs},
  journal      = {CoRR},
  volume       = {abs/2508.13345},
  year         = {2025},
  url          = {https://doi.org/10.48550/arXiv.2508.13345},
  doi          = {10.48550/arXiv.2508.13345},
  eprinttype   = {arXiv},
  bibsource    = {dblp computer science bibliography, https://dblp.org}
}

@article{ButtiZ20,
  title = {Sparsification of {{Binary CSPs}}},
  author = {Butti, Silvia and {\v Z}ivn{\'y}, Stanislav},
  year = {2020},
  month = jan,
  journal = {SIAM Journal on Discrete Mathematics},
  volume = {34},
  number = {1},
  pages = {825--842},
  publisher = {{Society for Industrial and Applied Mathematics}},
  issn = {0895-4801},
  doi = {10.1137/19M1242446}
}

@inproceedings{Carbonnel22,
  author       = {Cl{\'{e}}ment Carbonnel},
  editor       = {Christine Solnon},
  title        = {On Redundancy in Constraint Satisfaction Problems},
  booktitle    = {28th International Conference on Principles and Practice of Constraint
                  Programming, {CP} 2022, Haifa, Israel, July 31 - August 8, 2022},
  series       = {LIPIcs},
  pages        = {11:1--11:15},
  publisher    = {Schloss Dagstuhl - Leibniz-Zentrum f{\"{u}}r Informatik},
  year         = {2022},
  url          = {https://doi.org/10.4230/LIPIcs.CP.2022.11},
  doi          = {10.4230/LIPIcs.CP.2022.11},
  bibsource    = {dblp computer science bibliography, https://dblp.org}
}

@article{BrakensiekGJLW25,
  author       = {Joshua Brakensiek and
                  Venkatesan Guruswami and
                  Bart M. P. Jansen and
                  Victor Lagerkvist and
                  Magnus Wahlstr{\"{o}}m},
  title        = {The Richness of {CSP} Non-redundancy},
  journal      = {CoRR},
  volume       = {abs/2507.07942},
  year         = {2025},
  url          = {https://doi.org/10.48550/arXiv.2507.07942},
  doi          = {10.48550/arXiv.2507.07942},
  eprinttype   = {arXiv},
  bibsource    = {dblp computer science bibliography, https://dblp.org}
}

@article{ChenJP20,
  title = {Best-{{Case}} and {{Worst-Case Sparsifiability}} of {{Boolean CSPs}}},
  author = {Chen, Hubie and Jansen, Bart M. P. and Pieterse, Astrid},
  year = {2020},
  month = aug,
  journal = {Algorithmica},
  volume = {82},
  number = {8},
  pages = {2200--2242},
  issn = {1432-0541},
  doi = {10.1007/s00453-019-00660-y},
  langid = {english}
}

@article{LagerkvistW20,
  title = {Sparsification of {{SAT}} and {{CSP Problems}} via {{Tractable Extensions}}},
  author = {Lagerkvist, Victor and Wahlstr{\"o}m, Magnus},
  year = {2020},
  month = jun,
  journal = {ACM Transactions on Computation Theory},
  volume = {12},
  number = {2},
  pages = {1--29},
  issn = {1942-3454, 1942-3462},
  doi = {10.1145/3389411},
  langid = {english}
}

@article{BessiereCK20,
  title = {Chain {{Length}} and {{CSPs Learnable}} with {{Few Queries}}},
  author = {Bessiere, Christian and Carbonnel, Cl{\'e}ment and Katsirelos, George},
  year = {2020},
  month = apr,
  journal = {Proceedings of the AAAI Conference on Artificial Intelligence},
  volume = {34},
  number = {02},
  pages = {1420--1427},
  issn = {2374-3468},
  doi = {10.1609/aaai.v34i02.5499},
  copyright = {Copyright (c) 2020 Association for the Advancement of Artificial Intelligence},
  langid = {english}
}

@inproceedings{BessiereCHKLNQW13,
  author       = {Christian Bessiere and
                  Remi Coletta and
                  Emmanuel Hebrard and
                  George Katsirelos and
                  Nadjib Lazaar and
                  Nina Narodytska and
                  Claude{-}Guy Quimper and
                  Toby Walsh},
  editor       = {Francesca Rossi},
  title        = {Constraint Acquisition via Partial Queries},
  booktitle    = {{IJCAI} 2013, Proceedings of the 23rd International Joint Conference
                  on Artificial Intelligence, Beijing, China, August 3-9, 2013},
  pages        = {475--481},
  publisher    = {{IJCAI/AAAI}},
  year         = {2013},
  url          = {http://www.aaai.org/ocs/index.php/IJCAI/IJCAI13/paper/view/6659},
  bibsource    = {dblp computer science bibliography, https://dblp.org}
}

@article{JansenP19,
  author       = {Bart M. P. Jansen and
                  Astrid Pieterse},
  title        = {Optimal Sparsification for Some Binary CSPs Using Low-Degree Polynomials},
  journal      = {{ACM} Trans. Comput. Theory},
  volume       = {11},
  number       = {4},
  pages        = {28:1--28:26},
  year         = {2019},
  url          = {https://doi.org/10.1145/3349618},
  doi          = {10.1145/3349618},
  bibsource    = {dblp computer science bibliography, https://dblp.org}
}

@inproceedings{LagerkvistGE26,
  title={Towards Single Exponential Time for Temporal and Spatial Reasoning: A Study via Redundancy and Dynamic Programming},
  author={Lagerkvist, Victor and Groven, Johanna and Eriksson, Leif},
  booktitle={Proceedings of the AAAI Conference on Artificial Intelligence},
  volume={40},
  pages={14287--14294},
  year={2026},
  doi={https://doi.org/10.1609/aaai.v40i17.38443}
}

@misc{BrakensiekGP26,
      title={Classification of Non-redundancy of Boolean Predicates of Arity 4}, 
      author={Brakensiek, Joshua and Guruswami, Venkatesan and Putterman, Aaron},
      year={2026},
      archivePrefix={arXiv},
      primaryClass={cs.CC},
      url={https://arxiv.org/abs/2603.21353}, 
      journal={arXiv preprint arXiv:2603.21353},
}

@article {Whiston00,
    AUTHOR = {Whiston, Julius},
     TITLE = {Maximal independent generating sets of the symmetric group},
   JOURNAL = {J. Algebra},
  FJOURNAL = {Journal of Algebra},
    VOLUME = {232},
      YEAR = {2000},
    NUMBER = {1},
     PAGES = {255--268},
      ISSN = {0021-8693,1090-266X},
   MRCLASS = {20B30 (20B20)},
  MRNUMBER = {1783924},
MRREVIEWER = {Mohammad-Reza\ Darafsheh},
       DOI = {10.1006/jabr.2000.8399},
       URL = {https://doi.org/10.1006/jabr.2000.8399},
}

@inproceedings{SchiexFV95,
  author       = {Thomas Schiex and
                  H{\'{e}}l{\`{e}}ne Fargier and
                  G{\'{e}}rard Verfaillie},
  title        = {Valued Constraint Satisfaction Problems: Hard and Easy Problems},
  booktitle    = {Proceedings of the Fourteenth International Joint Conference on Artificial
                  Intelligence, {IJCAI} 95, Montr{\'{e}}al Qu{\'{e}}bec, Canada,
                  August 20-25 1995, 2 Volumes},
  pages        = {631--639},
  publisher    = {Morgan Kaufmann},
  year         = {1995},
  url          = {http://ijcai.org/Proceedings/95-1/Papers/083.pdf},
  bibsource    = {dblp computer science bibliography, https://dblp.org}
}

@book{GunningR65,
  title={Analytic Functions of Several Complex Variables},
  author={Gunning, Robert C. and Rossi, Hugo},
  year={1965},
  publisher={Prentice-Hall},
  doi={https://doi.org/10.1090/chel/368}
}

@article{SattathMLM16,
author = {Or Sattath  and Siddhardh C. Morampudi  and Chris R. Laumann  and Roderich Moessner },
title = {When a local Hamiltonian must be frustration-free},
journal = {Proceedings of the National Academy of Sciences},
volume = {113},
number = {23},
pages = {6433-6437},
year = {2016},
doi = {10.1073/pnas.1519833113},
}

@article{MichalakisZ13,
  title={Stability of frustration-free Hamiltonians},
  author={Michalakis, Spyridon and Zwolak, Justyna P},
  journal={Communications in Mathematical Physics},
  volume={322},
  number={2},
  pages={277--302},
  year={2013},
  publisher={Springer},
  doi={https://doi.org/10.1007/s00220-013-1762-6}
}

@article{AharonovALV10,
  title={Quantum Hamiltonian complexity and the detectability lemma},
  author={Aharonov, Dorit and Arad, Itai and Landau, Zeph and Vazirani, Umesh},
   archivePrefix={arXiv},
     journal={arXiv preprint arXiv:1011.3445},
  url={https://arxiv.org/abs/1011.3445},
  year={2010}
}

@article{Bravyi06,
    AUTHOR = {Bravyi, Sergey},
     TITLE = {Efficient algorithm for a quantum analogue of 2-{SAT}},
 BOOKTITLE = {Cross disciplinary advances in quantum computing},
    SERIES = {Contemp. Math.},
    VOLUME = {536},
     PAGES = {33--48},
 PUBLISHER = {Amer. Math. Soc., Providence, RI},
      YEAR = {2011},
      ISBN = {978-0-8218-4975-0},
   MRCLASS = {68Q12 (68Q15 81P68)},
  MRNUMBER = {2768792},
MRREVIEWER = {Salman\ Beigi},
       DOI = {10.1090/conm/536/10552},
       URL = {https://doi.org/10.1090/conm/536/10552},
}

@inproceedings{BeaudrapG16,
  author       = {J. Niel de Beaudrap and
                  Sevag Gharibian},
  editor       = {Ran Raz},
  title        = {A Linear Time Algorithm for Quantum 2-SAT},
  booktitle    = {31st Conference on Computational Complexity, {CCC} 2016, Tokyo, Japan,
                  May 29 - June 1, 2016},
  series       = {LIPIcs},
  pages        = {27:1--27:21},
  publisher    = {Schloss Dagstuhl - Leibniz-Zentrum f{\"{u}}r Informatik},
  year         = {2016},
  url          = {https://doi.org/10.4230/LIPIcs.CCC.2016.27},
  doi          = {10.4230/LIPIcs.CCC.2016.27},
  bibsource    = {dblp computer science bibliography, https://dblp.org}
}

@article{GossetN16,
  author       = {David Gosset and
                  Daniel Nagaj},
  title        = {Quantum 3-SAT Is QMA\({}_{\mbox{1}}\)-Complete},
  journal      = {{SIAM} J. Comput.},
  volume       = {45},
  number       = {3},
  pages        = {1080--1128},
  year         = {2016},
  url          = {https://doi.org/10.1137/140957056},
  doi          = {10.1137/140957056},
  bibsource    = {dblp computer science bibliography, https://dblp.org}
}

@article{PiddockM17,
  author       = {Stephen Piddock and
                  Ashley Montanaro},
  title        = {The complexity of antiferromagnetic interactions and 2D lattices},
  journal      = {Quantum Inf. Comput.},
  volume       = {17},
  number       = {7{\&}8},
  pages        = {636--672},
  year         = {2017},
  url          = {https://doi.org/10.26421/QIC17.7-8-6},
  doi          = {10.26421/QIC17.7-8-6},
  bibsource    = {dblp computer science bibliography, https://dblp.org}
}

@inproceedings{HwangNP0W23,
  author       = {Yeongwoo Hwang and
                  Joe Neeman and
                  Ojas Parekh and
                  Kevin Thompson and
                  John Wright},
  editor       = {Nikhil Bansal and
                  Viswanath Nagarajan},
  title        = {Unique Games hardness of Quantum Max-Cut, and a conjectured vector-valued
                  Borell's inequality},
  booktitle    = {Proceedings of the 2023 {ACM-SIAM} Symposium on Discrete Algorithms,
                  {SODA} 2023, Florence, Italy, January 22-25, 2023},
  pages        = {1319--1384},
  publisher    = {{SIAM}},
  year         = {2023},
  url          = {https://doi.org/10.1137/1.9781611977554.ch48},
  doi          = {10.1137/1.9781611977554.ch48},
  bibsource    = {dblp computer science bibliography, https://dblp.org}
}

@article{Piddock25,
  author       = {Stephen Piddock},
  title        = {Quantum Max-Cut is {NP} hard to approximate},
  journal      = {CoRR},
  volume       = {abs/2510.07995},
  year         = {2025},
  url          = {https://doi.org/10.48550/arXiv.2510.07995},
  doi          = {10.48550/arXiv.2510.07995},
  eprinttype   = {arXiv},
  bibsource    = {dblp computer science bibliography, https://dblp.org}
}

@article{CubittM16,
  author       = {Toby S. Cubitt and
                  Ashley Montanaro},
  title        = {Complexity Classification of Local Hamiltonian Problems},
  journal      = {{SIAM} J. Comput.},
  volume       = {45},
  number       = {2},
  pages        = {268--316},
  year         = {2016},
  url          = {https://doi.org/10.1137/140998287},
  doi          = {10.1137/140998287},
  bibsource    = {dblp computer science bibliography, https://dblp.org}
}

@inproceedings{GharibianP19,
  author       = {Sevag Gharibian and
                  Ojas Parekh},
  editor       = {Dimitris Achlioptas and
                  L{\'{a}}szl{\'{o}} A. V{\'{e}}gh},
  title        = {Almost Optimal Classical Approximation Algorithms for a Quantum Generalization
                  of Max-Cut},
  booktitle    = {Approximation, Randomization, and Combinatorial Optimization. Algorithms
                  and Techniques, {APPROX/RANDOM} 2019, Massachusetts Institute of Technology,
                  Cambridge, MA, USA, September 20-22, 2019},
  series       = {LIPIcs},
  pages        = {31:1--31:17},
  publisher    = {Schloss Dagstuhl - Leibniz-Zentrum f{\"{u}}r Informatik},
  year         = {2019},
  url          = {https://doi.org/10.4230/LIPIcs.APPROX-RANDOM.2019.31},
  doi          = {10.4230/LIPIcs.APPROX-RANDOM.2019.31},
  bibsource    = {dblp computer science bibliography, https://dblp.org}
}

\appendix

\end{document}